\documentclass[a4paper,onecolumn,11pt,accepted=2022-10-03]{quantumarticle}
\pdfoutput=1

\usepackage{amsmath, amsthm, amssymb,color,url,booktabs,enumitem}  %cite
\usepackage{graphicx}
\ifdefined\blackandwhite
\usepackage[pdftex]{hyperref} %pagebackref
\else
\usepackage[colorlinks,pdftex]{hyperref} %pagebackref
\hypersetup{citecolor = blue, linkcolor = red!70!black}
\fi

\usepackage{braket}
\usepackage{mathtools}
\usepackage{mdframed}

\usepackage{tikz}
\usepackage{authblk}

\usepackage[ruled,vlined,linesnumbered]{algorithm2e}

\newcommand{\nc}{\newcommand}
\nc{\rnc}{\renewcommand}

\DeclareMathOperator{\Det}{Det}

\newcommand{\Tr}{\operatorname*{Tr}}

\def\be#1\ee{\begin{equation}#1\end{equation}}
\def\bea#1\eea{\begin{eqnarray}#1\end{eqnarray}}
\def\beas#1\eeas{\begin{eqnarray*}#1\end{eqnarray*}}
\def\ba#1\ea{\begin{align}#1\end{align}}
\def\bas#1\eas{\begin{align*}#1\end{align*}}
\def\bpm#1\epm{\begin{pmatrix}#1\end{pmatrix}}
\nc{\non}{\nonumber}
\nc{\nn}{\nonumber}
\nc{\eq}[1]{(\ref{eq:#1})}
\nc{\eqs}[2]{(\ref{eq:#1}) and (\ref{eq:#2})}
\rnc{\L}{\left} 
\nc{\ra}{\rightarrow}
\nc{\grad}{{\vec{\nabla}}}
\newtheorem{theorem}{Theorem}
\newtheorem*{theorem*}{Theorem}
\newtheorem{coro}[theorem]{Corollary}

\newtheorem{lemma}{Lemma}
\newtheorem*{lemma*}{Lemma}

\newtheorem*{proposition*}{Proposition}

\newtheorem*{corollary*}{Corollary}

\newtheorem*{fact*}{Fact}

\nc\eps{\epsilon}

\nc\cA{\mathcal{A}}
\nc\cB{\mathcal{B}}
\nc\cC{\mathcal{C}}
\nc\cD{\mathcal{D}}
\nc\cE{\mathcal{E}}
\nc\cF{\mathcal{F}}
\nc\cG{\mathcal{G}}
\nc\cH{\mathcal{H}}
\nc\cI{\mathcal{I}}
\nc\cJ{\mathcal{J}}
\nc\cK{\mathcal{K}}
\nc\cL{\mathcal{L}}
\nc\cM{\mathcal{M}}
\nc\cN{\mathcal{N}}
\nc\cO{\mathcal{O}}
\nc\cP{\mathcal{P}}
\nc\cQ{\mathcal{Q}}
\nc\cR{\mathcal{R}}
\nc\cS{\mathcal{S}}
\nc\cT{\mathcal{T}}
\nc\cU{\mathcal{U}}
\nc\cV{\mathcal{V}}
\nc\cW{\mathcal{W}}
\nc\cX{\mathcal{X}}
\nc\cY{\mathcal{Y}}
\nc\cZ{\mathcal{Z}}

\nc\bbF{\mathbb{F}}
\nc\bbM{\mathbb{M}}
\nc\bbN{\mathbb{N}}
\nc\bbR{\mathbb{R}}
\nc\bbZ{\mathbb{Z}}

\nc\benum{\begin{enumerate}}
\nc\eenum{\end{enumerate}}
\nc\bit{\begin{itemize}}
\nc\eit{\end{itemize}}

\newcommand{\corref}[1]{Corollary~\ref{cor:#1}}

\nc{\Anote}[1]{\textcolor{red}{Aram note: #1}}

\def\begsub#1#2\endsub{\begin{subequations}\label{eq:#1}\begin{align}#2\end{align}\end{subequations}}
\nc\qand{\qquad\text{and}\qquad}
\nc\mnb[1]{\medskip\noindent{\bf #1}}

\nc{\pder}[2]{\frac{\partial {#1}}{\partial {#2}}}
\nc{\p}{\partial}
\usepackage{wasysym}
\usepackage{array}
\usepackage{multicol}

\definecolor{orange}{rgb}{1,0.5,0}

\usepackage{authblk}
\usepackage{hyperref,tabu, verbatim}% http://ctan.org/pkg/hyperref
\usepackage{bm}
\usepackage{mdframed}
\hypersetup{%
  colorlinks = true,
  linkcolor  = black
}

\begin{document}

\newgeometry{left=2.9cm,right=2.9cm}

\title{Holomorphic representation of quantum computations}

\author{Ulysse Chabaud}
\orcid{0000-0003-0135-9819}
\email{uchabaud@caltech.edu}
\affiliation{Institute for Quantum Information and Matter, California Institute of Technology, Pasadena, CA 91125, USA}

\author{Saeed Mehraban}
\orcid{0000-0002-1323-360X}
\email{saeed.mehraban@tufts.edu}
\affiliation{Computer Science, Tufts University, Medford, MA 02155, USA}
\affiliation{Institute for Quantum Information and Matter, California Institute of Technology, Pasadena, CA 91125, USA}

\maketitle
\thispagestyle{empty}

%------------------------------------------------------------------------

\begin{abstract}

We study bosonic quantum computations using the Segal--Bargmann representation of quantum states. We argue that this holomorphic representation is a natural one which not only gives a canonical description of bosonic quantum computing using basic elements of complex analysis but also provides a unifying picture which delineates the boundary between discrete- and continuous-variable quantum information theory. Using this representation, we show that the evolution of a single bosonic mode under a Gaussian Hamiltonian can be described as an integrable dynamical system of classical Calogero--Moser particles corresponding to the zeros of the holomorphic function, together with a conformal evolution of Gaussian parameters. We explain that the Calogero--Moser dynamics is due to unique features of bosonic Hilbert spaces such as squeezing. We then generalize the properties of this holomorphic representation to the multimode case, deriving a non-Gaussian hierarchy of quantum states and relating entanglement to factorization properties of holomorphic functions. Finally, we apply this formalism to discrete- and continuous-variable quantum measurements and obtain a classification of subuniversal models that are generalizations of Boson Sampling and Gaussian quantum computing.

\end{abstract}

%------------------------------------------------------------------------

\newgeometry{left=3cm,right=3cm}
\thispagestyle{empty}
\tableofcontents

%------------------------------------------------------------------------

\section{Introduction}
\label{sec:intro}

Information processing necessarily involves the encoding of information in the degrees of freedom of some physical system.
Quantum degrees of freedom process information in fundamentally different ways than classical ones and offer dramatic advantages in computing~\cite{shor1994algorithms,harrow2017quantum}, communication~\cite{wiesner1983conjugate}, cryptography~\cite{bennett1993teleporting} and sensing~\cite{giovannetti2004quantum}. Different quantum degrees of freedom, however, correspond to innately different physical platforms for information processing. Various physical degrees of freedom such as electron spins or the polarization of light are discrete in nature and are mathematically captured by quantum bits (qubits) over finite-dimensional Hilbert spaces. Other quantum degrees of freedom such as position or momentum, on the other hand, are continuous in nature and captured by quantum modes (qumodes) over infinite-dimensional Hilbert spaces. 

The term continuous-variable (CV) is also used indistinctly for quantum systems with infinite but discrete degrees of freedom, such as particle number, which are described in separable Hilbert spaces. These infinite-dimensional separable Hilbert spaces have a countable basis, which is manifested in quantum mechanics by the particle nature of quantum systems such as light. In that sense, the countable basis of separable Hilbert spaces reflects the duality between wave and particle nature of matter. As a matter of fact, quantum mechanics owes the term ``quantum'' to this counter-intuitive discrete particle aspect of physical states that are also sometimes observed continuous.

Quantum information processing with discrete variables (DV) is a well-studied subject and has been defined as the backbone of the standard model of quantum computing~\cite{NielsenChuang}. However, CV quantum information~\cite{Braunstein2005,adesso2014continuous} has attracted increasing interest in the past few years due to its appealing features in terms of bosonic quantum error correction~\cite{Gottesman2001,cai2021bosonic}, efficient measurements~\cite{vahlbruch2016detection}, and scalable entanglement generation~\cite{yokoyama2013ultra}.
The introduction of the Boson Sampling model by Aaronson and Arkhipov~\cite{Aaronson2013} to demonstrate quantum computational speedup using bosons---and photons in particular---also contributed to this traction, as the experimental challenge it represented prompted the introduction of more experimentally-friendly models of CV quantum computation~\cite{Lund2014,Hamilton2016,Douce2017,Lund2017,Chakhmakhchyan2017,chabaud2017continuous}.

The permanent polynomial appears very naturally as the amplitudes of non-interacting bosons. Indeed, Aaronson and Arkhipov used this observation to prove hardness of classical simulation for these systems~\cite{Aaronson2013}. While Boson Sampling states may be described within a finite-dimensional subspace of fixed particle number of the infinite-dimensional Hilbert space, its variant Gaussian Boson Sampling~\cite{Lund2014,Hamilton2016} makes use of a single-mode resource, squeezing, which is fundamentally continuous. This CV model was recently demonstrated experimentally~\cite{zhong2020quantum}, and followed the demonstration of DV quantum computational speedup~\cite{arute2019quantum}. The discussion of the differences between these two experiments and their possible quantum speedups is still ongoing~\cite{horner2021have,pan2021simulating,villalonga2021efficient}.
 
In principle, one can represent CV quantum computations to arbitrary accuracy using DV formalism. While ideas such as truncating infinite-dimensional Hilbert spaces have been used to simulate interesting instances of CV computations using discrete variables~\cite{tong2021provably}, the possibility that there would be quantum computations that are feasible efficiently on a quantum CV machine and not by DV models cannot be ruled out~\cite{lloyd1999quantum}. Other than that, CV have unique features that are distinct and more properly captured by an innately continuous language. Indeed, one of the central messages we wish to convey is that for continuous quantities such as those of CV quantum mechanics we gain perspective by using an innately continuous language. Consider the toy example of approximating the complex function $e^z$ by $1+z+z^2/2+z^3/6$: while the latter is a good quantitative approximation in some regime, it misses a geometric interpretation which the former representation naturally conveys (such as mapping vertical lines to circles in the complex plane). As a matter of fact, by discretizing $e^z$ operations such as taking powers of the function become more complicated to perform. As we will see, the contrast between exponential and polynomial functions such as the one in this example appears very naturally in the study of bosonic quantum systems.

Even though discrete- and continuous-variable quantum models appear to be qualitatively incomparable, they share remarkable similarities. For instance, similar in spirit to the DV \textit{Clifford/non-Clifford} dichotomy, the CV framework features a \textit{Gaussian/non-Gaussian} dichotomy, where Gaussian quantum computations can be simulated efficiently classically~\cite{Bartlett2002}. Hence, non-Gaussian quantum states, operations and measurements may be thought of as resources for quantum computational speedup~\cite{takagi2018convex,zhuang2018resource,albarelli2018resource}. A manifestation of this is in the Boson Sampling model, where non-Gaussian input states and measurements lead to a computational task that is believed to be beyond the reach of classical computers. On the other hand, Gaussian states and operations have more limitations than their DV Clifford counterparts, as non-Gaussian resources are also necessary for entanglement distillation and quantum error correction~\cite{eisert2002distilling,fiuravsek2002gaussian,niset2009no}.
Because of these distinct differences and similarities, it would hence be desirable to have a representation of quantum mechanics which puts the discrete and continuous models of quantum computing on an equal footing and delineates the boundary between them. Such a representation would also be helpful to make the CV realm more accessible to those familiar with the DV setting.

In this work, we present such a unifying formalism where quantum states are represented by analytic functions. DV representations are given by polynomials while CV ones are given by holomorphic functions, which are functions that are complex-differentiable in a neighborhood of each point in their domain. By an important theorem in complex analysis, these functions are also analytic, i.e., they can be written as complex power series~\cite{conway2012functions}. The formalism we use suggests that the duality between particle and wave aspects of physical quantum systems is captured by this duality between analytic (discrete countable) and holomorphic (continuous) aspects of complex functions.

In the CV case, these holomorphic functions naturally decompose into a Gaussian (exponential) part and a root (polynomial) part, the latter depending only on its zeros and accounting for the non-Gaussianity of the corresponding quantum state. The Gaussian part can be viewed as a purely CV resource and is best manipulated by a purely CV formalism---with this perspective, discretizing CV would imply approximating Gaussian functions by polynomials, thus replacing Gaussian states by non-Gaussian ones.
Beyond the Gaussian/non-Gaussian dichotomy, this picture also provides a geometric understanding of the transition between finite-dimensional and infinite-dimensional quantum systems, based on the zeros of these analytic representations. Moreover, this formalism allows us to obtain a natural description of the Gaussian evolution of infinite-dimensional single-mode quantum systems using an exactly solvable Coulomb-like \textit{classical} dynamical system known as the Calogero--Moser system~\cite{calogero1971solution}. We show that the appearance of this system is due to unique Gaussian features of the CV framework such as squeezing. Furthermore, our formalism allows us to analyze infinite-dimensional bosonic quantum computations using basic elements of complex analysis. 

To explain our approach and highlight the transition from DV to CV, we first give some background on analytic representations of quantum states in the following section, before detailing our results.

%------------------------------------------------------------------------

\subsection{Analytic representations of quantum states}
\label{sec:analytic}

Analytic representations provide a unifying way of looking at quantum states, as polynomials in the finite-dimensional case and holomorphic functions in the infinite-dimensional case: qubits may be represented by multivariate polynomials by mapping each computational basis state to a different monomial: for $(b_1,\dots,b_n)\in\{0,1\}^n$,
\begin{equation}
    \ket{b_1\dots b_n}\mapsto\prod_{k=1}^nz_k^{b_k}.
\end{equation}
The resulting polynomials have degree at most $1$ in each variable.
Similarly, qudits can also be represented by multivariate polynomials by mapping each computational basis state to a different monomial: for $(d_1,\dots,d_n)\in\{0,\dots,d-1\}^n$,
\begin{equation}
    \ket{d_1\dots d_n}\mapsto\prod_{k=1}^n\sqrt{\binom{d-1}{d_k}}z_k^{d_k}.
\end{equation}
The resulting polynomials have degree at most $d-1$ in each variable. With this normalization, the polynomial representation can be interpreted as the inner product between the state and a tensor product of unnormalized $SU(2)$ coherent states~\cite{radcliffe1971some,arecchi1972atomic}, defined as:
\begin{equation}
    \ket z_d=\sum_{i=0}^{d-1}\sqrt{\binom{d-1}i}z^i\ket i,
\end{equation}
for all $z\in\mathbb C$.
Once normalized, these states form an overcomplete family of the single-qudit Hilbert space span$(\ket i)_{i=0\dots d-1}$ with respect to the metric $\frac d\pi\frac{d^2z}{(1+|z|^2)^2}$. As a consequence, the polynomial representation of a qudit state of $n$ subsystems can be thought of as its wave-function on the phase space $(\mathbb C\cup\{+\infty\})^n$ (note that this phase space is distinct from the discrete phase space usually associated to finite-dimensional systems~\cite{gross2006hudson}). 

In the case of a single subsystem, this phase space corresponds to the Riemann sphere by anti-stereographic projection (see Table~\ref{tab:analytic}). In particular, the usual Bloch sphere representation of a single-qubit state is given by the anti-stereographic projection of the complex zero of its polynomial representation, while the Majorana stellar representation of a single-qudit state~\cite{majorana1932atomi} is given by the anti-stereographic projection of the complex zeros of its polynomial representation, counted with multiplicity. Hence, a qubit may be represented by a single point on the sphere, while a qudit of dimension $d$ may be represented by $d-1$ indistinguishable points on the sphere, consistent with the fact that a qudit state of dimension $d$ may be realized as the symmetric state of $d-1$ single-qubit states~\cite{harrow2013church}.

    \begin{center}
\begin{table}[t!]
\addtolength{\leftskip} {-0.7cm} % increase (absolute) value if needed
    \addtolength{\rightskip}{-0.7cm}
        \begin{tabular}{ || m{2cm} | m{1.9cm}| m{2.8cm} | m{4cm} | m{3.2cm}|| } 
 \hline
 Dimension & States & Analytic \mbox{representation} & Factorization & Geometric \mbox{representation} \\ [0.5ex] 
 \hline\hline
 2 (qubit) & $\alpha\ket0+\beta\ket1$ & $\alpha+\beta z$ & $(z-\frac\alpha\beta)$ & \includegraphics[width=32mm]{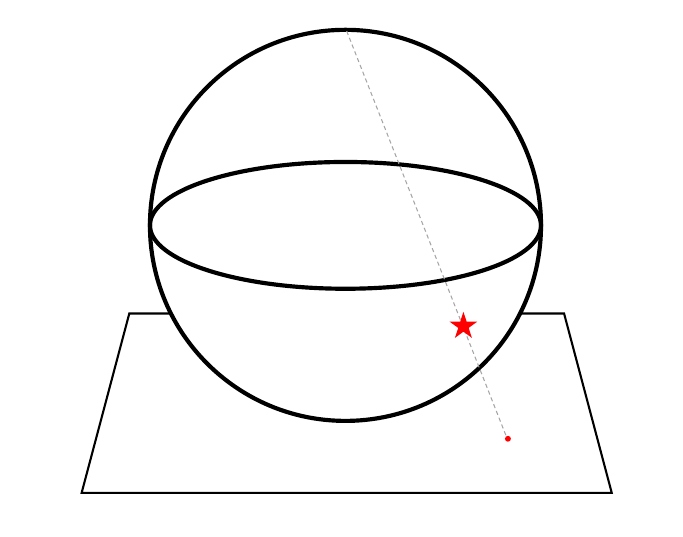}\\ 
 \hline
 $d$ (qudit) & $\sum_{j=0}^{d-1}\alpha_j\ket j$ & $\sum_{j=0}^{d-1}\sqrt{\binom{d-1}j}\alpha_jz^j$ & $\prod_{n=1}^{d-1}(z-\lambda_n)$ & \includegraphics[width=32mm]{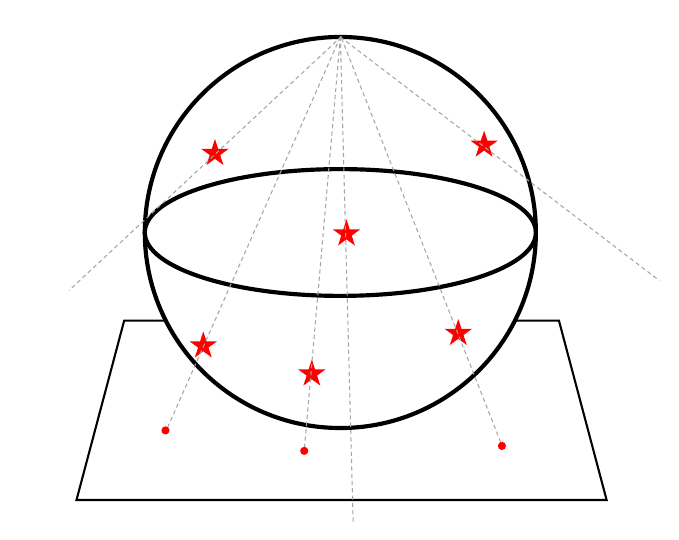} \\
 \hline
 $+\infty$ (mode) & $\sum_{n\ge0}\alpha_n\ket n$ & $\sum_{n\ge0}\frac1{\sqrt{n!}}\alpha_nz^n$ & $z^k\prod_n\left(1-\frac z\lambda_n\right)e^{\frac z{\lambda_n}+\frac12\frac{z^2}{\lambda_n^2}}$ $\times e^{-\frac12 az^2+bz}$ &  \includegraphics[width=32mm]{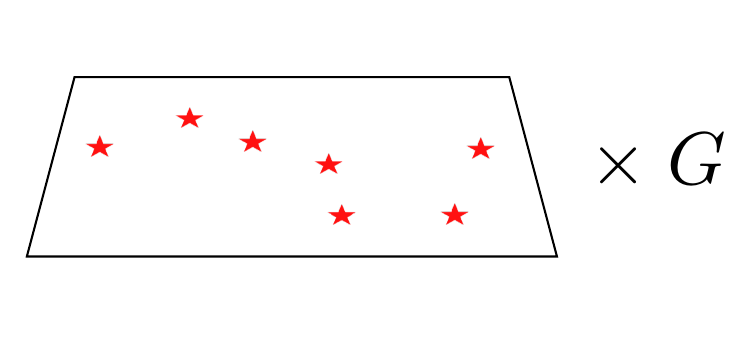}\\
 \hline
        \end{tabular}
\caption{Analytic representations of quantum states of a single system. These states can be described using the distribution of the zeros of their analytic representation in phase space. In the finite-dimensional cases, the phase space is the extended complex plane, which may be represented using the Riemann sphere by anti-stereographic projection. In the infinite-dimensional case, the phase space is the complex plane and the description is complemented by Gaussian parameters $a$ and $b$ defining a Gaussian function $G(z):=e^{-\frac12az^2+bz}$.}
\label{tab:analytic}
\end{table} 
    \end{center}

A similar construction can be done in the infinite-dimensional setting for single bosonic modes, with orthonormal basis $\{\ket n\}_{n\in\mathbb N}$, using the canonical coherent states:
\begin{equation}
    \ket z_{\infty}=\sum_{n\ge0}\frac{z^n}{\sqrt{n!}}\ket n.
\end{equation}
This construction gives rise to the so-called stellar function~\cite{chabaud2020stellar} in Segal--Bargmann space~\cite{segal1963mathematical,bargmann1961hilbert}, which is the space of holomorphic functions $F^\star$ over the complex plane satisfying the normalization condition
\begin{equation}\label{eq:SBnorm1}
\frac1{\pi}\int_{z\in\mathbb C}|F^\star(z)|^2e^{-|z|^2}d^{2}z<+\infty.
\end{equation}
These functions may be thought of as phase-space wave functions of the corresponding quantum states, as their square modulus yields, up to normalization, the Husimi phase-space quasiprobability distribution~\cite{husimi1940some}. Using the Hadamard--Weierstrass factorization theorem~\cite{conway2012functions}, such a function can be decomposed as
\begin{equation}\label{eq:HW1}
    F^\star(z)=\sum_{n\ge0}\frac1{\sqrt{n!}}\alpha_nz^n=e^{-\frac12 az^2+bz+c}z^k\prod_n\left(1-\frac z{\lambda_n}\right)e^{\frac z{\lambda_n}+\frac12\frac{z^2}{\lambda_n^2}},
\end{equation}
where $a,b,c\in\mathbb C$, $k\in\mathbb N$ is the multiplicity of $0$ as a root of $F^\star$, and $\{\lambda_n\}_n$ is the (possibly infinite) set of non-zero roots of $F^\star$, which is discrete since $F^\star$ is holomorphic.
Hence, a single-mode quantum state can be represented by the distribution of the zeros of its stellar function on the phase space $\mathbb C$, together with some additional Gaussian exponents. As it turns out, this representation can be thought of as the limit of contracted finite-dimensional polynomial representations for qudits~\cite{ricci1986contraction} when the dimension goes to infinity, which amounts to sending the radius of the Riemann sphere to infinity, thus mapping the zeros to the bottom tangent plane~\cite{arecchi1972atomic} (informally, $\ket{\frac z{\sqrt d}}_d\rightarrow\ket z_\infty$ when $d\rightarrow\infty$). The analytic representations for single systems are summarized in Table~\ref{tab:analytic}.

For single-qubit and single-qudit states, zeros of the analytic representation provide a faithful description. This is because the zeros contain all the information about the state and the number of zeros does not change under unitary dynamics. The description based on zeros of the stellar representation for single-mode quantum states, on the other hand, is not faithful: the number of zeros can vary under unitary operations, and the zeros do not contain all the information about the state. This is because the contracting limit of the finite-dimensional representations sends some of the zeros to the complex plane, and the rest of the zeros to infinity---informally, these contribute to the Gaussian term in the factorization of the infinite-dimensional holomorphic representation (see Eq.~(\ref{eq:HW1})).

What is the meaning of these zeros in the infinite-dimensional case and what unitary operations preserve their number? Since these zeros do not contribute to the Gaussian part of the stellar function, one could expect that they are signatures of the non-Gaussian character of a quantum state (recall the Gaussian/non-Gaussian dichotomy).
Indeed, a strong connection between the zeros of the stellar function and the non-Gaussian structure of quantum states was recently uncovered in the single-mode setting~\cite{chabaud2020stellar}. As it turns out, the unitary operations that preserve the number of zeros are exactly Gaussian unitary operations. The number of zeros then provides an operational measure of the non-Gaussian properties of quantum states, called the stellar rank, which induces a non-Gaussian hierarchy among single-mode pure quantum states, extended to mixed states by a convex roof construction. At rank $0$ lie Gaussian states and mixtures of Gaussian states, while quantum non-Gaussian states populate all higher ranks: the stellar rank has an operational meaning as the minimal number of elementary non-Gaussian operations needed to engineer a state from the vacuum, together with Gaussian unitary operations. 
One may then obtain an equivalent characterization of single-mode Gaussian unitaries as follows: for each $n\in\mathbb N$, define $R_n$ as the subset of Segal--Bargmann space of functions of the form $P\times G$, where $P$ is a polynomial of degree $n$ and $G$ is a Gaussian function, corresponding to the set of states of stellar rank $n$. Then, the unitary operations that leave the spaces $R_n$ invariant for all $n\in\mathbb N$ are exactly Gaussian unitary operations. In particular, one may describe Gaussian evolutions of infinite-dimensional single-mode quantum states by the motion of the zeros of their stellar function, together with an evolution of Gaussian parameters---since any description that is based on the zeros must also keep track of the Gaussian exponents for completeness---which will be a starting point of our approach.

In the multimode case, the stellar function is obtained similarly by mapping each computational basis state to a different monomial: for $(n_1,\dots,n_m)\in\mathbb N^m$, 
\begin{equation}
    \ket{n_1\dots n_m}\mapsto \prod_{k=1}^m\frac1{\sqrt{n_k!}}z_k^{n_k}.
\end{equation}
The resulting holomorphic functions are elements of the $m$-mode Segal--Bargmann space~\cite{segal1963mathematical,bargmann1961hilbert}. Defining the creation and annihilation operators $\hat a^\dag_k$ and $\hat a_k$ for each mode $k$ as:
\begin{equation}
    \begin{aligned}
        \hat a^\dag_k\ket{n_1\dots n_k\dots n_m}&=\sqrt{n_k+1}\ket{n_1\dots n_k+1\dots n_m},\\
        \hat a_k\ket{n_1\dots n_k\dots n_m}&=\sqrt{n_k}\ket{n_1\dots n_k-1\dots n_m},
    \end{aligned}
\end{equation}
for all $(n_1,\dots,n_m)\in\mathbb N^m$, these ladder operators have the representations $z_k\times$ (multiplication by the $k^{th}$ variable) and $\partial_{z_k}$ (partial derivative with respect to the $k^{th}$ variable) on Segal--Bargmann space, respectively~\cite{vourdas2006analytic}. Operators that are functions of creation and annihilation operators are then represented as the same functions of the multiplication and partial derivative operators: for instance, Gaussian Hamiltonians, which are quadratic functions of creation and annihilation operators~\cite{ferraro2005gaussian} correspond to polynomials of degree less or equal to $2$ in these differential operators over Segal--Bargmann space.

Let us consider the example of Boson Sampling~\cite{Aaronson2013} to illustrate the usefulness of analytic representations beyond the geometric intuition they provide. The hardness of classically simulating Boson Sampling computations stems from their output probability distributions being related to permanents of unitary matrices, which are notoriously hard to compute~\cite{valiant1979complexity}, and the stellar representation makes this relation apparent: the input state of Boson Sampling is the state $\ket{1\dots1\,0\dots0}$, corresponding to $n$ single photons in the first $n$ input modes of an $m$-mode interferometer; up to normalization, its stellar representation is given by $z_1\cdots z_n$; the interferometer, described by some $m\times m$ unitary matrix $U=(u_{jk})_{1\le j,k\le m}$, acts linearly on the vector of variables $(z_1,\dots,z_m)$, yielding the output state representation $\prod_{k=1}^n(\sum_{j=1}^mu_{jk}z_j)$; finally, the probability of an outcome $(s_1,\dots,s_m)\in\{0,1\}^m$ is related to the coefficient of $z_1^{s_1}\cdots z_m^{s_m}$ in the previous expression, which directly yields the permanent of a submatrix of $U$.
Note that the polynomial representation introduced by Aaronson and Arkhipov in~\cite{Aaronson2013} exactly corresponds to the stellar function of the states, i.e., their phase-space wave function, or expansion in canonical coherent state basis.

%------------------------------------------------------------------------

\subsection{Our results}
\label{sec:results}

Having outlined analytic representations of quantum states and reviewed properties of the single-mode stellar representation, we turn to our main contributions.
We derive our results gradually towards the holomorphic representation of quantum computations: starting from the single-mode case (Sec.~\ref{sec:CMintro}), we generalize properties of the stellar representation to the multimode case (Sec.~\ref{sec:multiintro}), and we apply these results to analyze natural infinite-dimensional quantum computing models using holomorphic functions (Sec.~\ref{sec:HQCintro}). Since non-Gaussian elements of a quantum computation can be thought of as resources for quantum speedup, we put emphasis on the stellar hierarchy and the operations that conserve the stellar rank in all of our results, making the identification of the non-Gaussian resources self-evident in each case.

\subsubsection{Single-mode rank-preserving evolutions}
\label{sec:CMintro}

We first consider the single-mode case and we use the stellar representation to study the unitary evolutions that leave invariant the stellar rank, i.e., the number of zeros of the stellar function. These evolutions are induced by Gaussian Hamiltonians (see Sec.~\ref{sec:analytic}).

The stellar representation describes single-mode pure quantum states in infinite-dimensional Hilbert spaces using the zeros of their Segal--Bargmann holomorphic representation, the stellar function, together with additional Gaussian parameters. Rank-preserving evolutions are then naturally described by updating these Gaussian parameters, and looking at the motion of the zeros in the complex plane. In this picture, the time evolution of a stellar function $F^\star$ under a Hamiltonian $\hat H=H(\hat a^\dag,\hat a)$ which is a function of the creation and annihilation operators is given by Schr\"odinger's equation in Segal--Bargmann space~\cite{schrodinger1926undulatory,leboeuf1991phase}:
\begin{equation}\label{eq:schrostellarintro}
    i\partial_tF^\star=H(z,\partial_z)F^\star,
\end{equation}
with the convention $\hbar=1$.

Since they form a dense subset of the Hilbert space~\cite{chabaud2020stellar}, we focus on pure states of finite stellar rank, i.e., the states whose stellar functions have a finite number of zeros, being of the form $P\times G$, where $P$ is a polynomial and $G$ is a Gaussian function (see Eq.~(\ref{eq:HW1})). Remarkably, we show that the evolution of the Gaussian parameters and of the zeros of the stellar function are decoupled under Gaussian Hamiltonians (Theorem~\ref{th:decoupling}). Up to a conformal evolution of the Gaussian parameters, this yields a dual description of single-mode Gaussian \textit{quantum} evolutions based on \textit{classical} many-body dynamics, where the zeros of the stellar function correspond to classical particles in the plane. 

Next, we study these dynamics and we show that they are governed by a Couloumb-like dynamics known as the classical Calogero--Moser model, a well-known Hamiltonian model in the field of integrable systems~\cite{calogero1971solution,calogero1976exactly,moser1976three,ol1976geodesic,olshanetsky1981classical,calogero2008calogero}. 
This model is ubiquitous across many disciplines~\cite{polychronakos2006physics} and captures the dynamics of poles of solutions of integrable differential equations. 
Here, we unveil a new connection between this model and quantum information theory. In particular, we show that this model governs the motion of the zeros of complex holomorphic solutions $u$ of the partial differential equation:
\begin{equation}\label{eq:shear}
    i\partial_tu+(z+\partial_z)^2u=0,
\end{equation}
relating to quadratic Gaussian evolutions of the stellar function by Eq.~(\ref{eq:schrostellarintro}), such as shearing and squeezing.

We use the integrability of the Calogero--Moser model to solve the equations of motions analytically in terms of determinants of efficiently computable matrices, obtaining closed formulas for the evolution of the stellar representation under Gaussian unitary evolutions (Lemmas~\ref{lem:evoD},~\ref{lem:evoR} and~\ref{lem:evoP}). For completeness, we also provide an alternative derivation of the evolution of the stellar function using the representation of Gaussian unitary operators in terms of differential operators acting on Segal--Bargmann space (Lemmas~\ref{lem:directevoD},~\ref{lem:directevoR} and~\ref{lem:directevoS}).

\subsubsection{The multimode stellar hierarchy}
\label{sec:multiintro}

We then turn to the multimode case, and study the properties of the stellar hierarchy induced by the stellar rank, generalized to the multimode setting.

\begin{figure}[ht!]
    \centering
        \includegraphics[width=\columnwidth]{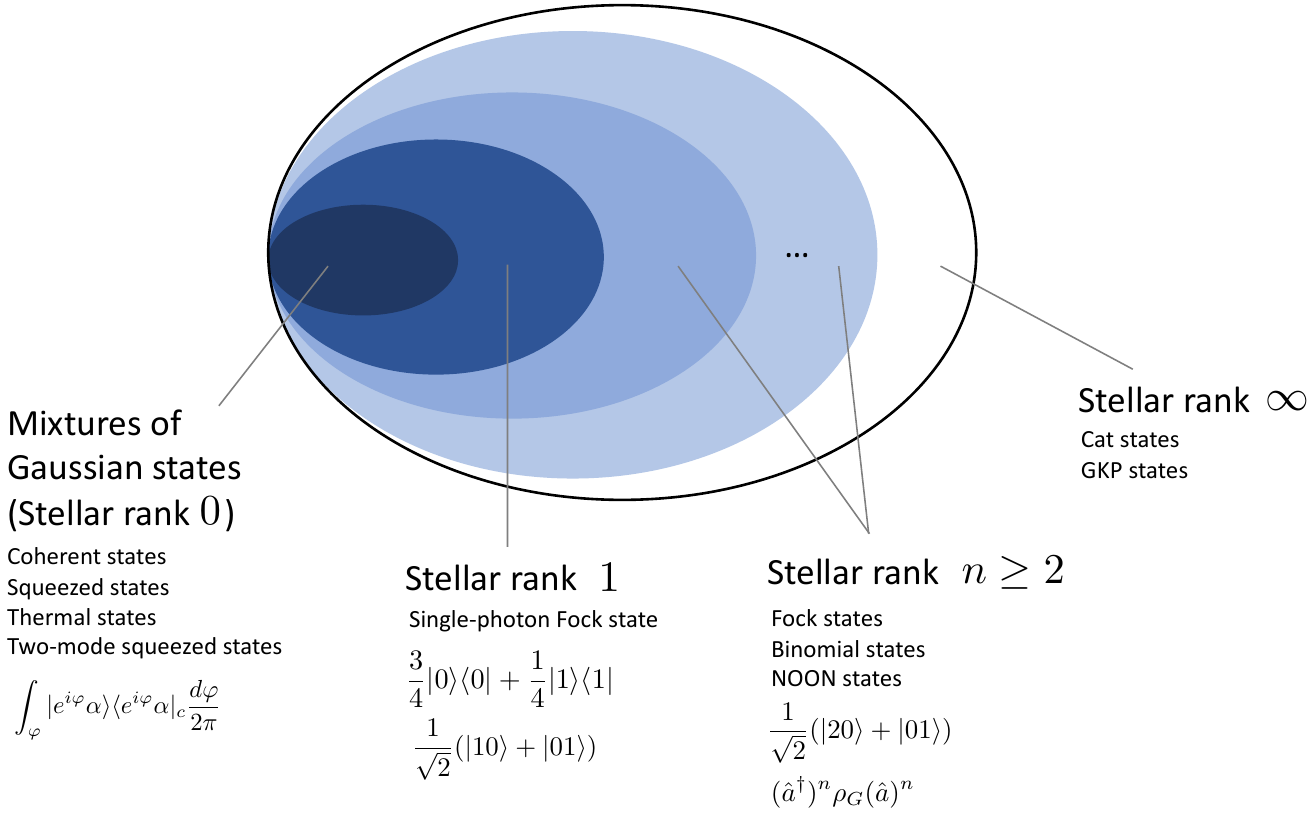}
        \caption{The multimode stellar hierarchy (extended from~\cite{PRXQuantum.2.030204}). For each rank, we give examples of states that are commonly encountered in the literature. The last example of stellar rank $0$ is a rotation-invariant mixture of coherent states. The density matrix $\rho_G$ denotes a Gaussian state. Pure states of finite stellar rank form a dense subset of the Hilbert space for the trace norm, while pure states of bounded stellar rank form closed subsets of the Hilbert space.}
        \label{fig:hierarchy}
\end{figure}

In that case, pure quantum states over $m$ modes are represented over phase space by their stellar function, a holomorphic function $F^\star$ over $\mathbb C^m$.
Unlike in the single-mode case, the construction of a non-Gaussian stellar hierarchy based on the number of zeros of such functions is no longer relevant in the multimode setting, since multivariate holomorphic functions may have an uncountable infinity of zeros. Instead, one may introduce a hierarchy of subsets $\bigcup_{k=0}^n\bm R_k$ of functions of the form $P\times G$, where $P$ is a multivariate polynomial of degree at most $n$ and $G$ is a multivariate Gaussian function. The multimode stellar rank is then given by the degree $n$ of the polynomial. Like in the single-mode case, the holomorphic functions that do not admit a decomposition $P\times G$ are of infinite rank in the hierarchy. The stellar rank of mixed states is then obtained by a convex roof construction, and the single-mode case exactly corresponds to the case $m=1$.

Here we show that the important single-mode properties of the stellar hierarchy also hold in the multimode setting: in particular, the unitary operations that leave invariant the spaces $\bm R_n$ of multivariate holomorphic functions of stellar rank $n$ for all $n\in\mathbb N$ are exactly multimode Gaussian unitary operations (Theorem~\ref{th:invar}), and the stellar rank thus induces a multimode non-Gaussian hierarchy. We further show that the stellar rank is non-increasing under Gaussian measurements and Gaussian channels (Corollary~\ref{coro:nonincr}), making it a non-Gaussian resource monotone.

We prove that the multivariate stellar functions of a given rank $n$ carry an operational interpretation similar to the single-mode case: these functions correspond to quantum states that cannot be engineered from the vacuum with less than $n$ applications of an elementary non-Gaussian operation (namely, the creation or annihilation operator), together with Gaussian unitary operations (Lemma~\ref{lem:addition}). This notion which may be thought of intuitively as a CV equivalent of the $T$-count in DV quantum circuits~\cite{amy2013meet,beverland2020lower}. Note however that there are important differences between these notions---unlike $T$ gates, creation and annihilation operators are not unitary---and other relevant infinite-dimensional generalizations of the $T$-count may be defined, depending on the quantum architecture at hand. For instance, the number of cubic phase gates provides a direct generalization of the $T$-count through the Gottesman--Kitaev--Preskill encoding~\cite{Gottesman2001}. This characterization of the stellar rank allows us to obtain two standard forms for states of finite stellar rank (Lemmas~\ref{lem:decomp1} and~\ref{lem:decomp2}) and to classify states commonly encountered in the literature in the stellar hierarchy (see Fig.~\ref{fig:hierarchy}). It also implies that the stellar rank is of direct relevance to experimental non-Gaussian state engineering in the multimode case, since applications of the creation and annihilation operators are commonly used to produce non-Gaussianity in current CV quantum optics experiments, through the physical operations of photon-addition and photon-subtraction, respectively~\cite{lvovsky2020production}.

We also study the topology of the multimode stellar hierarchy of pure states with respect to the trace norm, obtaining two main results: like in the single-mode case, the set of states of finite stellar rank is dense (Lemma~\ref{lem:dense}) and the set of states of stellar rank bounded by $k$ is closed for all $k\in\mathbb N$ (Theorem~\ref{th:robust}). These results have two separate important consequences: firstly, we may restrict to states of finite stellar rank in the multimode case without loss of generality; secondly, the multimode stellar hierarchy is robust and can be observed experimentally by direct fidelity estimation, since the set of pure states of stellar rank bounded by $k$ being closed implies that fidelities with target pure states of a given stellar rank provide witnesses for the stellar rank of (mixed) experimental states~\cite{chabaud2021certification,fiuravsek2022efficient}.

Finally, we relate entanglement to factorization properties of holomorphic functions: an $m$-mode pure quantum state is separable over a partition $I,J$ of the modes if and only if its stellar function can be factorized as $F^\star(\bm z)=F^\star_I(\bm z_I)\times F^\star_J(\bm z_J)$, where $\bm z=(z_k)_{1\le k\le m}$, $\bm z_I=(z_i)_{i\in I}$ and $\bm z_J=(z_j)_{j\in J}$, and where $F^\star_I$ and $F^\star_J$ are elements of Segal--Bargmann space over $|I|$ and $|J|$ variables, respectively.
This implies that a pure Gaussian state---whose stellar function is Gaussian---is separable over a partition $I,J$ of the modes if and only if its stellar function does not contain any cross-term exponent of the form $z_iz_j$ for $i\in I$ and $j\in J$. 

As a consequence, we obtain that entangled states of finite stellar rank may display both Gaussian and non-Gaussian entanglement, and the latter can be thought of as the entanglement of a finite-dimensional system within the infinite-dimensional quantum state: given a partition $I,J$ of the modes, we show that an $m$-mode stellar function $F^\star(\bm z)=P(\bm z)G(\bm z)$ of a pure quantum state  $\ket{\bm\psi}$ of finite stellar rank admits a unique decomposition:
\begin{equation}
    F^\star(\bm z)=\left(\sum_{k=1}^r\tilde\psi_kP_I^{(k)}(\bm z_I)P_J^{(k)}(\bm z_J)\right)G_I(\bm z_I)G_J(\bm z_J)e^{\sum_{(i,j)\in I\times J}\lambda_{ij}z_iz_j},
\end{equation}
where $(\tilde\psi_k)_k\in\mathbb C^r$, where $r\in\mathbb N^*$ is the Schmidt rank~\cite{NielsenChuang} of a finite-dimensional pure state $\ket C$ of stellar function $P(\bm z)$, and where $\lambda_{ij}$ are the cross-term exponents $z_iz_j$ of the stellar function $G(\bm z)$ of a Gaussian state $\ket G$. We show that the state $\ket{\bm\psi}$ is separable if and only if $r=1$ and $\lambda_{ij}=0$ for all $(i,j)\in I\times J$, i.e., both the finite-dimensional state and the Gaussian state are separable over the bipartite partition $I,J$.

\subsubsection{Holomorphic representation of quantum computations}    
\label{sec:HQCintro}

\begin{figure}[t]
    \centering
        \includegraphics[height=3cm]{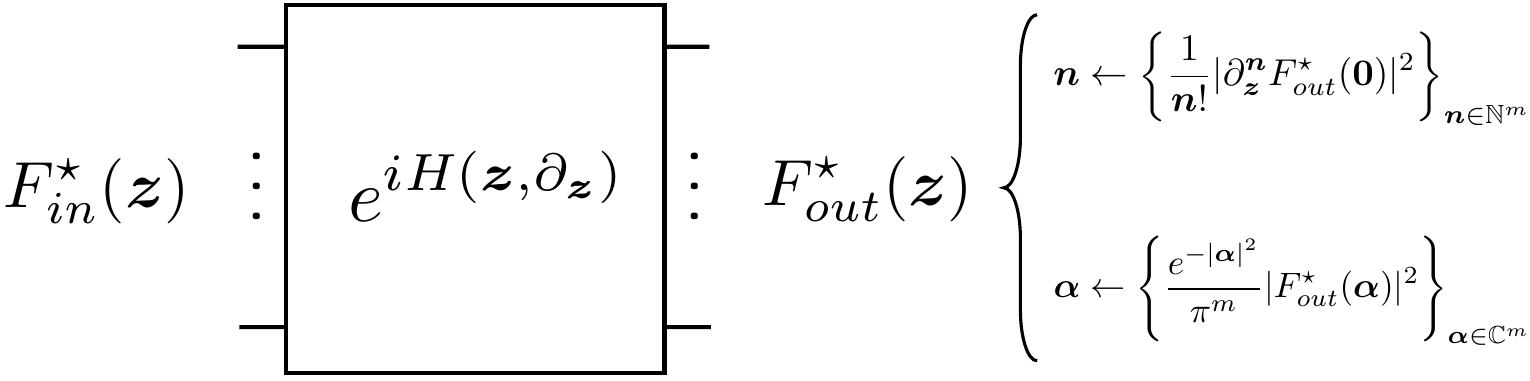}
        \caption{Holomorphic representation of bosonic quantum computations. The input is described by holomorphic functions over $\mathbb C^m$. The evolution is described by unitary operators generated by self-adjoint differential operators. For each mode, the outcome is either sampled from the values of the output holomorphic function or from its coefficients. These measurements can be non-adaptive or adaptive.}
        \label{fig:HQC}
\end{figure}

Having generalized the properties of the stellar hierarchy to the multimode case, we turn to the quantum computing aspect of our work. We study bosonic quantum computations using basic elements of complex analysis as follows: input states are mapped to holomorphic functions over $m$ variables, unitary gates to differential operators on the Segal--Bargmann space, and the outcome of the computation is obtained using one of two possible measurements for each mode (see Fig.~\ref{fig:HQC}). Up to normalization, the first measurement we consider is continuous and amounts to sampling from the square modulus of the output holomorphic function, yielding a complex outcome, while the second measurement is discrete and amounts to sampling from the sequence of square modulus of the coefficients of the analytic series defining the holomorphic function, yielding an integer outcome. In other words, the first type views the holomorphic function as a continuous distribution while the second views its monomial coefficients as a discrete distribution. 

The relation between these two type of measurements---probing continuous and discrete properties of the function, respectively---may be thought of as an instance of the concept of wave-particle duality in quantum mechanics.
Indeed, while the above choice of possible measurements is natural from a complex analysis perspective, we show that it corresponds to measurements commonly implemented in CV experiments, such as homodyne and heterodyne detection for the continuous measurement, and particle-number for the discrete measurement. Similarly, we show how differential unitary operators on Segal--Bargmann space correspond to physical Hamiltonians that are self-adjoint functions of the creation and annihilation operators of each mode.

Taking advantage of the structure of the multimode stellar hierarchy, we study an important subclass of bosonic computations: those that leave the stellar rank invariant, which we call Rank-Preserving Quantum Computations (RPQC). These computations are multimode analogues of the single-mode evolutions from Sec.~\ref{sec:CMintro}. In particular, by looking at sections of holomorphic functions, i.e., fixing all but one variable, we show how one can think of these computations as implementing sequences of Calogero--Moser dynamical evolutions, together with a conformal evolution of Gaussian parameters.

\begin{figure}[t]
    \centering
        \includegraphics[height=3cm]{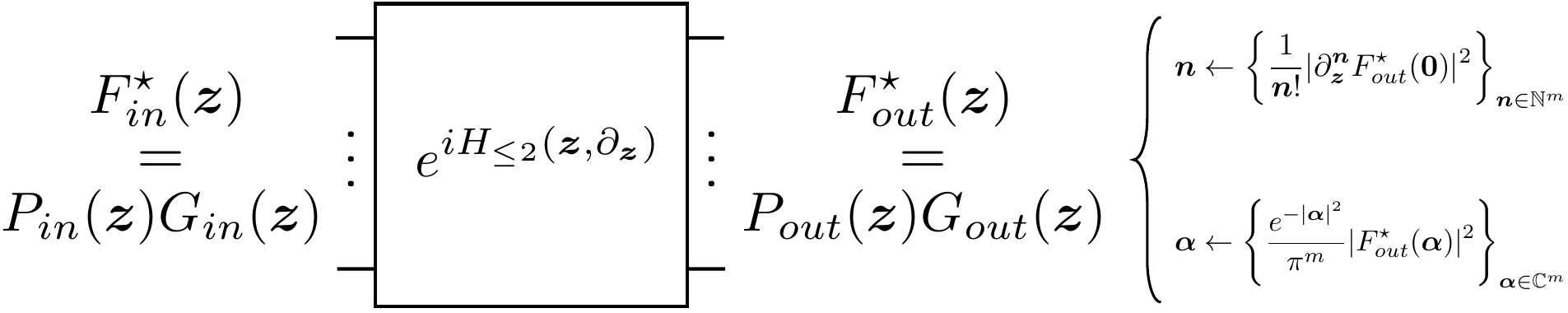}
        \caption{Holomorphic representation of Rank-Preserving Quantum Computations. The stellar function of the input is a holomorphic function over $\mathbb C^m$ of the form $P\times G$, where $P$ is a polynomial and $G$ is a Gaussian function. The evolution is generated by self-adjoint differential operators which are quadratic functions of the operators $z_k\times$ and $\partial_{z_k}$, for each $k\in\{1,\dots,m\}$. For each mode, the outcome is either sampled from the values of the output holomorphic function or from its coefficients. These measurements can be non-adaptive or adaptive.}
        \label{fig:PxG}
\end{figure}

The RPQC are described as follows: without loss of generality, the input has finite stellar rank, i.e., its stellar function is a holomorphic functions of the form $P\times G$, where $P$ is a multivariate polynomial and $G$ is a multivariate Gaussian function, unitary gates are generated by Hamiltonians $H(\bm z,\partial_{\bm z})$ that are quadratic functions, and the outcome of the computation is obtained using either the continuous or discrete measurement for each mode (see Fig.~\ref{fig:PxG}).
The RPQC correspond to infinite-dimensional quantum computations with Gaussian evolution and are generalizations of Boson Sampling and Gaussian quantum computing. At a high level, these may also be thought of as a CV version of DV Clifford computations with input magic states~\cite{bravyi2005universal,yoganathan2019quantum}. We consider both non-adaptive and adaptive versions of these computations, that is, when subsequent unitary operations in the computation may depend on the outcome of intermediate measurements. 

\begin{table}[t!]
    \begin{center}
    \addtolength{\leftskip} {-0.3cm} % increase (absolute) value if needed
    \addtolength{\rightskip}{-0.3cm}
        \begin{tabular}{ || m{7cm} | m{3.5cm}| m{3.5cm} || } 
 \hline
 Model & Continuous-variable measurement \mbox{complexity} & Discrete-variable measurement \mbox{complexity}\\ [0.5ex] 
 \hline\hline
 \!\!\vspace{0.4cm}\includegraphics[width=70mm]{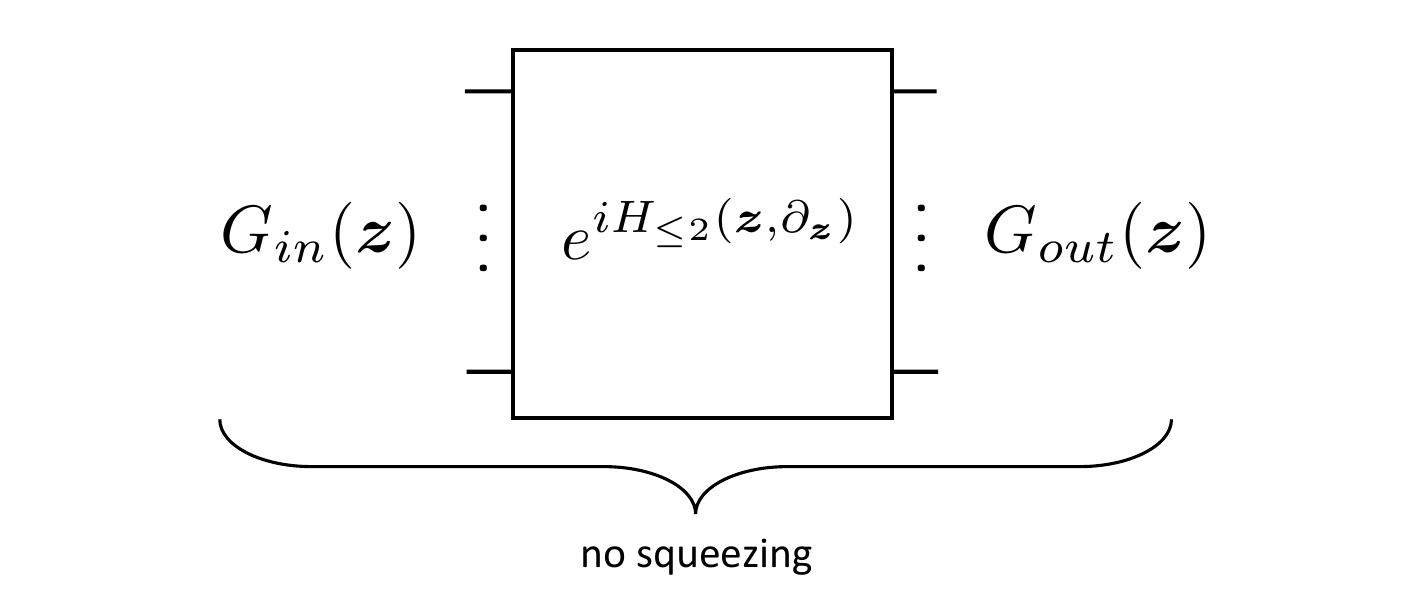} & \textsf{P} (Gaussian quantum computing~\cite{Bartlett2002}) & \textsf{P} (Coherent state \mbox{sampling}~\cite{Aaronson2013})\\
 \hline
 \!\!\vspace{0.3cm}\includegraphics[width=70mm]{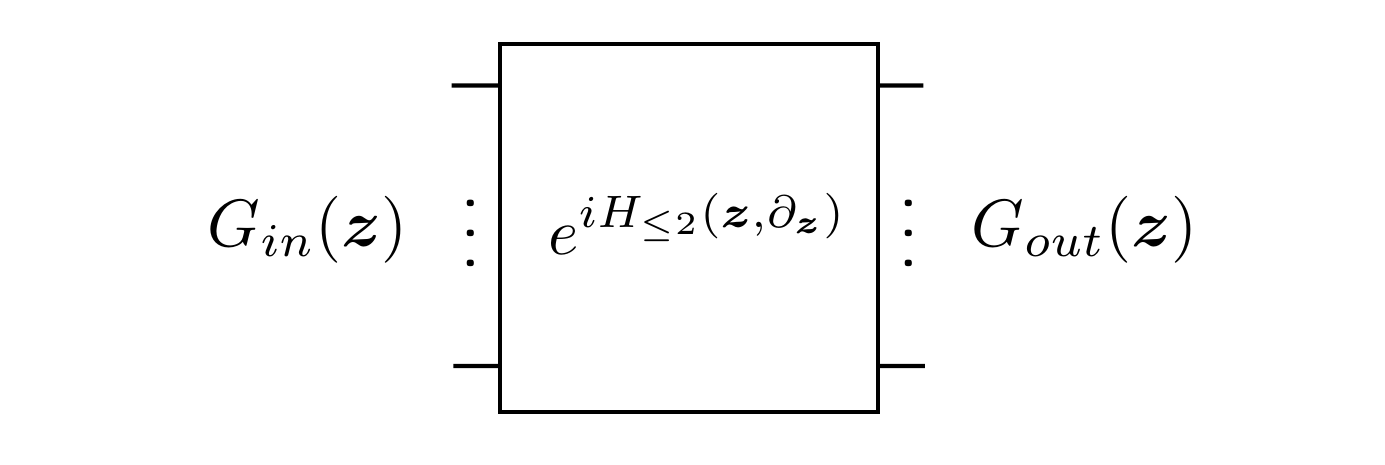} & \textsf{P} (Gaussian quantum computing~\cite{Bartlett2002}) & $\#$\textsf{P} (Gaussian Boson Sampling~\cite{Lund2014,Hamilton2016})\\
 \hline
\vspace{0.3cm}\includegraphics[width=70mm]{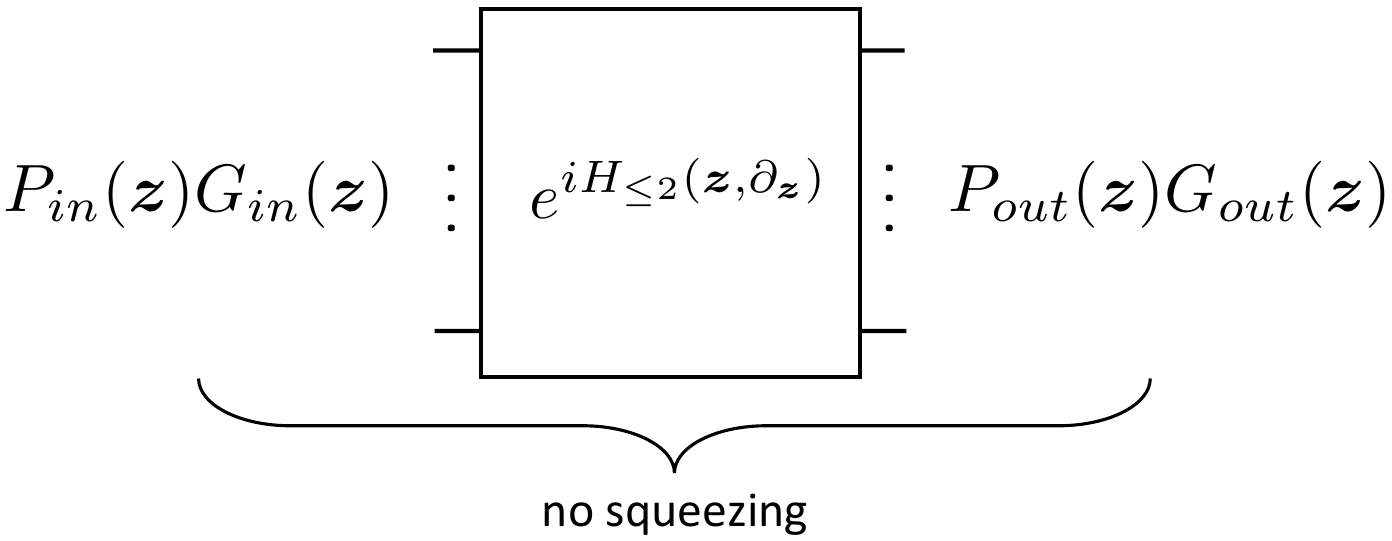} & \textsf{P} (Time-reversed coherent state sampling~\cite{Aaronson2013}) & \mbox{$\#$\textsf{P} (Boson} \mbox{Sampling~\cite{Aaronson2013})}\\
  \hline
 \vspace{0.3cm}\includegraphics[width=70mm]{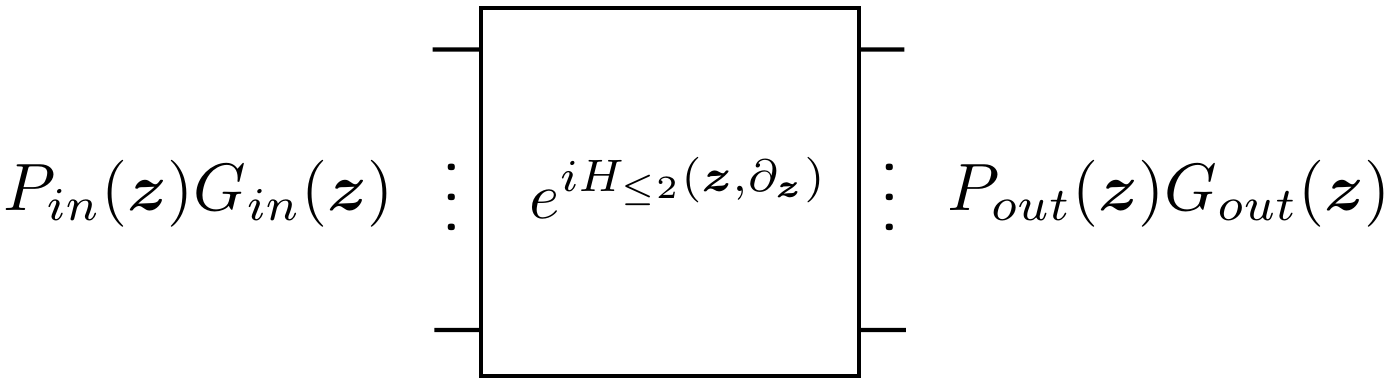} & $\#$\textsf{P} (Continuous-variable Boson Sampling~\cite{Lund2017,Chakhmakhchyan2017,chabaud2017continuous}) & \mbox{$\#$\textsf{P} (Boson} \mbox{Sampling~\cite{Aaronson2013}}, \mbox{Gaussian Boson} Sampling~\cite{Lund2014,Hamilton2016})\\
 \hline
        \end{tabular}
    \end{center}
\caption{Classification of infinite-dimensional quantum computing architectures preserving the stellar rank, based on their computational complexity and presence of Gaussian non-linearity such as squeezing. In all cases, the output probabilities can either be computed in \textsf{P}, or estimating output probabilities is $\#$\textsf{P}-hard in the worst case, i.e., there exist at least one instance with at least one outcome whose probability is $\#$\textsf{P}-hard to estimate multiplicatively up to inverse polynomial precision, which implies worst-case exact sampling hardness~\cite{Aaronson2013}. For each architecture, the classification is obtained using reductions to existing computational models, indicated between parentheses.}\label{tab:CVmodels}
\end{table} 

We show that non-adaptive RPQC capture existing subuniversal infinite-dimensional quantum computational models~\cite{Aaronson2013,Lund2014,Hamilton2016,Douce2017,Lund2017,Chakhmakhchyan2017,chabaud2017continuous}, and we obtain a classification of these models based on their computational complexity and the presence of Gaussian quadratic non-linearity such as squeezing (see Table~\ref{tab:CVmodels}). In particular, our results illustrate that on top of non-Gaussianity, any single-mode Gaussian quadratic non-linearity such as squeezing---captured by the Calogero--Moser model---is mandatory either in the input or in the evolution for quantum speedup in RPQC with continuous measurements or with Gaussian input. This is due to entanglement resulting from the combination of single-mode Gaussian non-linearities and passive linear operations (Gaussian entanglement), or single-mode non-Gaussianity and passive linear operations (non-Gaussian entanglement). This reflects the onset of quantum speedup from the combination of non-Gaussian and entanglement properties of CV quantum states.

Finally, we show that adaptive RPQC capture existing universal infinite-dimensional quantum computational models~\cite{lloyd1999quantum,knill2001scheme,Gottesman2001,Bartlett2002,menicucci2006universal,zhang2006continuous,Gu2009}. We prove that RPQC with adaptive continuous measurements and RPQC with adaptive discrete measurements may both encode \textsf{BQP}-complete computations (Theorems~\ref{th:poweradapCV} and~\ref{th:poweradapDV}), by relating the former to CV computing with Gottesman--Kitaev--Preskill states~\cite{Gottesman2001,baragiola2019all} and the latter to the Knill--Laflamme--Milburn scheme~\cite{knill2001scheme}.

%------------------------------------------------------------------------

\subsection{Related work}

Analytic methods have long been used in quantum mechanics~\cite{segal1963mathematical,bargmann1961hilbert,hall260holomorphic,vourdas2006analytic}, in particular for the study of quantum chaos~\cite{leboeuf1990chaos}. The representation of qudit states using the zeros of a phase-space function goes back to Majorana~\cite{majorana1932atomi} and the study of quantum dynamics based on these zeros was initiated by Leb\oe uf~\cite{leboeuf1991phase}. The case of bosons was considered in~\cite{ellinas1995motion}, although only the motion of zeros of the representation was described and not the evolution of its Gaussian parameters, thus providing an incomplete description of the dynamics. Our work gives a complete description and unveils the connection between Gaussian quantum dynamics and Calogero--Moser classical dynamics.

A notable use of complex analysis tools in CV quantum information was made by Hudson~\cite{hudson1974wigner} and Soto and Claverie~\cite{soto1983wigner} to show that pure states with positive Wigner functions are Gaussian states, a result later extended to the Husimi function~\cite{lutkenhaus1995nonclassical}. These methods were reintroduced more recently with the definition of the single-mode non-Gaussian stellar hierarchy~\cite{chabaud2020stellar}. A definition of the multimode stellar rank subsequently appeared in~\cite{chabaud2020classical}, wherein a classical simulation algorithm of specific CV quantum computations scaling with the stellar rank of the input state was obtained. This algorithm was recently generalized to arbitrary CV quantum computations, scaling with the stellar rank of both the input and the measurement~\cite{chabaud2022resources}. The stellar hierarchy was probed using experimental data in~\cite{chabaud2021certification}. On the other hand, the properties of the multimode stellar hierarchy have not been explored until our work. In particular, we show that our definition of the multimode stellar rank based on the structure of the holomorphic functions is equivalent to that of~\cite{chabaud2020classical}, based on the structure of the corresponding quantum states (Lemma~\ref{lem:decomp2}).

In classical computing there exist several models~\cite{bournez2008survey} such as algebraic circuits~\cite{Valiant79}, the Blum--Schub--Smale (BSS) machine~\cite{BSS2000}, or differential analyzers~\cite{Bush31} to mention a few, that rely on fundamentally different features than boolean models such as Turing machines. For instance, for the BSS machine the notion of halting states is related to stable fixed points of dynamical systems~\cite{BSS2000}, and the analogue of ${\sf NP}$-complete problems such as boolean satisfiability for this model is deciding Hilbert Nullstellensatz~\cite{smale98}.

In quantum computing, a variety of infinite-dimensional models have been defined over the years, including models that can simulate DV quantum computers efficiently using continuous measurements~\cite{lloyd1999quantum,Gottesman2001}, discrete measurements~\cite{knill2001scheme}, or both~\cite{menicucci2006universal}, and models that are argued to be both hard to simulate classically and subuniversal quantumly~\cite{Aaronson2013,Lund2014,Hamilton2016,Douce2017,Lund2017,Chakhmakhchyan2017,chabaud2017continuous}. Our work captures these quantum models within the framework of holomorphic representation.

%------------------------------------------------------------------------

\subsection{Discussion and open questions}

We have introduced a holomorphic representation of bosonic quantum computations based on preparing, evolving and measuring holomorphic functions, and showed that this framework successfully identifies the important resources for bosonic quantum speedups. 

The choice of holomorphic functions as the central mathematical object to describe quantum states allows us to relate the infinite-dimensional setting to its finite-dimensional counterpart through the lens of analytic representations of single systems. It is our hope that readers who are more familiar with the DV aspects of quantum information will find this connection appealing. We also hope that readers with background in different areas in mathematics such as algebra and complex analysis find this work as a suitable perspective for viewing quantum mechanics and computing.

We have studied the multimode generalization of the non-Gaussian stellar hierarchy, which singles out the important subclass RPQC of bosonic computations preserving the stellar rank. Non-adaptive and adaptive versions of these computations capture the existing CV quantum computations and highlight the interplay between different CV resources for quantum speedup.
Our work enables the use of tools from complex analysis for different research directions relating to CV quantum information, and brings interesting additional questions. We mention some of them in the following.

Physical systems are often categorized into the class of integrable (solvable) vs.\ chaotic (not solvable) systems~\cite{d2016quantum}. It is an intriguing question to understand the interplay between this notion of solvability and that of computational complexity. This direction was partially explored in~\cite{ABKM16}. Our results provide a different perspective to this line of research by describing quantum computations in terms of dynamics of classical particles. In particular, we describe single-mode Gaussian dynamics, which are known to be classically simulable on Gaussian inputs, as integrable classical dynamics of Calogero--Moser particles. Can this formulation be extended to non-Gaussian operations as well? These operations are challenging to analyse as they typically involve the motion of infinitely many zeros. We leave this as an interesting future direction.

CV provides a promising path to scalable, universal and fault-tolerant quantum computing, where one of the bottlenecks is the deterministic implementation of non-Gaussian gates~\cite{fukui2021building}. This bottleneck can be reduced to that of preparing non-Gaussian states through measurement-based constructions~\cite{baragiola2019all}. In this context, having access to an operational formalism that places the Gaussian/non-Gaussian dichotomy at its center, such as the stellar representation, is particularly relevant. From our results on the topology of the stellar hierarchy, we expect that a non-Gaussian resource theory~\cite{takagi2018convex,zhuang2018resource,albarelli2018resource} of the stellar rank may be derived, allowing for the study of the cost of CV protocols in terms of non-Gaussian elementary operations~\cite{menzies2009gaussian}, non-Gaussian state conversion rates and optimal finite stellar rank approximations of non-Gaussian states beyond the Gaussian/non-Gaussian dichotomy, in the context of quantum optics in particular.

Our findings on the structure of entanglement within the stellar hierarchy pave the way towards a detailed study of non-Gaussian entanglement using the stellar formalism. In particular, our characterization using finite-dimensional entanglement suggests the derivation of CV entanglement witnesses beyond Gaussian entanglement based on DV entanglement witnesses~\cite{gomes2009quantum,ourjoumtsev2009preparation,PRXQuantum.2.030204}. We also expect that the relation between entanglement and factorization properties of analytic functions can be further exploited.

We have shown that the subclass of adaptive RPQC already contains \textsf{BQP}-complete computations. While it is commonly assumed that \textsf{BQP} circuits can simulate CV quantum computations using discretization and a cutoff of the Hilbert space, there are limited formal results in this direction~\cite{tong2021provably}. Moreover, while performing some form of discretization of CV models is a reasonable way of avoiding the artificial appearance of uncomputable numbers, there exist examples of single-mode CV computations with a polynomial number of gates that generate quantum states of exponential effective dimension (for instance, interleaving a polynomial number of Fourier gates and cubic phase gates~\cite{Gu2009} yields a unitary transformation which maps a quadrature operator $\hat q$ to a polynomial in $\hat q$ and $\hat p$ of exponential degree), making the existence of an efficient cutoff dimension less trivial. Hence, it becomes relevant to ask: are there bosonic computations that cannot be simulated efficiently by \textsf{BQP} circuits? More generally, can a bosonic quantum computation in a Hilbert space of arbitrarily large dimension be simulated by an efficient quantum computation in a Hilbert space of fixed dimension?

Important prerequisites to answering this question are the formalization of continuous-variable quantum computing and the development of quantum complexity theory in the infinite-dimensional setting. Indeed, a great part of the theory developed for DV quantum complexity is still missing for CV. The standard gate model used in this case is the one introduced by Lloyd and Braunstein~\cite{lloyd1999quantum}, which features unitary gates generated by Hamiltonians that are polynomial functions in the creation and annihilation operators of the modes. Strikingly, and although this set of gates is frequently referred to as ``universal'' in the literature, it is not known whether these gates actually form a universal gate set, even in the very weak controllability sense that any state can be reached up to arbitrary precision from any other state using sequences of these gates. Hence, characterizing the trace norm closure of the set of states generated by sequences of such unitary gates starting from, say, the vacuum state, is an interesting open problem (already in the single-mode setting) for the foundations of CV quantum complexity, of which the question of universality of the corresponding set of infinite-dimensional unitary gates is a natural follow-up. 
Relating to this problem, an equivalent of the celebrated Solovay--Kitaev theorem~\cite{dawson2005solovay} has been recently obtained for Gaussian unitary operations under energy-constrained diamond norm~\cite{becker2021energy}. Whether and how this result can be generalized to include non-Gaussian unitary gates is also an interesting open problem. Similarly, the extension of quantum Hamiltonian complexity to the CV setting is an interesting direction. In classical complexity theory, the analogue of the \textsf{NP}-complete $3$-\textsf{SAT} problem for the continuous-variable classical Blum--Shub--Smale machine \cite{BSS2000} is whether a multivariate degree $4$ polynomial has a zero. We expect that the stellar formalism would allow us to formulate similar problems based on zeros of holomorphic functions for quantum complexity classes over continuous variables.

%------------------------------------------------------------------------

\section{Preliminaries}
\label{sec:preli}

Bosonic systems such as quantum states of light are described by quantum theory over infinite-dimensional separable Hilbert spaces.
Such a Hilbert space admits various insightful representations relating to different physical observables. For instance, the expansion in the countable eigenbasis of the particle-number observable---the so-called number basis or Fock basis---is captured by the isomorphism with the space of square-summable complex sequences $\ell^2$. The position- and momentum-like observables, on the other hand, are both captured by the isomorphism with the space of square-integrable functions $L^2(\mathbb R)$, corresponding to position and momentum wave functions, respectively. 

The phase-space formalism gives yet another perspective which places position-like and momentum-like operators on equal footing~\cite{cahill1969density}. It assigns quasiprobability distributions to quantum states, like the Wigner function, allowing for a classification of the latter: states described by Gaussian quasiprobability distributions are dubbed Gaussian~\cite{ferraro2005gaussian} and all the other states are referred to as non-Gaussian~\cite{PRXQuantum.2.030204}. Similarly, operations taking Gaussian states to Gaussian states are called Gaussian operations. These states and operations are both of theoretical and experimental relevance, as quantum computations consisting only of Gaussian elements are the simplest to implement in the lab, but can also be simulated efficiently by classical computers~\cite{Bartlett2002}. In that context, non-Gaussian states and operations can be thought of as both necessary resources and bottlenecks for achieving a quantum speedup over classical computations.

While the phase-space formalism makes use of quasiprobability distributions to describe quantum states, the concept of phase-space wave function is captured by the isomorphism with the Segal--Bargmann space of holomorphic functions over complex space which are square-summable with respect to the Gaussian measure~\cite{segal1963mathematical,bargmann1961hilbert}. In complex analysis, a holomorphic function is one that is complex differentiable in a neighboorhood of any point of its domain in complex space. An analytic function, on the other hand, is one which is, at any point in its domain, equal to a convergent power series. By a crucial theorem in complex analysis, holomorphic functions are analytic~\cite{conway2012functions}. Comparing the relation between holomorphic and analytic functions to the isomorphism between the space of square-integrable functions and the Fock space of square-summable sequences, we interpret this as another manifestation of wave/particle duality when describing quantum states.

In what follows, we introduce some notations and technical background.

%------------------------------------------------------------------------

\subsection{Notations and technical background}

We use bold math for multi-index notations, see Appendix~\ref{app:multi-index}. The symbol $\sqcup$ denotes the disjoint union of sets. Let $m$ denote the number of modes. The corresponding infinite-dimensional separable Hilbert space is denoted $\mathcal H^{\otimes m}$, with orthonormal Fock basis $\{\ket{\bm n}\}_{\bm n\in\mathbb N^m}=\{\ket{n_1}\otimes\dots\otimes\ket{n_m}\}_{(n_1,\dots,n_m)\in\mathbb N^m}$, and an $m$-mode (unnormalized) pure state is expressed as:
\begin{equation}
    \ket{\bm\psi}=\sum_{\bm n\ge\bm0}\psi_{\bm n}\ket{\bm n},
\end{equation}
where $\psi_{\bm n}\in\mathbb C$ for all $\bm n\in\mathbb N^m$, and $\sum_{\bm n\ge\bm 0}|\psi_{\bm n}|^2<+\infty$. Operators over Hilbert spaces are denoted with a hat. 

The creation and annihilation operators of each mode $k\in\{1,\dots,m\}$ are denoted $\hat a_k^\dag$ and $\hat a_k$, respectively (we drop the subscript when there is no ambiguity). They satisfy $[\hat a_k,\hat a^\dag_l]=\delta_{kl}\hat{\mathbb I}$, where $\delta$ is the Kronecker symbol.
These operators allow to define representations of the so-called CCR algebra on the Hilbert space spanned by the Fock basis. We refer the reader to~\cite{Petz1990invitation,derezinski2006introduction} for a discussion of these aspects.
Their action on the single-mode Fock basis is given by
\begin{equation}
    \begin{aligned}
        &\hat a^\dag\ket n=\sqrt{n+1}\ket{n+1},\quad\forall n\in\mathbb N,\\
        &\hat a\ket 0=0,\quad&\\
        &\hat a\ket n=\sqrt{n-1}\ket{n-1},\quad\forall n\in\mathbb N^*.
    \end{aligned}
\end{equation}
For brevity we write $\hat{\bm a}^\dag=(\hat a_1^\dag,\dots,\hat a_m^\dag)$ and $\hat{\bm a}=(\hat a_1,\dots,\hat a_m)$. These operators do not induce completely positive trace-preserving maps, but may nevertheless be approximated by physical heralded operations, such as single-photon addition and single-photon subtraction in the optical setting~\cite{zavatta2004quantum,lvovsky2020production}. The position- and momentum-like operators are defined as
\begin{equation}
    \begin{aligned}
        &\hat q=\frac1{\sqrt2}(\hat a+\hat a^\dag)\\
        &\hat p=\frac1{i\sqrt2}(\hat a-\hat a^\dag),
    \end{aligned}
\end{equation}
with $[\hat q_k,\hat p_l]=i\delta_{kl}\hat{\mathbb I}$. We write $d^2z=d\Re(z)d\Im(z)$, and
\begin{equation}
    d\mu(z)=\frac1\pi e^{-|z|^2}d^2z,
\end{equation}
and
\begin{equation}
    d\mu(\bm z)=\frac1{\pi^m}e^{-|\bm z|^2}d^{2m}\bm z,
\end{equation}
the Gaussian measures over $\mathbb C$ and $\mathbb C^m$, respectively.

\subsection{Phase-space and quasiprobability distributions}
\label{sec:Wigner}

In this section, we give a brief review of the necessary phase-space quasiprobability representations for single-mode quantum states, and we refer the reader to~\cite{cahill1969density,serafini2017quantum} for more in-depth treatments.

The most widely used phase-space representation is the so-called Wigner function~\cite{wigner1997quantum}, defined as
\begin{equation}
    W_\rho(x,p)=W_\rho(\alpha)=2\Tr\left[\hat D(\alpha)\hat\Pi\hat D(\alpha)^\dag\rho\right],
\end{equation}
for all $x,p\in\mathbb R$ and for any a density operator $\rho$, where $\alpha:=\frac{x+ip}{\sqrt2}$. As a direct consequence,
\begin{equation}\label{eq:Wignerbound}
    |W_\rho(x,p)|\le2,
\end{equation}
for all $x,p\in\mathbb R$. The Wigner function is referred to as a quasiprobability distribution: it is a real-valued normalized function:
\begin{equation}
    \Tr(\rho)=1=\int_{\mathbb R\times R}W_\rho(x,p)\frac{dxdp}{2\pi},
\end{equation}
but it can take negative values. The Wigner function allows to compute the overlap of two quantum states $\rho$ and $\sigma$:
\begin{equation}\label{eq:Wigneroverlap}
    \Tr(\rho\sigma)=\int_{\mathbb R\times R}W_\rho(x,p)W_\sigma(x,p)\frac{dxdp}{2\pi}.
\end{equation}
In particular, 
\begin{equation}\label{eq:Wignerpurity}
    \Tr(\rho^2)=\int_{\mathbb R\times R}W_\rho(x,p)^2\frac{dxdp}{2\pi}.
\end{equation}
Another commonly used phase-space representation is the so-called Husimi $Q$ function~\cite{husimi1940some}, defined as:
\begin{equation}
    Q_\rho(x,p)=Q_\rho(\alpha)=\frac1\pi\braket{\alpha|\rho|\alpha},
\end{equation}
for all $x,p\in\mathbb R$ and for any a density operator $\rho$, where $\alpha:=\frac{x+ip}{\sqrt2}$, with $\ket\alpha=e^{-\frac12|\alpha|^2}\sum_{n\ge0}\frac{\alpha^n}{\sqrt{n!}}\ket n$ a coherent state of amplitude $\alpha\in\mathbb C$. This nonnegative quasiprobability distribution is related to the Wigner function by a Gaussian convolution~\cite{cahill1969density}.

\subsection{Gaussian unitary operations and states}

In this section, we review Gaussian unitary operations and states~\cite{ferraro2005gaussian,weedbrook2012gaussian,serafini2017quantum}.
Gaussian unitary transformations are the unitaries $\hat U=e^{i\hat H}$ where $\hat H$ is a Hamiltonian which is a hermitian polynomial of degree less or equal to $2$ in the creation and annihilation operators of the modes.

The single-mode displacement operator of amplitude $\beta\in\mathbb C$ is defined by $\hat D(\beta):=e^{\beta\hat a^\dag-\beta^*\hat a}$. The single-mode squeezing operator of parameter $\xi\in\mathbb C$ is defined by $\hat S(\xi):=e^{\frac12(\xi\hat a^{\dag2}-\xi^*\hat a^2)}$. The single-mode phase-shifting operator of angle $\varphi\in[0,2\pi]$ is defined by $\hat R(\varphi):=e^{i\varphi\hat a^\dag\hat a}$. Up to a global phase, any single-mode Gaussian unitary transformation can be written as a combination of these three operators. These operators satisfy the relations:
\begin{equation}\label{eq:commutDSR}
    \begin{aligned}
        \hat D(\beta)\hat a^\dag\hat D^\dag(\beta)&=\hat a^\dag-\beta^*\\
        \hat S(\xi)\hat a^\dag\hat S^\dag(\xi)&=(\cosh r)\hat a^\dag-e^{-i\theta}(\sinh r)\hat a\\
        \hat R(\varphi)\hat a^\dag\hat R^\dag(\varphi)&=e^{i\varphi}\hat a^\dag,
    \end{aligned}
\end{equation}
where we have set $\xi:=re^{i\theta}$. This yields the braiding relations~\cite{agarwal2012quantum}
\begin{equation}\label{eq:braidDSR}
    \begin{aligned}
        \hat D(\beta)\hat S(\xi)&=\hat S(\xi)\hat D(\gamma)\\
        \hat S(\xi)\hat R(\varphi)&=\hat R(\varphi)\hat S(e^{-2i\varphi}\xi)\\
        \hat R(\varphi)\hat D(\beta)&=\hat D(e^{i\varphi}\beta)\hat R(\varphi),
    \end{aligned}
\end{equation}
where $\gamma=\beta\cosh r-\beta^*e^{i\theta}\sinh r$, for $\xi=re^{i\theta}$, and the product formulas
\begin{equation}\label{eq:prodDSR}
    \begin{aligned}
        \hat D(\alpha)\hat D(\beta)&=e^{\frac12(\alpha\beta^*-\alpha^*\beta)}\hat D(\alpha+\beta)\\
        \hat S(\lambda)\hat S(\mu)&=e^{\frac12i\chi}\hat S(\nu)\hat R(\chi)\\
        \hat R(\varphi)\hat R(\chi)&=\hat R(\varphi+\chi),
    \end{aligned}
\end{equation}
where, writing $\lambda:=r_\lambda e^{i\theta_\lambda},\mu:=r_\mu e^{i\theta_\mu},\nu:=r_\nu e^{i\theta_\nu}\in\mathbb C$ and $t:=-e^{i\theta}\tanh r$, we have $t_\nu:=\frac{t_\lambda+t_\mu}{1+t_\lambda^*t_\mu}$ and $\chi=\arg(1+t_\lambda t_\mu^*)$. This implies that any single-mode Gaussian unitary transformation can be written as a combination of a single displacement, squeezing and phase shift, in any order by playing with the braiding relations in Eq.~(\ref{eq:braidDSR}).
Alternatively, the squeezing with complex parameter may be replaced by the shearing operator $\hat P(s):=e^{is\hat q^2}$, for $s\in\mathbb R$, which acts on the quadrature operators as $\hat P(s)\hat q\hat P^\dag(s)=\hat q$ and $\hat P(s)\hat p\hat P^\dag(s)=\hat p+2s\hat q$. In particular, the shearing operator satisfies
\begin{equation}
        \hat P(s)\hat a^\dag\hat P^\dag(s)=(1-is)\hat a^\dag-is\hat a.
\end{equation}
With Eq.~(\ref{eq:commutDSR}), the shearing operator can thus be written as
\begin{equation}\label{eq:PRS}
    \hat P(s)=\hat R(\varphi)\hat S(\xi),
\end{equation}
where $\varphi=\arg(1-is)$, and $\xi:=re^{i\theta}$ with $r=\sinh^{-1} s$ and $\theta=-\arg(1-is)-\frac\pi2$. Conversely, any squeezing with complex parameter $\xi\in\mathbb C$ may be written as $\hat S(\xi)=\hat R(\varphi)\hat P(s)\hat R(\chi)$, for some $\varphi,\chi,s\in\mathbb R$. Squeezing and shearing thus are equivalent Gaussian quadratic non-linearities (generated by quadratic Hamiltonians) in the sense that they both allow to generate any single-mode Gaussian unitary when combined with displacements and phase shifters.

Gaussian states are defined as the states whose Wigner function is a Gaussian function, and Gaussian unitaries as the operations that map Gaussian states to Gaussian states~\cite{demoen1977completely,giedke2002characterization}. In particular, single-mode Gaussian pure states are obtained from the vacuum state $\ket0$ by Gaussian unitaries. Since $\hat R(\varphi)\ket0=\ket0$ for all $\varphi\in[0,2\pi]$, any single-mode Gaussian pure state may be written as a squeezed coherent state
\begin{equation}
    \hat S(\xi)\hat D(\beta)\ket0,
\end{equation}
for some $\beta,\xi\in\mathbb C$.

A passive linear operation $\hat U$ over $m$ modes corresponds to a unitary transformation of the vector $\hat{\bm a}^\dag$ of creation operators and is a direct multimode generalization of the phase-shifting operation defined above. Its action is described by an $m\times m$ unitary matrix $U$ such that
\begin{equation}\label{eq:actionU}
    \hat U\hat a_k^\dag\hat U^\dag=(U\hat{\bm a}^\dag)_k,
\end{equation}
for all $k\in\{1,\dots,m\}$.
Writing $U=(u_{ij})_{1\le i,j\le m}$, this gives
\begin{equation}\label{eq:commutU}
    \hat U\hat a_k^\dag\hat U^\dag=\sum_{j=1}^mu_{kj}\hat a_j^\dag,
\end{equation}
for all $k\in\{1,\dots,m\}$. A particular passive linear operation over $2$ modes is the one described by the $2\times2$ unitary matrix
\begin{equation}
    H=\frac1{\sqrt2}\begin{pmatrix}
        1&1\\
        1&-1
    \end{pmatrix}.
\end{equation}
In the context of quantum optics, this operation describes the action of a balanced beam splitter over two modes.

Any passive linear operation $\hat U$ over $m$ modes may be decomposed efficiently as a product of beam splitters $\hat H$ and phase-shifters $\hat R$~\cite{reck1994experimental}.
Moreover, the Euler decomposition~\cite{ferraro2005gaussian} (also sometimes called Bloch--Messiah decomposition) implies that any $m$-mode Gaussian unitary operation $\hat G$ can be decomposed as
\begin{equation}\label{eq:decompGmulti}
    \hat G=\hat U\hat S(\bm\xi)\hat D(\bm\alpha)\hat V,
\end{equation}
where $\hat U$ and $\hat V$ are passive linear operations, $\hat S(\bm\xi)=\bigotimes_{j=1}^m\hat S(\xi_j)$ is a tensor product of single-mode squeezing operators, and $\hat D(\bm\alpha)=\bigotimes_{j=1}^m\hat D(\alpha_j)$ is a tensor product of single-mode displacement operators. Equivalently, the action of these operations on the vector of creation and annihilation operators of the modes can be described by a single symplectic matrix~\cite{ferraro2005gaussian,weedbrook2012gaussian,serafini2017quantum}. In what follows however, we describe Gaussian unitaries using the explicit Euler decomposition above.

We also have the braiding relation
\begin{equation}\label{eq:braidDU}
    \hat U\hat D(\bm\alpha)=\hat D(\bm\beta)\hat U,
\end{equation}
where $\bm\beta=U^T\bm\alpha$, which implies with Eq.~(\ref{eq:braidDSR}) that the position of the displacement operators in the Euler decomposition of Eq.~(\ref{eq:decompGmulti}) is irrelevant by using the braiding relations (thus modifying the displacement amplitudes). Note also that a composition of $m$-mode passive linear operations $\hat U\hat V$ is described by the $m\times m$ unitary matrix $UV$.
Putting everything together, all Gaussian unitary operations may be decomposed efficiently as a combination of single-mode operators $\hat D$, $\hat S$ and $\hat R$, and two-mode operators $\hat H$.

Multimode Gaussian states are obtained from the multimode vacuum state $\ket{\bm0}$ by Gaussian unitaries. Since $\hat V\ket{\bm0}=\ket{\bm0}$ for all passive linear operations $\hat V$, any $m$-mode Gaussian state may be written as
\begin{equation}
    \hat U\hat S(\bm\xi)\hat D(\bm\beta)\ket{\bm0},
\end{equation}
for some $\bm\beta,\bm\xi\in\mathbb C^m$.

\subsection{Segal--Bargmann space and the single-mode stellar hierarchy}
\label{sec:SBstellar}

The Segal--Bargmann space~\cite{segal1963mathematical,bargmann1961hilbert} is defined as the separable infinite-dimensional Hilbert space of holomorphic functions $F^\star$ over $\mathbb C^m$, satisfying the normalization condition:
\begin{equation}\label{eq:SBnormSM}
\|F^\star\|^2:=\int_{\bm z\in\mathbb C^m}|F^\star(\bm z)|^2d\mu(\bm z)<+\infty,
\end{equation}
where $\mu(\bm z)$ is the Gaussian measure over $\mathbb C^m$. Hereafter, we use the terms holomorphic function to refer to an element of Segal--Bargmann space\footnote{Note that since the domain of holomorphic functions in Segal--Bargmann space is the full complex space, they are sometimes also referred to as entire functions.}. The inner product of two functions is given by
\begin{equation}\label{eq:SBscalarprod}
    \braket{F_1^\star|F_2^\star}=\int_{\bm z\in\mathbb C^m}F_1^\star(\bm z)^*F_2^\star(\bm z)d\mu(\bm z).
\end{equation}
We restrict to the single-mode case $m=1$ in this section and extend the results presented here to the multimode case in Sec.~\ref{sec:multi}.
Compared to the Fock space description used in the previous sections, quantum states are described by holomorphic functions rather than infinite countable vectors, through the correspondence~\cite{vourdas2006analytic}
\begin{equation}
    \ket n\leftrightarrow\left(z\mapsto\frac{z^n}{\sqrt{n!}}\right),
\end{equation}
for all $n\in\mathbb N$. In particular, a quantum state $\ket\psi=\sum_{n\ge0}\psi_n\ket n$ is mapped to
\begin{equation}
    \ket\psi\leftrightarrow F_\psi^\star(z):=\sum_{n\ge0}\frac{\psi_n}{\sqrt{n!}}z^n,
\end{equation}
the so-called stellar function of the state $\ket\psi$, which corresponds to an expansion in the overcomplete basis of so-called (Glauber) canonical coherent states
\begin{equation}
    \ket z=e^{-\frac12|z|^2}\sum_{n\ge0}\frac{z^n}{\sqrt{n!}}\ket n,
\end{equation}
for all $z\in\mathbb C$, which are eigenstates of the non-hermitian annihilation operator $\hat a$.

Operators that are functions of creation and annihilation operators acting on the Hilbert space are mapped to differential operators acting on the Segal--Bargmann space through the correspondence~\cite{vourdas2006analytic}
\begin{equation}\label{eq:corresp}
    \hat a^\dag\leftrightarrow z\times\quad\text{and}\quad\hat a\leftrightarrow\partial_z,
\end{equation}
where the first operation is the multiplication by the complex variable $z$. In particular, the differential operator corresponding to the displacement operator $\hat D(\alpha)$ is given by
\begin{equation}
    e^{\alpha z-\alpha^*\partial_z},
\end{equation}
for all $\alpha\in\mathbb C$. The differential operator corresponding to the squeezing operator $\hat S(\xi)$ is given by
\begin{equation}
 e^{\frac12(\xi z^2-\xi^*\partial_z^2)},
\end{equation}
for all $\xi\in\mathbb C$. The differential operator corresponding to the phase-shifting operator $\hat R(\varphi)$ is given by
\begin{equation}
 e^{i\varphi z\partial_z},
\end{equation}
for all $\varphi\in[0,2\pi]$. 

The stellar function is related to the Husimi $Q$ function by 
\begin{equation}
    Q_\psi(\alpha)=\frac{e^{-|\alpha|^2}}\pi\left|F_\psi^\star(\alpha^*)\right|^2,
\end{equation}
for all $\alpha\in\mathbb C$. As it turns out, a pure quantum state is non-Gaussian if and only if its $Q$ function has zeros~\cite{lutkenhaus1995nonclassical}, or equivalently if and only if its stellar function has zeros. Note that since this is a univariate holomorphic function, its set of zeros is discrete~\cite{conway2012functions}.

The so-called stellar hierarchy of quantum states has been defined in~\cite{chabaud2020stellar} based on the number of zeros of the stellar function, which we review hereafter. The name stellar originates from the fact that this representation is obtained as the limit case of the finite-dimensional Majorana representation~\cite{majorana1932atomi}, where zeros in the complex plane are mapped to points on the surface of a sphere by anti-stereographic projection. These points on the sphere are then thought of as stars on the celestial vault, for an observer located inside the sphere (see Sec.~\ref{sec:analytic}). 

Using the Hadamard--Weierstrass factorization theorem~\cite{conway2012functions}, the stellar function can be written as
\begin{equation}
    F^\star(z)=e^{-\frac12 az^2+bz+c}z^k\prod_n\left(1-\frac z{\lambda_n}\right)e^{\frac z{\lambda_n}+\frac12\frac{z^2}{\lambda_n^2}},
\end{equation}
where $k\in\mathbb N$ is the multiplicity of $0$ as a root of $F^\star$, where $\{\lambda_n\}_n$ is the (possibly infinite) discrete set of non-zero roots of $F^\star$, and where $a,b,c\in\mathbb C$.
In particular, the set of stellar functions with $n$ zeros corresponds to the set of holomorphic functions in the single-mode Segal--Bargmann space of the form $P\times G$, where $P$ is a polynomial of degree $n\in\mathbb N$ and $G$ is a Gaussian function. The stellar rank is then defined as the number of zeros, or equivalently the degree of the polynomial $P$. States having a stellar function which is not of the form $P\times G$ have infinite rank (infinite number of zeros), and these states can be approximated arbitrarily well in trace distance by states of finite rank.

Importantly, the stellar rank is invariant under Gaussian unitary operations and thus the stellar hierarchy is a non-Gaussian hierarchy. Pure states of rank $0$ are Gaussian states, of the form $\hat S(\xi)\hat D(\beta)\ket0$, for $\xi,\beta\in\mathbb C$. Their stellar function is given by
\begin{equation}\label{eq:stellarG}
    G_{\xi,\beta}^\star(z)=(1-|a|^2)^{1/4}e^{-\frac12az^2+bz+c},
\end{equation}
for all $z\in\mathbb C$, where
\begin{equation}
    a:=-e^{i\theta}\tanh r,\quad b:=\beta\sqrt{1-|a|^2},\quad 
    c:=\frac12a^*\beta^2-\frac12|\beta|^2,
\end{equation}
and where we have set $\xi:=re^{i\theta}$. All other ranks are populated by non-Gaussian states. Moreover, states of finite stellar rank $n$ can be obtained from the vacuum by exactly $n$ applications of the creation operator $\hat a^\dag$, together with Gaussian unitary operations. Since their stellar function are of the form $P\times G$, where $P$ is a polynomial of degree $n\in\mathbb N$ and $G$ is a Gaussian function, they can be described by the zeros of $P$ and the exponents of $G$, which is called the stellar representation of the state. Pure quantum states with a polynomial stellar function are called core states. These correspond to states with a bounded support over the Fock basis.

The stellar rank can be extended to mixed states using a convex roof construction, i.e., the stellar rank of a mixed state is given by the maximum rank of the pure states in its convex decomposition, minimized over all possible convex decompositions. 
While this definition makes the stellar rank of mixed states challenging to compute, this problem can be mitigated by using the robustness of the stellar hierarchy of pure states: if any mixed state is close enough in trace distance to a pure state of a given stellar rank, then a tight lower bound on the stellar rank of the mixed state can be derived efficiently.
We refer the reader to~\cite{chabaud2020stellar,chabaud2021certification,PRXQuantum.2.030204} for a detailed exposition of the stellar hierarchy of pure and mixed single-mode quantum states and its experimental demonstration.

\subsection{The Calogero--Moser model}
\label{sec:CM}

The Calogero--Moser model~\cite{calogero1971solution,moser1976three} is a well-known example of integrable many-body problem which has arisen in many different settings and relates in particular to the evolution of poles of solutions of the Korteweg--de Vries equation~\cite{korteweg1895xli,airault1977rational} and the Kadomtsev--Petviashvili equation~\cite{kadomtsev1970stability,krichever1978rational} modelling fluid dynamics. Both the classical and quantum versions of this model are integrable, and we will focus here on the classical one. We refer the reader to~\cite{olshanetsky1981classical,calogero2008calogero,polychronakos2006physics,etingof2006lectures} for comprehensive reviews of this model.

The $n$-body classical Calogero--Moser Hamiltonian is given by: 
\begin{equation}\label{eq:HCM}
    H^{CM}=\frac12\sum_{k=1}^n\left(p_k^2+\omega^2q_k^2\right)+\frac12g^2\sum_{k=1}^n\sum_{j\neq k}\frac1{(q_k-q_j)^2},
\end{equation}
where $\omega,g\in\mathbb R$ and where $q_k,p_k$ are the canonical variables (position and momentum of the $k^{th}$ particle). The equations of motion are given by Hamilton's equations $\frac{dq_k}{dt}=p_k$ and $\frac{dp_k}{dt}=-\frac{\partial H}{\partial q_k}$, for all $k\in\{1,\dots,n\}$:
\begin{equation}
    \frac{d^2q_k}{dt^2}+\omega^2q_k=2g^2\sum_{k=1}^n\sum_{j\neq k}\frac1{(q_k-q_j)^3}.
\end{equation}
These correspond to the scattering process of $n$ particles of unit mass on a line interacting with an inverse square potential, which we will refer to as Calogero--Moser particles. This description may be extended to the complex plane by allowing $g\in\mathbb C$ and $q_k,p_k\in\mathbb C$ for all $k\in\{1,\dots,n\}$.

Setting $\omega=0$ gives the isolated Calogero--Moser system, or simply Calogero--Moser system for short, while $\omega\in\mathbb R\setminus\{0\}$ gives the harmonic Calogero--Moser system. We also consider $\omega\in i\mathbb R\setminus\{0\}$, which we call the hyperbolic Calogero--Moser system.

These integrable models may be solved using, e.g., the Olshanetsky--Perelomov projection method~\cite{ol1976geodesic}, which consists in recovering the Calogero--Moser equations as the projection of simpler higher-dimensional equations. Then, the canonical variables $q_k$ are obtained as eigenvalues of an analytical matrix solution of these higher-dimensional equations of motion (see Appendix~\ref{app:CM} for the case $\omega=0$ and~\cite{olshanetsky1981classical} for the general case).

In the case $\omega\neq0$ (system V in~\cite{olshanetsky1981classical}), the solutions $q_k(t)$ are given by the eigenvalues of a $n\times n$ matrix $\Lambda(t)$ with time-dependent coefficients, defined as:
\begin{equation}
        \Lambda_{kl}(t):=\begin{cases}q_k(0)\cos(\omega t)+p_k(0)\frac{\sin(\omega t)}{\omega},\quad &k=l\\\frac{ig\sin(\omega t)}{\omega(\lambda_l(0)-\lambda_k(0))},\quad&k\neq l.\end{cases}
\end{equation}
The analysis extends to the case where $\omega\in i\mathbb R\setminus\{0\}$, in which case the trigonometric functions may be replaced by hyperbolic ones. 

When $\omega=0$ (system I in~\cite{olshanetsky1981classical}), the solutions $q_k(t)$ are given instead by the eigenvalues of the matrix which is the limit of the previous one when $\omega\rightarrow0$:
\begin{equation}
        \Lambda_{kl}(t):=\begin{cases}q_k(0)+p_k(0)t,\quad &k=l\\\frac{igt}{\lambda_l(0)-\lambda_k(0)},\quad&k\neq l.\end{cases}
\end{equation}
Moreover, in this case the whole scattering process may be summarized by a permutation of the trajectories: since the potential vanishes at large distances, writing $q_k(t)\sim p_k^\pm t+q_k^\pm$ when $t\rightarrow\pm\infty$, there exists a permutation matrix $P_\sigma$ such that $(q_1^+,\dots,q_n^+)=P_\sigma(q_1^-,\dots,q_n^-)$ and $(p_1^+,\dots,p_n^+)=P_\sigma(p_1^-,\dots,p_n^-)$.

%------------------------------------------------------------------------

\section{Single-mode rank-preserving evolutions and the Calogero--Moser model}
\label{sec:SMdynamics}

In this section, we study how evolutions preserving the stellar rank, i.e., Gaussian unitary evolutions, affect single-mode pure quantum states of finite stellar rank. For a discussion of the stellar hierarchy for pure and mixed single-mode quantum states, we refer the reader to~\cite{chabaud2020stellar}.

As discussed in the previous section, pure quantum states of finite stellar rank form a dense subset of the Hilbert space, and can be described through their stellar representation, which consists in the list of the complex zeros of their stellar function, together with its Gaussian exponents: a pure state of finite stellar rank $n$ has a stellar function $F^\star$ of the form:
\begin{equation}
    F^\star(z)=P(z)G(z),
\end{equation}
for all $z\in\mathbb C$, where $P$ is a polynomial of degree $n$ and $G$ is a Gaussian function. We can rewrite this function as
\begin{equation}
    F^\star(z)=\prod_{k=1}^n(z-\lambda_k)e^{-\frac12az^2+bz+c},
\end{equation}
where $\{\lambda_k\}_{k=1,...,n}\in\mathbb C^n$ are the zeros of $F^\star$ repeated with multiplicity and where $a,b,c\in\mathbb C$ are the Gaussian parameters.

Although the evolution of these parameters may seem intricate at first, we show in the next section that the evolution of the zeros on one hand and the Gaussian parameters on the other hand can actually be decoupled (Theorem~\ref{th:decoupling}). This allows us to interpret the motion of the zeros as an integrable dynamical system of classical particles, together with a conformal evolution of the Gaussian parameters. Interestingly, we show that this integrable dynamical system coincides with the Calogero--Moser model presented in Sec.~\ref{sec:CM}.
For practical purposes, we also derive analytical expressions for the stellar function evolving under Gaussian unitaries in Sec.~\ref{sec:direct}.

\subsection{Stellar dynamics}

The evolution of a single-mode quantum state $\ket\psi$ under a Hamiltonian $\hat H=H(\hat a^\dag,\hat a)$ which is a function of the creation and annihilation operators is described by Schr\"odinger's equation~\cite{schrodinger1926undulatory}:
\begin{equation}
    i\hbar\frac{d\ket{\psi(t)}}{dt}=H(\hat a^\dag,\hat a)\ket{\psi(t)},
\end{equation}
for all $t\ge0$. We set $\hbar=1$ in what follows.
A similar equation holds in Segal--Bargmann space for the stellar function~\cite{leboeuf1991phase}, which yields an equivalent description for the evolution of the quantum state:
\begin{equation}\label{schrostellar}
    i\partial_tF^\star(z,t)=H(z,\partial_z)F^\star(z,t),
\end{equation}
for all $z\in\mathbb C$ and all $t\ge0$.

Gaussian Hamiltonians are defined as polynomials in the creation and annihilation operators of the modes of degree less or equal to $2$, and induce Gaussian unitary evolutions. These are the evolutions that leave invariant the set $R_n$ of states of stellar rank $n$~\cite{chabaud2020stellar}, so if $z\mapsto F^\star(z,0)\in R_n$, then under a Gaussian evolution we also have $z\mapsto F^\star(z,t)\in R_n$ for all $t>0$. In particular, for a pure quantum state of finite stellar rank $n$ evolving under a Gaussian Hamiltonian, there exist a polynomial $P$ of degree $n$ and a Gaussian function $G$ with time-dependent coefficients such that for all $t\ge0$,
\begin{equation}\label{eq:FPxG}
    F^\star(z,t)=P(z,t)G(z,t),
\end{equation}
for all $z\in\mathbb C$. We can rewrite this function as
\begin{equation}
    F^\star(z,t)=\prod_{k=1}^n(z-\lambda_k(t))e^{-\frac12a(t)z^2+b(t)z+c(t)},
\end{equation}
where $\{\lambda_k(t)\}_{k=1,...,n}\in\mathbb C^n$ are the zeros of $z\mapsto F^\star(z,t)$ repeated with multiplicity, and where $a(t),b(t),c(t)\in\mathbb C$ are the Gaussian parameters, with $|a(t)|<1$ (this condition ensures the boundedness of the norm in Eq.~(\ref{eq:SBnormSM})). 

A stellar function with $n$ zeros is thus described by $n+3$ complex parameters: its $n$ complex zeros and three Gaussian exponents. The parameter $a$ can be thought of as a squeezing parameter, the parameter $b$ as a displacement parameter, and the parameter $c$ as a global phase and normalization factor.

The complex zeros of the stellar function may then be interpreted as particles in phase space. Hence, an appealing way of studying the Gaussian dynamics of single-mode quantum states is to consider the motion of these particles in the complex plane~\cite{leboeuf1991phase}. However, it only gives a faithful description of the whole dynamics if the evolution preserves the number of zeros---in this case, when the evolution is Gaussian---and if the evolution of the Gaussian parameters is also taken into account. Previous attempts to study quantum dynamics based on the zeros of the stellar function were not considering the evolution of the Gaussian parameters~\cite{ellinas1995motion}, thus yielding an incomplete description of the evolution.

The displacement, squeezing, phase-shifting and shearing Hamiltonians are defined as:
\begin{equation}
    \begin{aligned}
        \hat H_\alpha^D&:=i\alpha\hat a^\dag-i\alpha^*\hat a\\
        \hat H_\xi^S&:=\frac i2\xi\hat a^{\dag2}-\frac i2\xi^*\hat a^2\\
        \hat H_\varphi^R&:=-\varphi\hat a^\dag\hat a\\
        \hat H_s^P&:=-s\hat q^2=-\frac s2\hat a^{\dag2}-\frac s2\hat a^2-s\hat a^\dag\hat a-\frac s2\hat 1,
    \end{aligned}
\end{equation}
where $\alpha,\xi\in\mathbb C$ and $\varphi,s\in\mathbb R$, inducing the unitary evolutions $\hat D(\alpha t)$, $\hat S(\xi t)$, $\hat R(\varphi t)$, and $\hat P(st)$, via $\hat U=e^{-i\hat Ht}$.
Note that the shearing Hamiltonian satisfies $\hat H_s^P=\hat H_{is}^S+\hat H_s^R-\frac s2\hat 1$. A general single-mode Gaussian Hamiltonian is then given by:
\begin{equation}
    H_{\alpha,\xi,\varphi,\gamma}(\hat a^\dag,\hat a)=\hat H_\alpha^D+\hat H_\xi^S+\hat H_\varphi^R+\gamma\hat 1,
\end{equation}
where $\alpha,\xi\in\mathbb C$ and $\varphi,\gamma\in\mathbb R$. The Hamiltonian $\gamma\hat 1$ only leads to a global phase, so we set $\gamma=0$ in what follows and write $H_{\alpha,\xi,\varphi}(\hat a^\dag,\hat a)$. For a generic Gaussian Hamiltonian $\hat H=H_{\alpha,\xi,\varphi}(\hat a^\dag,\hat a)$ we have
\begin{equation}\label{iHGstellar}
    -iH_{\alpha,\xi,\varphi}(z,\partial_z)=\frac 12\xi z^2-\frac 12\xi^*\partial^2_z+i\varphi z\partial_z+\alpha z-\alpha^*\partial_z.
\end{equation}
With Eq.~(\ref{schrostellar}), we derive the dynamical system governing the evolution of the stellar representation under a generic Gaussian Hamiltonian:

\begin{theorem}\label{th:decoupling}
Let $F^\star(z,t)=\prod_{k=1}^n(z-\lambda_k(t))e^{-\frac12a(t)z^2+b(t)z+c(t)}$. The evolution of the stellar function under a Gaussian Hamiltonian $\hat H=H_{\alpha,\xi,\varphi}(\hat a^\dag,\hat a)$ may be recast as the following dynamical system:
\begin{equation}\label{eq:dynasyst}
    \begin{cases}
        \frac{da(t)}{dt}=\xi^*a^2(t)+2i\varphi a(t)-\xi,\\
        \frac{db(t)}{dt}=(i\varphi+\xi^*a(t))b(t)+\alpha+\alpha^*a(t),\\
        \frac{dc(t)}{dt}=\frac12\xi^*a(t)-\frac12\xi b^2(t)-\alpha^*b(t)+n(\xi^*a(t)+i\varphi)\\
        \frac{d^2\lambda_k(t)}{dt^2}=(|\xi|^2-\varphi^2)\lambda_k(t)+(\xi^*\alpha-i\varphi\alpha^*)-2\xi^{*2}\sum_{j\neq k}\frac1{(\lambda_k(t)-\lambda_j(t))^3},\quad\forall k\in\{1,\dots,n\}.
    \end{cases}
\end{equation}
The equations for the Gaussian parameters and the zeros are decoupled, and the relation between the root system $\{\lambda_k\}_{k=1\dots n}$ and the Gaussian parameters $a,b$ is given by the initial conditions:
\begin{equation}\label{eq:CIth}
    \frac{d\lambda_k(t)}{dt}\bigg\vert_{t=0}=-\left(\xi^*a(0)+i\varphi\right)\lambda_k(0)+\xi^* b(0)+\alpha^*+\xi^*\sum_{j\neq k}\frac1{\lambda_k(0)-\lambda_j(0)},\quad\forall k\in\{1,\dots,n\}.
\end{equation}
\end{theorem}

\noindent We give a proof in Appendix~\ref{app:decoupling}.
Note that the equations become singular whenever the particles collide. However, the set of initial conditions such that this happens is of measure $0$, and this does not prevent a full resolution by considering multiple zeros as the limit case of single zeros. 

Remarkably, the second order differential equations decouple the evolution of the zeros from the evolution of the Gaussian parameters. In this picture, the role played by the Gaussian parameters is only to constrain the initial velocities of the particles through Eq.~(\ref{eq:CIth}). Hence, we may describe the evolution of the zeros of the stellar function by $n$ classical particles of unit mass moving in the complex plane and interacting pairwise with the classical Hamiltonian:
\begin{equation}\label{eq:Hstellar}
    H=\frac12\sum_{k=1}^n\left(p_k^2+\omega^2q_k^2+2\delta q_k\right)+\frac12g^2\sum_{k=1}^n\sum_{j\neq k}\frac1{(q_k-q_j)^2},
\end{equation}
with $q_k=\lambda_k$, $p_k=\frac{d\lambda_k}{dt}$, $\omega:=\sqrt{\varphi^2-|\xi|^2}\in\mathbb R\cup i\mathbb R$, $\delta:=i\varphi\alpha^*-\xi^*\alpha\in\mathbb C$, and $g:=i\xi^*\in\mathbb C$. We recover the equations of motions through Hamilton's equations $\frac{dq_k}{dt}=p_k$ and $\frac{dp_k}{dt}=-\frac{dH}{dt}$ for all $k\in\{1,\dots,n\}$.
The Hamiltonian in Eq.~(\ref{eq:Hstellar}) is reminiscent of the classical Calogero--Moser Hamiltonian reviewed in Sec.~\ref{sec:CM}: 
\begin{equation}\label{eq:HCM2}
    H^{CM}=\frac12\sum_{k=1}^n\left(p_k^2+\omega^2q_k^2\right)+\frac12g^2\sum_{k=1}^n\sum_{j\neq k}\frac1{(q_k-q_j)^2},
\end{equation}
with complex parameters. 
The equations of motions for the zeros of the stellar function in Theorem~\ref{th:decoupling} may thus be understood as follows:
\begin{equation}\label{eq:interpretation}
    \frac{d^2\lambda_k(t)}{dt^2}=\underset{\text{symplectic term}}{\underbrace{(|\xi|^2-\varphi^2)\lambda_k(t)}}+\underset{\text{displacement term}}{\underbrace{(\xi^*\alpha-i\varphi\alpha^*)}}-2\xi^{*2}\underset{\text{Calogero--Moser interaction}}{\underbrace{\sum_{j\neq k}\frac1{(\lambda_k(t)-\lambda_j(t))^3}}},\quad\forall k\in\{1,\dots,n\}.
\end{equation}
Three cases arise depending on the value of the phase-shift parameter $\varphi$:

\begin{itemize}

\item $\varphi=|\xi|$: parabolic case (shearing); up to a global accelerated translation $\frac12\delta t^2=\frac12(i\varphi\alpha^*-\xi^*\alpha)t^2$ of the zeros, we obtain an isolated Calogero--Moser system.
\item $\varphi>|\xi|$: elliptic case (phase shift); up to a global offset $\frac\delta{\omega^2}=\frac{i\varphi\alpha^*-\xi^*\alpha}{\varphi^2-|\xi|^2}$ of the zeros, we obtain a harmonic Calogero--Moser system.
\item $\varphi<|\xi|$: hyperbolic case (squeezing); up to a global offset $\frac\delta{\omega^2}=\frac{i\varphi\alpha^*-\xi^*\alpha}{\varphi^2-|\xi|^2}$  of the zeros, we obtain a hyperbolic Calogero--Moser system.

\end{itemize}

\noindent In all three cases, up to the global translation or offset accounting for the displacement term, the zeros are given by the solutions of the corresponding Calogero--Moser model. Recall that Gaussian unitary operations act on the vector of quadrature operators through their symplectic representation~\cite{ferraro2005gaussian}. The elements of the single-mode symplectic Lie algebra are classified as parabolic/elliptic/hyperbolic according to their trace~\cite{shackerley2017reachable}. This categorization has a geometrical interpretation where parabolics are related to shears (shearing), elliptics to rotations (phase shift), and hyperbolics to squeezes (squeezing), and the symplectic term in Eq.~(\ref{eq:interpretation}) echoes this classification.

The representation of Gaussian evolutions based on the motion of zeros of the stellar function can be used to visualize Gaussian dynamics of non-Gaussian states using worldline diagrams (see Figures~\ref{fig:CMdisplacement},~\ref{fig:CMphaseshift}, and~\ref{fig:CMshear}). Up to the conformal evolution of the Gaussian parameters of the stellar function, these diagrams provide a complete intuitive description of the quantum Gaussian dynamics.

The system in Eq.~(\ref{eq:dynasyst}) describes the stellar dynamics under a generic Gaussian Hamiltonian. Since any single-mode Gaussian unitary may be written as a composition $\hat D(\alpha)\hat R(\varphi)\hat P(s)\hat R(\chi)$ of displacement, phase shift, and shear, for some $\alpha\in\mathbb C$ and $\varphi,\chi,s\in\mathbb R$, the evolved stellar function at time $t$ under a generic Gaussian Hamiltonian may be obtained by a sequence of evolutions corresponding to these three operators, respectively. 

We solve the dynamical systems for these three cases in Appendix~\ref{app:dynamicsDRPS} and we summarize hereafter the results, starting with displacements (for completeness, we also solve the dynamical system for squeezing in Appendix~\ref{app:dynamicsDRPS}):

\begin{lemma}\label{lem:evoD}
Let $F^\star(z,t)=\prod_{k=1}^n(z-\lambda_k(t))e^{-\frac12a(t)z^2+b(t)z+c(t)}$ and let $\alpha\in\mathbb C$. The evolution under the displacement Hamiltonian $\hat H_\alpha^D$ is given by
\begin{equation}
    \begin{cases}
        a(t)=a(0),\\
        b(t)=(\alpha+\alpha^*a(0))t+b(0),\\
        c(t)=\frac12(\alpha^{*2}a(0)-|\alpha|^2)t^2+b(0)t+c(0),
    \end{cases}
\end{equation}
and the motion of the zeros is given by
\begin{equation}
    \lambda_k(t)=\alpha^*t+\lambda(0),\quad\forall k\in\{1,\dots,n\}.
\end{equation}
\end{lemma}

\begin{figure}[ht!]
    \centering
        \includegraphics[width=0.5\columnwidth]{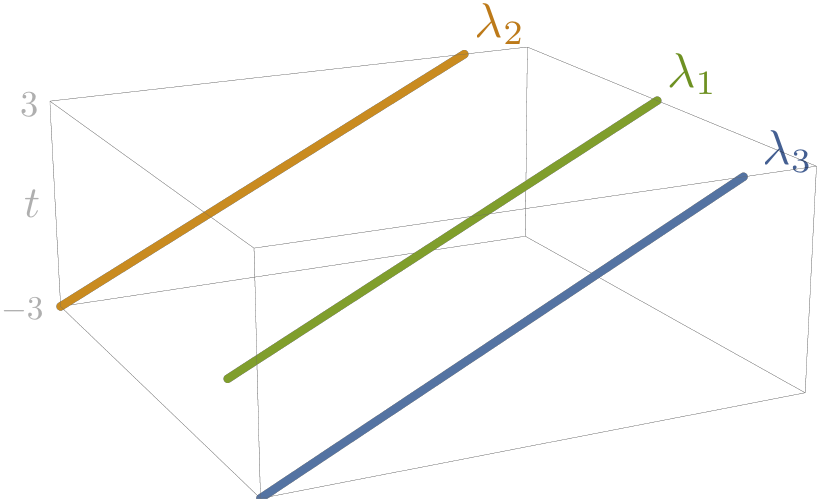}
        \caption{Motion of zeros under a displacement Hamiltonian $\hat H_1^D=i\hat a^\dag-i\hat a$, from $t=-T$ to $t=T$ in the complex plane. The trajectories are obtained with Lemma~\ref{lem:evoD}, with the values $\lambda_1(0)=1$, $\lambda_2(0)=i$, $\lambda_3(0)=-i$, and $T=3$.}
        \label{fig:CMdisplacement}
\end{figure}

\noindent In this case, the zeros correspond to non-interacting particles which follow a global uniform translation: $\frac{d^2\lambda_k(t)}{dt^2}=0$, for all $k\in\{1,\dots,n\}$ (see Figure~\ref{fig:CMdisplacement}).
For phase shifts we have:

\begin{lemma}\label{lem:evoR}
Let $F^\star(z,t)=\prod_{k=1}^n(z-\lambda_k(t))e^{-\frac12a(t)z^2+b(t)z+c(t)}$ and let $\varphi\in[0,2\pi]$. The evolution under the phase-shift Hamiltonian $\hat H_\varphi^R$ is given by
\begin{equation}
    \begin{cases}
        a(t)=e^{2i\varphi t}a(0),\\
        b(t)=e^{i\varphi t}b(0),\\
        c(t)=e^{in\varphi t}c(0),
    \end{cases}
\end{equation}
and the motion of the zeros is given by
\begin{equation}
    \lambda_k(t)=e^{-i\varphi t}\lambda(0),\quad\forall k\in\{1,\dots,n\}.
\end{equation}
\end{lemma}

\begin{figure}[ht!]
    \centering
        \includegraphics[width=0.5\columnwidth]{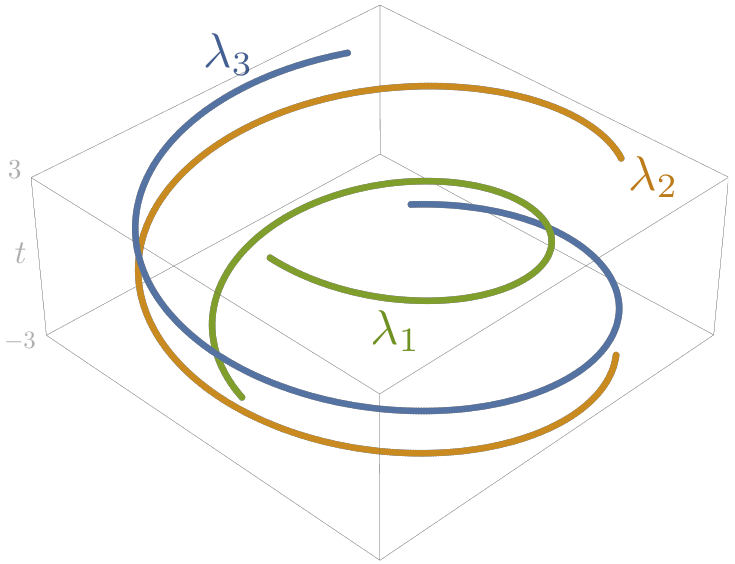}
        \caption{Motion of zeros under a phase-shifting Hamiltonian $\hat H_1^R=-\hat a^\dag\hat a$, from $t=-T$ to $t=T$ in the complex plane. The trajectories are obtained with Lemma~\ref{lem:evoR}, with the values $\lambda_1(0)=-1$, $\lambda_2(0)=1-i$, $\lambda_3(0)=1+i$, and $T=3$.}
        \label{fig:CMphaseshift}
\end{figure}

\noindent In this case, the zeros correspond to non-interacting particles which follow a global uniform rotation around the origin: $\frac{d^2\lambda_k(t)}{dt^2}+\varphi^2\lambda_k(t)=0$, for all $k\in\{1,\dots,n\}$ (see Figure~\ref{fig:CMphaseshift}). Finally, we turn to the case of shearing:

\begin{lemma}\label{lem:evoP}
Let $F^\star(z,t)=\prod_{k=1}^n(z-\lambda_k(t))e^{-\frac12a(t)z^2+b(t)z+c(t)}$ and let $s\in\mathbb R$. The evolution under the shearing Hamiltonian $\hat H_s^P$ is given by
\begin{equation}
    \begin{cases}
        a(t)=\frac{a(0)-ist(1-a(0))}{1-ist(1-a(0))},\\
        b(t)=\frac{b(0)}{1-ist(1-a(0))},\\
        c(t)=c(0)-\frac{ist}2-\left(n+\frac12\right)\log(1-ist(1-a(0)))-\frac{b^2(0)}{2(1-a(0))(1-ist(1-a(0))},
    \end{cases}
\end{equation}
and the zeros $\lambda_k(t)$ are the eigenvalues of the matrix $\Lambda(t)$ defined as
\begin{equation}\label{eq:Lambdakl}
\Lambda_{kl}(t):=\begin{cases}\lambda_k(0)-ist\left[(1-a(0))\lambda_k(0)+b(0)+\sum_{j\neq k}\frac1{\lambda_k(0)-\lambda_j(0)}\right],\quad&k=l\\\frac{ist}{\lambda_l(0)-\lambda_k(0)},\quad&k\neq l.\end{cases}
\end{equation}
\end{lemma}

\begin{figure}[t!]
    \centering
        \includegraphics[width=\columnwidth]{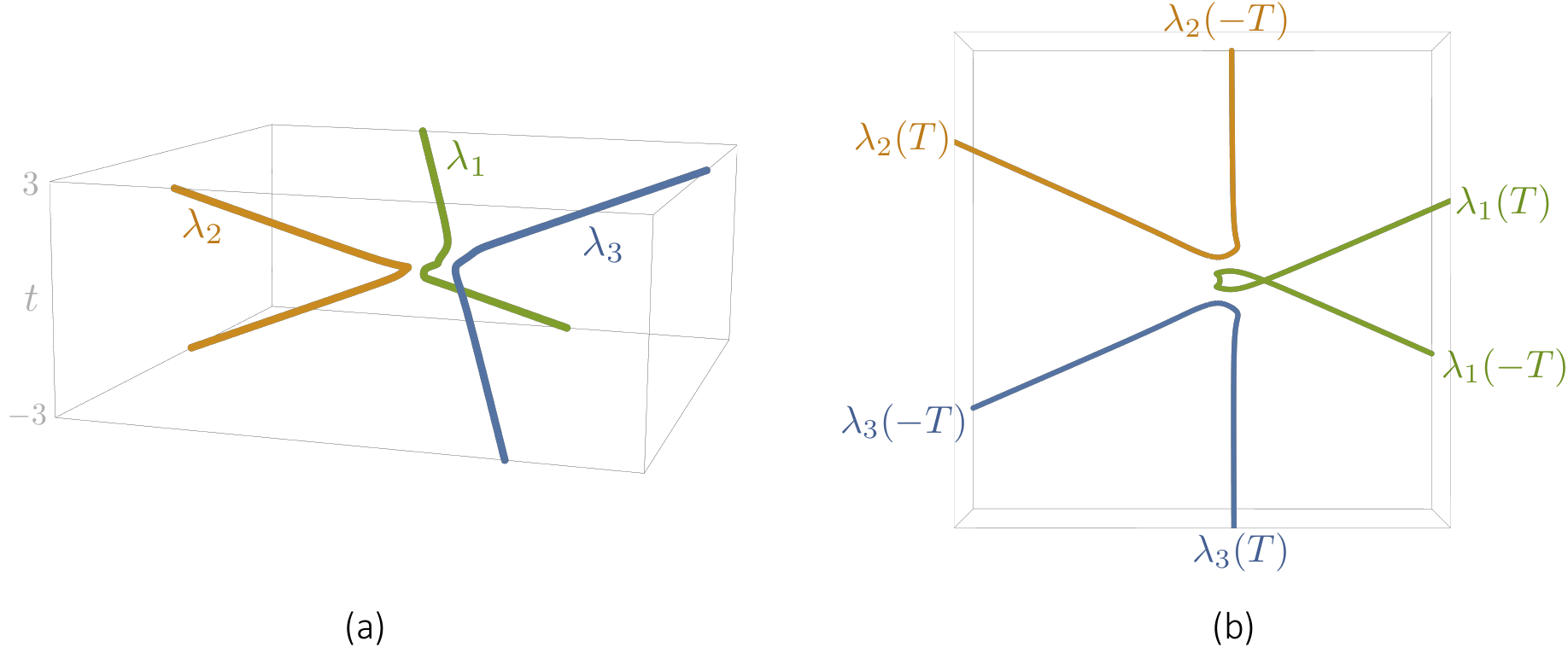}
        \caption{(a) Calogero--Moser motion of zeros under a shearing Hamiltonian $\hat H_1^P=-\hat q^2$, from $t=-T$ to $t=T$ in the complex plane. The trajectories are obtained with Lemma~\ref{lem:evoP}, with the values $\lambda_1(0)=0$, $\lambda_2(0)=\frac i2$, $\lambda_3(0)=-\frac i2$, $a(0)=b(0)=0$, and $T=3$. (b) Top view of the same evolution. Note the cyclic permutation of the asymptotic trajectories between the input and the output.}
        \label{fig:CMshear}
\end{figure}

\noindent In this case, the evolution of the zeros corresponds to that of $n$ classical Calogero--Moser particles in the complex plane:
\begin{equation}
    \frac{d^2\lambda_k(t)}{dt^2}=2s^2\sum_{j\neq k}\frac1{(\lambda_k(t)-\lambda_j(t))^3},\quad\forall k\in\{1,\dots,n\},
\end{equation}
where $s\in\mathbb R$, constrained by the initial conditions
\begin{equation}
        \frac{d\lambda_k(t)}{dt}\bigg\vert_{t=0}=is(a(0)-1)\lambda_k(0)-is b(0)-is\sum_{j\neq k}\frac1{\lambda_k(0)-\lambda_j(0)},\quad\forall k\in\{1,\dots,n\}.
\end{equation}
With Eq.~(\ref{schrostellar}), we thus find that the Calogero--Moser system describes the motion of the zeros of the complex holomorphic solutions $u$ to the partial differential equation
\begin{equation}
    i\partial_tu+(z+\partial_z)^2u=0,
\end{equation}
capturing the qualitative action of Gaussian quadratic non-linear evolutions such as shearing (or squeezing, up to an additional phase-shifting evolution).

The short-time dynamics of the zeros can be quite elaborate if they are close to each other, but the long-time dynamics only amounts to a free motion together with a permutation of the trajectories of the particles. Indeed, we show in Appendix~\ref{app:CM} that the scattering parameters before and after the interaction are conserved, and the particles resume their free motion regardless of the interaction (see Fig.~\ref{fig:CMshear}).

Lemma~\ref{lem:evoP} allows us to express the stellar function evolving under a shearing Hamiltonian $\hat H_s^P$ in a compact form as
\begin{equation}\label{eq:beautiful}
    F^\star(z,t)=\Det\left(z\mathbb I-\Lambda(t)\right)G(z,t),
\end{equation}
where $G(z,t)=e^{-\frac12a(t)z^2+b(t)z+c(t)}$, with $a(t)=\frac{a(0)-ist(1-a(0))}{1-ist(1-a(0))}$, $b(t)=\frac{b(0)}{1-is(1-a(0))t}$ and where the parameter $c(t)$ corresponds to a global phase and a normalizing factor. Note the conformal evolution of the squeezing and displacement parameters $a$ and $b$, while the zeros evolve as $n$ classical Calogero--Moser particles in the complex plane. The expressions in Eqs.~(\ref{eq:Lambdakl}) and~(\ref{eq:beautiful}) allow us to compute the evolved stellar function under a shearing from the initial stellar representation $\lambda_1(0),\dots,\lambda_n(0),a(0),b(0),c(0)$.

Note also that if the coefficients of the polynomial part of the stellar function are known rather than its zeros, then Eq.~(\ref{eq:beautiful}) still allows for the computation of the evolved stellar function: expanding the determinant---using, e.g., Faddeev--Le Verrier algorithm~\cite{le1845memoire,faddeev1972problems}---yields an expression with symmetric functions of the initial zeros, which may in turn be replaced by initial coefficients of the polynomial part of the stellar function, thanks to Newton's equations.

Summarizing the results obtained in this section, any Gaussian evolution of the stellar representation of a single-mode quantum state of finite stellar rank may be obtained by a combination of displacement, rotations and Calogero--Moser motion of the zeros in the complex plane, together with a conformal evolution of the Gaussian squeezing and displacement exponents.

\subsection{Direct Gaussian evolution of the stellar function}
\label{sec:direct}

The previous section gives a qualitative description of the Gaussian evolution of single-mode pure quantum states of finite stellar rank. For practical purposes, it may be useful to have developed expressions of the stellar function under Gaussian transformations, which is what we derive in this section.

\begin{lemma}\label{lem:directevoD}
The displacement operator $\hat D(\alpha)$ acts on the stellar function $F^\star$ as
\be
F^\star(z)\mapsto e^{\alpha z-\frac12|\alpha|^2}F^\star(z-\alpha^*),
\ee
for all $\alpha\in\mathbb C$.
\end{lemma}

\begin{proof}

The stellar function of a pure quantum state $\ket\psi$ satisfies:
\begin{equation}
    \ket\psi=F^\star_\psi(\hat a^\dag)\ket0.
\end{equation}
Hence, for all $\ket\psi=\sum_{n\ge0}{\psi_n\ket n}$ we have
\begin{equation}
\begin{aligned}
\hat D(\alpha)\ket\psi&=\hat D(\alpha)F_\psi^\star(\hat a^\dag)\ket0\\
&=\sum_{n\ge0}{\frac{\psi_n}{\sqrt{n!}}\hat D(\alpha)\,(\hat a^\dag)^n\ket0}\\
&=\sum_{n\ge0}{\frac{\psi_n}{\sqrt{n!}}(\hat a^\dag-\alpha^*)^n\hat D(\alpha)\ket0}\\
&=F_\psi^\star(\hat a^\dag-\alpha^*)\hat D(\alpha)\ket0\\
&=F_\psi^\star(\hat a^\dag-\alpha^*)e^{\alpha\hat a^\dag-\frac12|\alpha|^2}\ket0,
\end{aligned}
\end{equation}
where we used Eq.~(\ref{eq:commutDSR}) in the third line and Eq.~(\ref{eq:stellarG}) in the last line. Hence, the displacement operator $\hat D(\alpha)$ acts on the stellar function as
\be
F_\psi^\star(z)\mapsto e^{\alpha z-\frac12|\alpha|^2}F_\psi^\star(z-\alpha^*),
\label{dispstellar}
\ee
for all $\alpha\in\mathbb C$. 

\end{proof}

\begin{lemma}\label{lem:directevoR}
The phase-shifting operator $\hat R(\varphi)$ acts on the stellar function as
\be
F_\psi^\star(z)\mapsto F_\psi^\star(e^{i\varphi}z),
\ee
for all $\varphi\in[0,2\pi]$. 
\end{lemma}

\begin{proof}

For all $\ket\psi$ we have
\begin{equation}
    \begin{aligned}
        \hat R(\varphi)\ket\psi&=\hat R(\varphi)F_\psi^\star(\hat a^\dag)\ket0\\
        &=F_\psi^\star(e^{i\varphi}\hat a^\dag)\hat R(\varphi)\ket0\\
        &=F_\psi^\star(e^{i\varphi}\hat a^\dag)\ket0,
    \end{aligned}
\end{equation}
where $\varphi\in[0,2\pi]$, and where we used Eq.~(\ref{eq:commutDSR}) in the second line. Hence, the phase-shifting operator $\hat R(\varphi)$ acts on the stellar function as
\be\label{eq:actionR}
F_\psi^\star(z)\mapsto F_\psi^\star(e^{i\varphi}z),
\ee
for all $\varphi\in[0,2\pi]$. 

\end{proof}

\begin{lemma}\label{lem:directevoS}
The squeezing operator $\hat S(\xi)$ acts on the stellar function as
\be
F_\psi^\star(z)\mapsto\frac1{\sqrt{\cosh r}}F_\psi^\star(\cosh rz-\sinh re^{-i\theta}\partial_z)e^{\frac12e^{i\theta}(\tanh r)z^2},
\ee
for all $\xi=re^{i\theta}\in\mathbb C$.
\end{lemma}

\begin{proof}

Similarly, for all $\ket\psi$ we have
\begin{equation}
    \begin{aligned}
        \hat S(\xi)\ket\psi&=\hat S(\xi)F_\psi^\star(\hat a^\dag)\ket0\\
        &=F_\psi^\star(\cosh r\hat a^\dag-\sinh re^{-i\theta}\hat a)\hat S(\xi)\ket0\\
        &=\frac1{\sqrt{\cosh r}}F_\psi^\star(\cosh r\hat a^\dag-\sinh re^{-i\theta}\hat a)e^{\frac12e^{i\theta}(\tanh r)\hat a^{\dag2}}\ket0,
    \end{aligned}
\end{equation}
where $\xi=re^{i\theta}$, and where we used Eq.~(\ref{eq:commutDSR}) in the second line and Eq.~(\ref{eq:stellarG}) in the last line. Hence, the squeezing operator $\hat S(\xi)$ acts on the stellar function as
\be
F_\psi^\star(z)\mapsto\frac1{\sqrt{\cosh r}}F_\psi^\star(\cosh rz-\sinh re^{-i\theta}\partial_z)e^{\frac12e^{i\theta}(\tanh r)z^2},
\label{squestellar}
\ee
for all $\xi=re^{i\theta}\in\mathbb C$. 

\end{proof}

\noindent In particular, for all $n\in\mathbb N$,
\begin{equation}\label{eq:sqFock}
    \begin{aligned}
        F^\star_{\hat S(\xi)\ket n}(z)&=\frac1{\sqrt{n!}}(\cosh rz-\sinh re^{-i\theta}\partial_z)^ne^{\frac12e^{i\theta}\tanh rz^2}\\
        &=\frac{(e^{-i\theta}\tanh r)^{n/2}}{\sqrt{n!\cosh r}}\mathit{He}_n\left(\frac z{\cosh r\sqrt{e^{-i\theta}\tanh r}}\right)e^{\frac12e^{i\theta}(\tanh r)z^2},
    \end{aligned}
\end{equation}
for all $z\in\mathbb C$, where $\xi:=re^{i\theta}$, where $\mathit{He}_n(x)=\sum_{k=0}^{\lfloor\frac n2\rfloor}\frac{(-1)^kn!}{2^kk!(n-2k)!}x^{n-2k}$ is the $n^{th}$ Hermite polynomial, and where the last line is obtained by induction using $\mathit{He}_{n+1}(z)=z\mathit{He}_n(z)-\partial_z\mathit{He}_n(z)$~\cite{abramowitz1965handbook}. Hence, for all $\ket\psi=\sum_{n\ge0}{\psi_n\ket n}$ we have
\begin{equation}\label{eq:actionS}
    F^\star_{\hat S(\xi)\ket\psi}=e^{\frac12e^{i\theta}(\tanh r)z^2}\sum_{n\ge0}\frac{(e^{-i\theta}\tanh r)^{n/2}\psi_n}{\sqrt{n!\cosh r}}\mathit{He}_n\left(\frac z{\cosh r\sqrt{e^{-i\theta}\tanh r}}\right).
\end{equation}
Finally, for all $s\in\mathbb R$, the action of the shearing operator $\hat P(s)$ is obtained by combining Lemmas~\ref{lem:directevoR} and~\ref{lem:directevoS}, since by Eq.~(\ref{eq:PRS}) we have $\hat P(s)=\hat R(\varphi)\hat S(\xi)$ with $\varphi=\arg(1-is)$, and $\xi:=re^{i\theta}$ with $r=\sinh^{-1} s$ and $\theta=-\arg(1-is)-\frac\pi2$.

Having explored the dynamics of single-mode quantum states through the lens of their stellar representation, we turn to the multimode case in the following sections.

%------------------------------------------------------------------------

\section{Multimode stellar hierarchy}
\label{sec:multi}

The single-mode stellar hierarchy has been defined in~\cite{chabaud2020stellar} and reviewed in Sec.~\ref{sec:SBstellar}. It provides a very useful tool for studying non-Gaussian quantum states in the single-mode case~\cite{chabaud2020stellar,chabaud2021certification,PRXQuantum.2.030204}.
Since non-Gaussian elements of an infinite-dimensional quantum computation may be thought of as resources for achieving a computational speedup~\cite{Bartlett2002}, this motivates a generalization of the stellar hierarchy to the multimode setting.

In this section, we give the formal definition of the multimode stellar hierarchy (see Fig.~\ref{fig:hierarchy}) for both pure and mixed quantum states and prove that its single-mode properties extend to the multimode case:
\begin{enumerate}[label=(\roman*)]
    \item\label{enum:i} The stellar function of a pure state is the unique holomorphic function which generates this state from the vacuum when evaluated with creation operators (Lemma~\ref{lem:unique}). In particular, it gives a unique prescription for engineering a quantum state from the vacuum by applying a generating function of the creation operators of the modes.
    \item\label{enum:ii} The form $P\times G$ of the stellar function translates to the structure of the corresponding quantum state: pure states of finite stellar rank have an expansion of the form $P(\hat{\bm a}^\dag)\ket G$ where $P$ is a polynomial and $\ket G$ is a Gaussian state (Lemma~\ref{lem:decomp1}). 
    \item\label{enum:iii} A unitary  operation  is  Gaussian  if  and  only  if  it  leaves  the stellar  rank invariant for any state (Theorem~\ref{th:invar}). Moreover, the stellar rank is non-increasing under Gaussian channels and measurements (Corollary~\ref{coro:nonincr}). As a result, the stellar rank is a non-Gaussian monotone which provides a measure of the non-Gaussian character of a quantum state.
    \item\label{enum:iv} A state of stellar rank $n$ cannot be obtained from the vacuum by using less than $n$ applications of creation operators, together with Gaussian unitary operations (Lemma~\ref{lem:addition}). In other words, the stellar rank has an operational interpretation relating to point~\ref{enum:i}, as a lower bound on the minimal number of elementary non-Gaussian operations (applications of a single-mode creation operator) needed to engineer the state from the vacuum---a notion which may be thought of intuitively as a possible CV equivalent to the $T$-count in DV quantum circuits~\cite{amy2013meet,beverland2020lower}.
    \item\label{enum:v} The structure from point~\ref{enum:ii} can be reverted, yielding another useful characterization: pure states of finite stellar rank also have an expansion of the form $\hat G\ket C$, where $\hat G$ is a Gaussian unitary and $\ket C$ is a core state, i.e., a state with polynomial stellar function (Lemma~\ref{lem:decomp2}). As a result, pure states of finite stellar rank $n$ form the orbit under Gaussian unitaries of pure states with particle number bounded by $n$.
    \item\label{enum:vi} The stellar hierarchy is robust with respect to the trace distance (Theorem~\ref{th:robust}), i.e., every state of a given finite stellar rank only has states of equal or higher rank in its close vicinity. As a consequence, the stellar rank of multimode states can be witnessed experimentally.
    \item\label{enum:vii} Pure states of finite stellar rank form a dense subset of the Hilbert space (Lemma~\ref{lem:dense}), so that pure states of infinite rank can be approximated arbitrarily well in trace distance by pure states of finite stellar rank.
\end{enumerate}

\noindent This generalization of the stellar hierarchy allows us to consider entanglement and computational aspects of infinite-dimensional non-Gaussian quantum states, which we study in Sec.~\ref{sec:entanglement} and Sec.~\ref{sec:MMdynamics}, respectively.

\subsection{Definitions and preliminary results}

The stellar function of a pure $m$-mode quantum state $\ket{\bm\psi}=\sum_{\bm n\ge\bm0}\psi_{\bm n}\ket{\bm n}$ is defined as
\begin{equation}\label{eq:defstellarmult}
    F^\star_{\bm\psi}(\bm z)=\sum_{\bm n\ge\bm0}\frac{\psi_{\bm n}}{\sqrt{\bm n!}}\bm z^{\bm n},
\end{equation}
for all $\bm z\in\mathbb C^m$. 
The stellar function of a state $\ket{\bm\psi}$ is related to its Husimi $Q$ function~\cite{husimi1940some} as
\begin{equation}\label{eq:Husimistellarmulti}
    Q_{\bm\psi}(\bm\alpha)=\frac{e^{-|\bm\alpha|^2}}{\pi^m}\left|F^\star_{\bm\psi}(\bm\alpha^*)\right|^2.
\end{equation}
Since the Husimi $Q$ function is a quasiprobability density over phase space, the stellar function may be thought of as a phase-space wave function. In particular, the stellar function and the Husimi $Q$ function share the same zeros, up to complex conjugation and multiplicity.
We have the following result~\cite{lutkenhaus1995nonclassical}: 

\begin{theorem}\label{th:Husimizeros}
A pure state is non-Gaussian if and only if its Husimi $Q$ function has zeros.
\end{theorem}

\noindent Note that this is equivalent to Hudson's theorem~\cite{hudson1974wigner,soto1983wigner}, which asserts that a pure state is non-Gaussian if and only if its Wigner function has negative values, since the Husimi $Q$ function is a non-negative function obtained from the Wigner function by a Gaussian convolution~\cite{cahill1969density}. A direct consequence is that a pure state is non-Gaussian if and only if its stellar function has zeros.

The $m$-mode coherent state of amplitude $\bm\alpha\in\mathbb C^m$ is defined as
\begin{equation}
    \ket{\bm\alpha}=e^{-\frac12|\bm\alpha|^2}\sum_{\bm n\ge\bm0}\frac{\bm\alpha^{\bm n}}{\sqrt{\bm n!}}\ket{\bm n}.
\end{equation}
The stellar function may then be expressed using the overlap with a coherent state of amplitude $\bm z^*\in\mathbb C^m$:
\begin{equation}
    F^\star_{\bm\psi}(\bm z)=e^{\frac12|\bm z|^2}\braket{\bm z^*|\bm\psi},
\end{equation}
for all $\bm z\in\mathbb C^m$. Hence, a pure state is non-Gaussian if and only if it is orthogonal to at least one coherent state.

From the definition, we also obtain that the stellar function of a tensor product of pure states is given by the product of the stellar functions of each state, over different variables: let $m_1,m_2\in\mathbb N^*$, let $\ket{\bm\psi_1}\in\mathcal H^{\otimes m_1}$ and $\ket{\bm\psi_2}\in\mathcal H^{\otimes m_2}$, then
\begin{equation}\label{eq:tensorproduct}
    F^\star_{\bm\psi_1\otimes\bm\psi_2}(\bm z_1,\bm z_2)=F^\star_{\bm\psi_1}(\bm z_1)F^\star_{\bm\psi_2}(\bm z_2),
\end{equation}
for all $\bm z_1\in\mathbb C^{m_1}$ and all $\bm z_2\in\mathbb C^{m_2}$. Note that the stellar function is also linear with respect to superpositions, while the Husimi $Q$ function is linear with respect to mixtures.

In the multimode case, a natural generalization of the single-mode stellar hierarchy is defined by considering multivariate holomorphic functions of the form $P\times G$, where $P$ is a multivariate polynomial and $G$ is a multivariate Gaussian function. The multimode stellar rank $r^\star$ is then defined as the degree of the multivariate polynomial $P$. The set of $m$-mode pure quantum states of stellar rank $n$ is denoted $\bm R_n$. The pure quantum states whose stellar function is not of the form $P\times G$ have infinite stellar rank by definition. They form the set $\bm R_\infty$. As in the single-mode case, the states with a polynomial stellar function are referred to as core states. These are the states with bounded support over the Fock basis.

For mixed states, the stellar rank is then given by a convex roof construction:
\begin{equation}\label{eq:rankmixed}
    r^\star(\rho)=\inf_{\{p_i,\psi_i\}}\sup_ir^\star(\psi_i),
\end{equation}
where the infimum is taken over the statistical ensembles $\{p_i,\psi_i\}$ such that $\rho=\sum_ip_i\ket{\psi_i}\!\bra{\psi_i}$. This quantity is challenging to compute in practice, but thanks to the robustness of the stellar hierarchy~\ref{enum:vi} there are practical and efficient quantum algorithms which provide robust lower bounds on the stellar rank of mixed states through direct fidelity estimation with target pure states~\cite{chabaud2021certification,fiuravsek2022efficient}.

A consequence of Eq.~(\ref{eq:tensorproduct}) is the stellar rank is fully additive for pure states: $r^\star(\rho\otimes\psi)=r^\star(\rho)+r^\star(\psi)$, where $\ket\psi$ is pure. For mixed states $\rho$ and $\sigma$ we obtain: 
\begin{equation}\label{eq:tensorproductrank}
    \max\left(r^\star(\rho),r^\star(\sigma)\right)\le r^\star(\rho\otimes\sigma)\le r^\star(\rho)+r^\star(\sigma).
\end{equation}
The upper bound is directly given by the additivity of the stellar rank for pure states and the definition of the stellar rank for mixed states in Eq.~(\ref{eq:rankmixed}). For the lower bound, writing $\rho\otimes\sigma=\sum_kp_k\ket{\chi_k}\!\bra{\chi_k}$ such that $\sup_kr^\star(\chi_k)=r^\star(\rho\otimes\sigma)$, the states $\ket{\chi_k}$ are separable since $\rho\otimes\sigma$ is. Writing $\ket{\chi_k}=\ket{\chi_k^\rho}\otimes\ket{\chi_k^\sigma}$, we have $r^\star(\rho\otimes\sigma)=\sup_k[r^\star(\chi_k^\rho)+r^\star(\chi_k^\sigma)]$. Taking the partial trace of $\rho\otimes\sigma$ with respect to the second subsystem yields $\rho=\sum_kp_k\ket{\chi_k^\rho}\!\bra{\chi_k^\rho}$. Hence, $r^\star(\rho)\le\sup_kr^\star(\chi_k^\rho)\le r^\star(\rho\otimes\sigma)$. The same result is obtained for $\sigma$ by taking the partial trace with respect to the first subsystem.

In particular, for any mixture of Gaussian states $\rho_G$ we have $r^\star(\rho_G)=0$, as we show in the following section, and thus $r^\star(\rho\otimes\rho_G)=r^\star(\rho)+r^\star(\rho_G)=r^\star(\rho)$. We leave the full additivity of the stellar rank for mixed states as an interesting open question.
From the definition, the stellar rank also satisfies:
\begin{equation}\label{eq:weakconvex}
    r^\star\left(\sum_np_n\rho_n\right)\le\sup_nr^\star(\rho_n).
\end{equation}
This inequality is also valid for continuous sums and is not tight in general: for example, setting $p_n=\frac{|\alpha|^2}{n!}$ for $\alpha\in\mathbb C$ and $\rho_n=\ket n\!\bra n$ for all $n\in\mathbb N$, $\sum_np_n\rho_n$ is a rotation-invariant coherent state of amplitude $\alpha$: $\int_\varphi\hat R(\varphi)\ket\alpha\!\bra\alpha\hat R(-\varphi)\frac{d\varphi}{2\pi}$, of stellar rank $0$, while $\rho_n$ is a Fock state of stellar rank $n$---this example also shows the importance of the infimum over all decompositions in the definition of the stellar rank for mixed states.

\subsection{Properties of the multimode stellar hierarchy}

\noindent We now turn to the formal statements of the properties \ref{enum:i}-\ref{enum:vii}.

\medskip

\ref{enum:i} The stellar function is a holomorphic function over $\mathbb C^m$, which may be thought of as the generating function of the state:

\begin{lemma}\label{lem:unique}
Let $\ket{\bm\psi}\in\mathcal H^{\otimes m}$. Then, $F^\star_{\bm\psi}$ is the unique holomorphic function over $\mathbb C^m$ such that
\begin{equation}
    \ket{\bm\psi}=F^\star_{\bm\psi}(\hat{\bm a}^\dag)\ket{\bm0}.
\end{equation}
In other words, the map $\ket{\bm\psi}\mapsto F^\star_{\bm\psi}$ is an isometry of the Hilbert space spanned by the Fock basis onto the Segal--Bargmann space with inverse $F\mapsto F(\hat{\bm a}^\dag)\ket{\bm0}$.
\end{lemma}

\begin{proof}
This is a standard result of the Segal--Bargmann representation and we refer to~\cite{hall260holomorphic} for a proof.
\end{proof}

\noindent This Lemma is a very useful tool to obtain the stellar function of a given state $\ket{\bm\psi}$: rather than using the definition in Eq.~(\ref{eq:defstellarmult}), it will often be simpler to exhibit a holomorphic function $f$ such that $\ket{\bm\psi}=f(\bm a^\dag)\ket{\bm0}$. We illustrate this hereafter for the case of Gaussian states.

Recall that by the Euler decomposition in Eq.~(\ref{eq:decompGmulti}), any $m$-mode Gaussian unitary operation $\hat G$ can be decomposed as
\begin{equation}
    \hat G=\hat U\hat S(\bm\xi)\hat D(\bm\beta)\hat V,
\end{equation}
where $\hat U$ and $\hat V$ are passive linear operations, $\hat S(\bm\xi)=\bigotimes_{j=1}^m\hat S(\xi_j)$ is a tensor product of single-mode squeezing operators, and $\hat D(\bm\beta)=\bigotimes_{j=1}^m\hat D(\beta_j)$ is a tensor product of single-mode displacement operators.
The passive linear operations do not change the number of photons and thus leave the vacuum state invariant, and map core states to core states of the same stellar rank. Any $m$-mode Gaussian pure state $\ket G$ may be obtained from the vacuum as
\begin{equation}\label{eq:stateGsimplified}
    \begin{aligned}
        \ket G&=\hat U\hat S(\bm\xi)\hat D(\bm\beta)\hat V\ket{\bm 0}\\
        &=\hat U\hat S(\bm\xi)\hat D(\bm\beta)\ket{\bm 0}\\
        &=\hat U\left(\bigotimes_{j=1}^m\hat S(\xi_j)\hat D(\beta_j)\ket0\right)\\
        &=\hat U\bigotimes_{j=1}^m\ket{G_j},
    \end{aligned}
\end{equation}
where we have set $\ket{G_j}:=\hat S(\xi_j)\hat D(\beta_j)\ket0$.
With Eqs.~(\ref{eq:stellarG}) and~(\ref{eq:tensorproduct}), the stellar function of $\bigotimes_{j=1}^m\ket{G_j}$ is given by
\begin{equation}
    \prod_{j=1}^m\left[(1-|t_j|^2)^{1/4}e^{-\frac12t_jz_j^2+b_jz_j+c_j}\right]=\left[\prod_{j=1}^m(1-|t_j|^2)^{1/4}\right]e^{-\frac12\sum_{j=1}^mt_jz_j^2+\sum_{j=1}^mb_jz_j+\sum_{j=1}^mc_j},
\end{equation}
where $t_j=-e^{i\theta_j}\tanh r_j$, $b_j=\beta_j\sqrt{1-|t_j|^2}$, $c_j=\frac12t_j^*\beta_j^2-\frac12|\beta_j|^2$, where we have set $\xi_j:=r_je^{i\theta_j}$, for all $j\in\{1,\dots,m\}$. With an application of Lemma~\ref{lem:unique}, Eq.~(\ref{eq:stateGsimplified}) then gives
\begin{equation}
    \begin{aligned}
        \ket G&=\hat U\prod_{j=1}^m(1-|t_j|^2)^{1/4}e^{-\frac12\sum_{j=1}^mt_j\hat a_j^{\dag2}+\sum_{j=1}^mb_j\hat a_j^\dag+\sum_{j=1}^mc_j}\ket{\bm0}\\
        &=\left[\prod_{j=1}^m(1-|t_j|^2)^{1/4}\right]e^{-\frac12\sum_{j=1}^mt_j\left(\sum_{k=1}^mu_{jk}\hat a_k^\dag\right)^2+\sum_{j=1}^mb_j\sum_{k=1}^mu_{jk}\hat a_k^\dag+\sum_{j=1}^mc_j}\hat U\ket{\bm0}\\
        &=\left[\prod_{j=1}^m(1-|t_j|^2)^{1/4}\right]e^{-\frac12\sum_{k,l=1}^m\left(\sum_{j=1}^mt_ju_{jk}u_{jl}\right)\hat a_k^\dag\hat a_l^\dag+\sum_{k=1}^m\left(\sum_{j=1}^mb_ju_{jk}\right)\hat a_k^\dag+\sum_{j=1}^mc_j}\ket{\bm0},
    \end{aligned}
\end{equation}
where $U=(u_{ij})_{1\le i,j\le m}$ is the unitary matrix describing the action of $\hat U$ on the creation operators of the modes, and where we used Eq.~(\ref{eq:commutU}) in the second line and $\hat U\ket{\bm0}=\ket{\bm0}$ in the last line.

Hence, by Lemma~\ref{lem:unique}, the stellar function of this state is given by
\begin{equation}\label{eq:stellarGmulti}
    \begin{aligned}
        G(\bm z):=F^\star_G(\bm z)=\frac1{\mathcal N}e^{-\frac12\bm z^TA\bm z+B^T\bm z+C},
    \end{aligned}
\end{equation}
where
\begin{equation}\label{eq:stellarGmultinotations}
    \begin{aligned}
        \mathcal N&=\prod_{j=1}^m(1-|t_j|^2)^{-1/4}=\sqrt{\cosh r_1\cdots\cosh r_m}>0\\        A&:=(A_{kl})_{1\le k,l\le m}=U^T\text{Diag}_{j=1,\dots,m}(t_j)U\in\mathbb C^{m\times m}\\
        B&=U^T\bm b=U^T(b_1,\dots,b_m)^T\in\mathbb C^m\\
        C&=\sum_{j=1}^mc_j\in\mathbb C,
    \end{aligned}
\end{equation}
with $|t_j|<1$ for all $j\in\{1,\dots,m\}$.

\medskip

\ref{enum:ii} As another consequence of Lemma~\ref{lem:unique}, we give the following useful characterization of states of a given finite rank:

\begin{lemma}\label{lem:decomp1}
Let $n\in\mathbb N$ and let $\ket{\bm\psi}\in\mathcal H^{\otimes m}$ be a state of stellar rank $n$, with $F^\star_{\bm\psi}(\bm z)=P(\bm z)G(\bm z)$, with $P$ a multivariate polynomial of degree $n$ and $G$ a multivariate Gaussian function as in Eq.~(\ref{eq:stellarGmulti}). Then, there exists a unique $m$-mode Gaussian state $\ket{G_{\bm\psi}}$ such that
\begin{equation}
    \ket{\bm\psi}\propto P(\hat{\bm a}^\dag)\ket{G_{\bm\psi}},
\end{equation}
where $\propto$ denotes equality up to a normalization factor.
Moreover, writing $\ket{G_{\bm\psi}}=\hat G_{\bm\psi}\ket{\bm 0}$, where
\begin{equation}
    \hat G_{\bm\psi}:=\hat U\hat S(\bm\xi)\hat D(\bm\beta),
\end{equation}
with $\bm\beta,\bm\xi\in\mathbb C^m$, and
\begin{equation}
    G(\bm z):=e^{-\frac12\bm z^TA\bm z+B^T\bm z},
\end{equation}
with $A\in\mathbb C^{m\times m}$, $B\in\mathbb C^m$ and $c\in\mathbb C$, we have
\begin{equation}\label{eq:ABGaussian}
    \begin{aligned}
        A&=U^T\text{Diag}_{j=1,\dots,m}(-e^{i\theta_j}\tanh r_j)U\\
        B&=U^T\left(\frac{\beta_1}{\cosh r_1},\dots,\frac{\beta_m}{\cosh r_m}\right)^T,
    \end{aligned}
\end{equation}
where we have set $\xi_j:=r_je^{i\theta_j}$, for all $j\in\{1,\dots,m\}$, and where $U$ is the unitary matrix describing the action of $\hat U$ on the creation operators of the modes.

Conversely, any state of the form $P(\hat{\bm a}^\dag)\ket G$, where $P$ is a multivariate polynomial of degree $n$ and $\ket G$ is a Gaussian state, has a stellar rank equal to $n$.
\end{lemma}

\begin{proof}

Note that the general case is a direct corollary of the case $n=0$. Moreover, when $n=0$ the stellar function of $\ket{\bm\psi}$ has no zeros, so by Theorem~\ref{th:Husimizeros} the state $\ket{\bm\psi}$ is Gaussian. A simple rewriting of Eqs.~(\ref{eq:stellarGmulti}) and~(\ref{eq:stellarGmultinotations}) completes the proof, using $t_j=-e^{i\theta_j}\tanh r_j$ and $b_j=\beta_j\sqrt{1-|a_j|^2}=\frac{\beta_j}{\cosh r_j}$, for all $j\in\{1,\dots,m\}$.
The other direction follows by Eq.~(\ref{eq:stellarGmulti}) and Lemma~\ref{lem:unique}.

\end{proof}

\noindent In particular, the set $\bm R_0$ of pure states of stellar rank $0$ is exactly the set of pure Gaussian states, while non-Gaussian states populate all higher ranks.

\medskip

\ref{enum:iii} Using this characterization, we obtain that the multimode stellar hierarchy indeed corresponds to a non-Gaussian hierarchy:

\begin{theorem}\label{th:invar}
Let $\hat U$ be a unitary operator over $m$ modes and let $\hat{\mathbb I}$ denote the identity operator over $m$ modes.
The following propositions are equivalent: 
    \begin{enumerate}
    \item $\hat U$ is a Gaussian unitary;
    \item $\hat U$ maps Gaussian states to Gaussian states;
    \item $\hat U$ and $\hat U\otimes\hat{\mathbb I}$ preserve the stellar rank of any quantum state.
    \end{enumerate}
\end{theorem}

\noindent As a result, any $m$-mode Gaussian unitary operation preserves the $m$-mode stellar rank of any pure or mixed quantum state. The proof of this Theorem is given in Appendix~\ref{app:proofsMM}, together with the proof of the following result:

\begin{coro}\label{coro:nonincr}
The stellar rank is non-increasing under Gaussian channels and measurements.
\end{coro}

\medskip

\ref{enum:iv} A direct consequence is the following operational property of the multimode stellar rank:

\begin{lemma}\label{lem:addition}
A quantum state of stellar rank $n$ cannot be engineered from the vacuum using less than $n$ applications photon-additions or subtraction, together with Gaussian unitary operations.
\end{lemma}

\begin{proof}

The operation of photon-addition on mode $k$ corresponds to the application of the creation operator $\hat a_k^\dag$, which in the stellar representation amounts to a multiplication by $z_k$---thus increasing the finite stellar rank by one. 

Similarly, the operation of photon-subtraction on mode $k$ of a state of finite stellar rank corresponds to the application of the annihilation operator $\hat a_k$, which in the stellar representation amounts to a partial derivative with respect to $z_k$, of a function of the form $P(\bm z)G(\bm z)$, where $P$ is a polynomial and $G$ a Gaussian function. This can at most increase the stellar rank by $1$, if $G$ has a non-zero term of the form $z_k^2$ and if the maximal monomial of $P$ depends on $z_k$.

Since the vacuum state has stellar rank $0$, and given that Gaussian unitary operations leave the stellar rank invariant by Theorem~\ref{th:invar}, a pure state of stellar rank $n$ cannot be engineered from the vacuum using less than $n$ photon-additions or subtraction, together with Gaussian unitary operations. The case of mixed states is a direct consequence of Eq.~(\ref{eq:rankmixed}).

\end{proof}

\medskip
 
\ref{enum:v} The decomposition from Lemma~\ref{lem:decomp1} gives an expansion for states of finite stellar rank as a polynomial in creation operators applied to a Gaussian state, with Gaussian stellar function. These states also have a reverted expansion where they are obtained as a Gaussian unitary applied to a core state, with polynomial stellar function:

\begin{lemma}\label{lem:decomp2}
Let $n\in\mathbb N$ and let $\ket{\bm\psi}\in\mathcal H^{\otimes m}$ be a state of stellar rank $n$, with $F^\star_{\bm\psi}(\bm z)=P(\bm z)G(\bm z)$, with $P$ a multivariate polynomial of degree $n$ and $G$ a multivariate Gaussian function. We write $\ket{\bm\psi}\propto P(\hat{\bm a}^\dag)\hat G_{\bm\psi}\ket{\bm0}$ as in Lemma~\ref{lem:decomp1}, where
\begin{equation}
    \hat G_{\bm\psi}:=\hat U\hat S(\bm\xi)\hat D(\bm\beta),
\end{equation}
with $\bm\beta=(\beta_1,\dots,\beta_m),\bm\xi=(\xi_1,\dots,\xi_m)\in\mathbb C^m$. We also write $\xi_j=r_je^{i\theta_j}$ for all $j\in\{1,\dots,m\}$ and denote by $U$ the unitary matrix describing the action of $\hat U$ on the creation and annihilation operators of the modes. Then, there exists a unique core state $\ket{C_{\bm\psi}}$ of stellar rank $n$ such that $\ket{\bm\psi}=\hat G_{\bm\psi}\ket{C_{\bm\psi}}$. Its stellar function is given by
\begin{equation}
    F^\star_{C_{\bm\psi}(\bm z)}=P\left[(\cosh(r_j)(U^\dag\bm z)_j+e^{i\theta_j}\sinh(r_j)\partial(U\bm z)_j-\beta_j^*)_{j=1,\dots,m}\right]\cdot1.
\end{equation}
Conversely, any state of the form $\hat G\ket C$, where $\hat G$ is a Gaussian unitary and $\ket C$ is a core state of stellar rank $n$, also has stellar rank equal to $n$.
\end{lemma}

\begin{proof}

By Eqs.~(\ref{eq:commutDSR}) and (\ref{eq:actionU}) we have:
\begin{equation}
    \begin{aligned}
        \hat G_{\bm\psi}^\dag\ket{\bm\psi}&=\hat G_{\bm\psi}^\dag P(\hat{\bm a}^\dag)\hat G_{\bm\psi}\ket{\bm0}\\
        &=P\left[(\cosh(r_j)(U^\dag\hat{\bm a}^\dag)_j+e^{i\theta_j}\sinh(r_j)(U\hat{\bm a})_j-\beta_j^*)_{j=1,\dots,m}\right]\ket{\bm0},
    \end{aligned}
\end{equation}
so $\ket{\bm\psi}=\hat G_{\bm\psi}\ket{C_{\bm\psi}}$ where $\ket{C_{\bm\psi}}$ is a core state with
\begin{equation}\label{eq:corestateSF}
    F^\star_{C_{\bm\psi}}(\bm z)=P\left[(\cosh(r_j)(U^\dag\bm z)_j+e^{i\theta_j}\sinh(r_j)\partial(U\bm z)_j-\beta_j^*)_{j=1,\dots,m}\right]\cdot1.
\end{equation}
The Gaussian state $\ket{G_{\bm\psi}}$ is unique by Lemma~\ref{lem:decomp1}, so the Gaussian unitary operation $\hat G_{\bm\psi}$ of the form $\hat U\hat S\hat D$ such that $\ket{G_{\bm\psi}}=\hat G_{\bm\psi}\ket{\bm0}$ is unique. Hence, the core state $\ket{C_{\bm\psi}}=\hat G_{\bm\psi}^\dag\ket{\bm\psi}$ is also unique.

Conversely, if $\ket C$ is a core state of stellar rank $n$, then by Theorem~\ref{th:invar} $\hat G\ket C$ also has stellar rank equal to $n$ for any Gaussian unitary $\hat G$.

\end{proof}

\noindent A definition of the multimode stellar rank was given in~\cite{chabaud2020classical}, based on the structure $\hat G\ket C$ of quantum states. Lemma~\ref{lem:decomp2} shows that this definition coincides with ours, which is based on the structure of the stellar function itself.

\medskip

\ref{enum:vi} A crucial property of the stellar hierarchy which makes it practically relevant is its experimental certifiability, due to its robustness. Namely, we show that given a state of stellar rank $n$, all states in its close vicinity have rank $n$ or higher. This implies that an estimate of the fidelity of an experimental state with a pure state of stellar rank $n$ may be used as a witness for the stellar rank~\cite{chabaud2021certification,fiuravsek2022efficient}.

Recall that for all $n\in\mathbb N\cup\{+\infty\}$, $\bm R_n$ denotes the set of pure states of stellar rank $n$.

\begin{theorem}\label{th:robust}
The $m$-mode stellar hierarchy is robust for the trace norm. Formally, for all $n\in\mathbb N$,
\begin{equation}
    \overline{\bm R_n}=\bigcup_{0\le k\le n}\bm R_k,
\end{equation}
where $\overline X$ denotes the closure of $X$ for the trace norm.
\end{theorem}

\noindent The proof of this Theorem is detailed in Appendix~\ref{app:proofsMM}.

\medskip

\ref{enum:vii} The following result shows that any pure quantum state can be approximated arbitrarily well by a sequence of pure states of finite stellar ranks.

\begin{lemma}\label{lem:dense}
The set of pure states of finite stellar rank is dense for the trace norm. Formally,
\begin{equation}
    \overline{\bigcup_{n\in\mathbb N}\bm R_n}=\mathcal H^{\otimes m}.
\end{equation}
\end{lemma} 

\begin{proof}

Let $\ket{\bm\psi}=\sum_{\bm n\ge\bm0}\psi_{\bm n}\ket{\bm n}\in\mathcal H^{\otimes m}$. This state is normalized so it has a nonzero coefficient $\psi_{\bm m}$. Let $m=|\bm m|$. For all $l\ge m$, the sequence of cutoff states
\begin{equation}
    \ket{\bm\psi_l}=\frac1{\sqrt{\sum_{|\bm n|\le l}|\psi_{\bm n}|^2}}\sum_{|\bm n|\le l}\psi_{\bm n}\ket{\bm n},
\end{equation}
is well-defined, and converges in trace norm to $\ket{\bm\psi}$. Since these cutoff states are normalized core states of finite stellar rank $l\in\mathbb N$, this proves the claim.

\end{proof}

\subsection{Entanglement and factorization of the stellar function}
\label{sec:entanglement}

We have shown that the important properties of the stellar hierarchy generalize from the single-mode to the multimode case.
On the other hand, some features are specific to the single-mode setting, because in this case the set of zeros of the stellar function is countable---and even discrete---while in the multimode case it is in general uncountable. In particular, the Hadamard--Weierstrass factorization theorem does not have a generalization for multivariate analytic functions (this is known as the second Cousin problem). A direct consequence is that states of finite stellar rank $n$ in the single-mode case can be obtained from the vacuum using exactly $n$ photon-additions, together with Gaussian unitary operations~\cite{chabaud2020stellar}, while states of finite stellar rank $n$ in the multimode case need at least $n$ photon-additions---but some require more (for instance, the state $\frac1{\sqrt2}(\ket{20}+\ket{01})$ has stellar rank $2$ but cannot be obtained exactly from the vacuum using any finite number of photon-additions, together with Gaussian unitary operations~\cite{chabaud2020classical}).

In practice, this difference does not manifest in experimental scenarios, because the robust property of the stellar hierarchy is having a stellar rank greater of equal to $n$, rather than exactly $n$ (Theorem~\ref{th:robust}). In both the single-mode and multimode cases, the conclusion reached is a lower bound on the number of applications of creation operators necessary to engineer the state, rather than the exact number needed.

This difference stems from the factorization properties of the stellar function that become nontrivial for multivariate functions. As we show hereafter, these relate to entanglement properties: if a pure quantum state is separable, then its analytic representation is factorizable.

Recall that in the finite-dimensional case, pure quantum states may be represented by multivariate polynomials (see Sec.~\ref{sec:analytic}). This representation is multiplicative with respect to the tensor product, so we obtain the following relation between separability and factorization:

\begin{lemma}\label{lem:sepfactDV}
Let $n,d\ge2$, let $I,J\subset\{1,\dots,n\}$ such that $I\sqcup J=\{1,\dots,n\}$, and let us write $\bm z=(z_1,\dots,z_n)\in\mathbb C^n$. A pure quantum state of $n$ qudits of dimension $d$ with polynomial representation $\bm z\mapsto P(\bm z)$ with degree at most $d$ in each variable is separable over the bipartite partition $I,J$ if and only if there exist polynomials $P_I$ and $P_J$ over $|I|$ and $|J|$ variables, respectively, such that $P(\bm z)=P_I(\bm z_I)\times P_J(\bm z_J)$, where $\bm z_I=(z_i)_{i\in I}$ and $\bm z_J=(z_j)_{j\in J}$.
\end{lemma}
 
\noindent In the CV case, a similar statement holds for holomorphic functions:

\begin{lemma}\label{lem:sepfactCV}
Let $m\ge2$, let $I,J\subset\{1,\dots,m\}$ such that $I\sqcup J=\{1,\dots,m\}$, and let us write $\bm z=(z_1,\dots,z_m)\in\mathbb C^m$. A pure $m$-mode quantum state with holomorphic representation $\bm z\mapsto F^\star(\bm z)$ is separable over the bipartite partition $I,J$ if and only if there exist holomorphic functions $F^\star_I$ and $F^\star_J$ in Segal--Bargmann space, over $|I|$ and $|J|$ variables, respectively, such that $F^\star(\bm z)=F^\star_I(\bm z_I)\times F^\star_J(\bm z_J)$, where $\bm z_I=(z_i)_{i\in I}$ and $\bm z_J=(z_j)_{j\in J}$.
\end{lemma}

\noindent Recall that for pure states the stellar rank is additive with respect to tensor products. As a direct consequence:

\begin{lemma}\label{lem:sepfinite}
A pure quantum state of finite stellar rank is separable if and only if it can be written as a tensor product of pure quantum states of finite stellar rank.
\end{lemma}

\begin{proof}
Let $\ket\psi$ be a separable state of finite stellar rank. Writing $\ket\psi=\ket{\psi_I}\otimes\ket{\psi_J}$, we have $r^\star(\psi)=r^\star(\psi_I)+r^\star(\psi_J)$ so $\ket{\psi_I}$ and $\ket{\psi_J}$ also have finite stellar ranks.
\end{proof}

\noindent Combining this result with the link between entanglement and factorization in Lemma~\ref{lem:sepfactCV}, we obtain the following characterization of entanglement for pure states of finite stellar rank:

\begin{lemma}\label{lem:sepfinitedecomp}
Let $m\ge2$, let $I,J\subset\{1,\dots,m\}$ such that $I\sqcup J=\{1,\dots,m\}$, let $n\in\mathbb N$ and let us write $\bm z=(z_1,\dots,z_m)\in\mathbb C^m$. A pure $m$-mode quantum state of stellar rank $n$ with holomorphic representation $\bm z\mapsto F^\star(\bm z)=P(\bm z)G(\bm z)$ with $\text{deg}(P)=n$ is separable over the bipartite partition $I,J$ if and only if there exist polynomials $P_I$ and $P_J$ with $\text{deg}(P_I)+\text{deg}(P_J)=n$ and Gaussian functions $G_I$ and $G_J$ over $|I|$ and $|J|$ variables, respectively, such that $F^\star(\bm z)=P_I(\bm z_I)P_J(\bm z_J)G_I(\bm z_I)G_J(\bm z_J)$, where $\bm z_I=(z_i)_{i\in I}$ and $\bm z_J=(z_j)_{j\in J}$.

Moreover, in that case $P(\bm z)=P_I(\bm z_I)P_J(\bm z_J)$ and $G(\bm z)=G_I(\bm z_I)G_J(\bm z_J)$.
\end{lemma}

\noindent The second part of the lemma is obtained directly by comparing the growth of the functions. Informally, this result may be understood as follows: entanglement of pure states of finite stellar rank comes in two flavours corresponding to the Gaussian and polynomial parts of their stellar function. The Gaussian part of the stellar function corresponds to Gaussian entanglement: applying Lemma~\ref{lem:sepfinitedecomp} for states of stellar rank $n=0$ shows that separable Gaussian states have a stellar representation of the form $G(\bm z)=G_I(\bm z_I)G_J(\bm z_J)$. With Lemma~\ref{lem:sepfactDV}, the polynomial part of the stellar function corresponds to finite-dimensional entanglement of a core state whose stellar function is given by $P(\bm z)$.

Given its stellar function, testing whether a pure Gaussian state is entangled is almost immediate: writing $G(\bm z)=e^{Q(\bm z)}$, it amounts to checking whether the quadratic exponent may be decomposed as $Q(\bm z)=Q_I(\bm z_I)+Q_J(\bm z_J)$, which may be done by ensuring that $Q(\bm z)$ does not contain any cross term of the form $z_iz_j$, for $(i,j)\in I\times J$. Entanglement in Gaussian states originates from the interplay between passive linear operations and squeezing~\cite{ferraro2005gaussian}. This is well captured by the stellar formalism: starting from a separable Gaussian state, a passive linear operation (corresponding to a linear change of basis for the variables $z_1,\dots,z_m$) would only create entanglement (cross terms of the form $z_iz_j$) if the input Gaussian stellar function contains non-zero quadratic exponents (otherwise the exponent would stay linear), and the quadratic exponents correspond to squeezing amplitudes by Eq.~(\ref{eq:ABGaussian}).
In particular, the prototypical example of entangled Gaussian state, the so-called two-mode squeezed state~\cite{ferraro2005gaussian}:
\begin{equation}
    \sqrt{1-|\lambda|^2}\sum_{n\ge0}\lambda^n\ket n\otimes\ket n,
\end{equation}
for $|\lambda|<1$, has the stellar function $G_\lambda(z_1,z_2)\propto e^{\lambda z_1z_2}$. Hence, the standard form of pure Gaussian states~\cite{botero2003modewise,giedke2003entanglement}, which relates a Gaussian state to a tensor product of two-mode squeezed states by local Gaussian unitaries, can be obtained straightforwardly from its stellar function by isolating the cross terms of the form $z_iz_j$.

On the other hand, checking whether the polynomial part of the stellar function factorizes is more challenging. Similar to the DV case, one may obtain a Schmidt decomposition~\cite{NielsenChuang} by describing the multivariate polynomial with a matrix in a basis of monomials and performing a singular value decomposition.
Given a partition $I,J$ of the modes, this yields a unique standard form for the stellar function $F^\star_{\bm\psi}(\bm z)=P(\bm z)G(\bm z)$ of an $m$-mode state $\ket{\bm\psi}$ of finite stellar rank as follows:
\begin{equation}
    F^\star(\bm z)=\left(\sum_{k=1}^r\tilde\psi_kP_I^{(k)}(\bm z_I)P_J^{(k)}(\bm z_J)\right)G_I(\bm z_I)G_J(\bm z_J)e^{\sum_{(i,j)\in I\times J}\lambda_{ij}z_iz_j},
\end{equation}
where $\bm z_I=(z_i)_{i\in I}$ and $\bm z_J=(z_j)_{j\in J}$, where $r$ is the Schmidt rank of a core state $\ket C$ of stellar function $P(\bm z)$, and where $\lambda_{ij}$ are the cross-term exponents of the stellar function $G(\bm z)$ of a Gaussian state $\ket G$. The state $\ket{\bm\psi}$ is separable if and only if $r=1$, i.e., the core state is separable over the bipartite partition $I,J$, and if $\lambda_{ij}=0$ for all $(i,j)\in I\times J$, i.e., the Gaussian state is separable over the bipartite partition $I,J$.

Note that while the Gaussian state $\ket G$ is equal to the Gaussian state $\ket{G_{\bm\psi}}$ appearing in the decomposition of the state $\ket{\bm\psi}$ in Lemma~\ref{lem:decomp1}, the core state $\ket C$ is not in general equal to the core state $\ket{C_{\bm\psi}}$ appearing in the decomposition of the state $\ket{\bm\psi}$ in Lemma~\ref{lem:decomp2}: the stellar function of the core state $\ket C$ is given by $P(\bm z)$, while that of the core state $\ket{C_{\bm\psi}}$ is given by Eq.~(\ref{eq:corestateSF}). Although these stellar functions are related, one of these core states may be entangled and the other separable, depending on whether the Gaussian state $\ket G$ is separable or not. By the preceding discussion, the core state $\ket C$ being separable is a necessary requirement for the state $\ket{\bm\psi}$ to be separable, while by Lemma~\ref{lem:decomp2} the core state $\ket{C_{\bm\psi}}$ being separable is both a necessary requirement for the state $\ket{\bm\psi}$ to be separable and a sufficient requirement for the state $\ket{\bm\psi}$ to be Gaussian-separable, i.e., that there exists a Gaussian unitary operation $\hat G$ such that $\hat G\ket{\bm\psi}$ is separable. Experimentally, Gaussian-separable states correspond to bipartite states that can be engineered using local states and Gaussian unitary operations.

This last statement can be made more precise by introducing the notion of passive separability~\cite{walschaers2017statistical}: a state $\ket{\bm\phi}$ is passively separable if there exists a passive linear unitary operation $\hat V$ such that $\hat V\ket{\bm\phi}$ is separable. Experimentally, these states correspond to bipartite states that can be engineered using local states and passive linear unitary operations, such as linear optical evolutions.

\begin{lemma}
A pure state of finite stellar rank $\ket{\bm\psi}$ is Gaussian-separable if and only if its core state $\ket{C_{\bm\psi}}$ (as in Lemma~\ref{lem:decomp2}) is passively separable.
\end{lemma}

\begin{proof}

By Lemma~\ref{lem:decomp2}, there exists a Gaussian unitary operation
\begin{equation}
    \hat G_{\bm\psi}=\hat U\hat S(\bm\xi)\hat D(\bm\beta),
\end{equation}
with $\bm\beta=(\beta_1,\dots,\beta_m),\bm\xi=(\xi_1,\dots,\xi_m)\in\mathbb C^m$ and $\hat U$ a passive linear unitary operation such that $\ket{\bm\psi}=\hat G_{\bm\psi}\ket{C_{\bm\psi}}$. Hence, 
\begin{equation}
    \ket{\bm\psi}=\hat U\hat S(\bm\xi)\hat D(\bm\beta)\ket{C_{\bm\psi}}.
\end{equation}
If $\ket{C_{\bm\psi}}$ is passively separable, there exists a passive linear unitary $\hat V$ such that $\hat V\ket{C_{\bm\psi}}$ is separable. Hence, $\hat V\hat D^\dag(\bm\beta)\hat S^\dag(\bm\xi)\hat U^\dag\ket{\bm\psi}$ is separable, and thus $\ket{\bm\psi}$ is Gaussian-separable.

Conversely, if $\ket{\bm\psi}$ is Gaussian-separable, there exists a Gaussian unitary $\hat G$ such that $\hat G\ket{\bm\psi}$ is separable. Hence, $\hat G\hat G_{\bm\psi}\ket{C_{\bm\psi}}$ is separable, so by Lemmas~\ref{lem:sepfinitedecomp} and~\ref{lem:decomp2} there exists Gaussian unitaries $\hat G_I$ and $\hat G_J$ and core states $\ket{C_I}$ and $\ket{C_J}$ such that
\begin{equation}
    \hat G\hat G_{\bm\psi}\ket{C_{\bm\psi}}=(\hat G_I\ket{C_I})\otimes(\hat G_J\ket{C_J}),
\end{equation}
and thus
\begin{equation}\label{eq:interlinsep}
    \hat G_{\bm\psi}\ket{C_{\bm\psi}}=\hat G^\dag(\hat G_I\otimes\hat G_J)\ket{C_I}\otimes\ket{C_J}.
\end{equation}
Let $\hat G':=\hat G^\dag(\hat G_I\otimes\hat G_J)$. This is a Gaussian unitary operation, so there exist $\bm\beta'=(\beta_1',\dots,\beta_m'),\bm\xi'=(\xi_1',\dots,\xi_m')\in\mathbb C^m$ and $\hat U',\hat V'$ passive linear unitary operations such that $\hat G'=\hat U'\hat S(\bm\xi')\hat D(\bm\beta')\hat V'$. Let us define $\ket{C_{\bm\psi}'}:=\hat V'(\ket{C_I}\otimes\ket{C_J})$. Since $\hat V'$ is a passive linear unitary, $\ket{C_{\bm\psi}'}$ is a core state, and is by construction passively separable. Plugging these expressions into Eq.~(\ref{eq:interlinsep}) we finally obtain:
\begin{equation}
    \begin{aligned}
        \ket{\bm\psi}&=\hat G_{\bm\psi}\ket{C_{\bm\psi}}\\
        &=\hat U'\hat S(\bm\xi')\hat D(\bm\beta')\ket{C_{\bm\psi}'}.
    \end{aligned}
\end{equation}
By the uniqueness of the decomposition in Lemma~\ref{lem:decomp2}, this implies $\hat G_{\bm\psi}=\hat U'\hat S(\bm\xi')\hat D(\bm\beta')$ and $\ket{C_{\bm\psi}}=\ket{C_{\bm\psi}'}$, so $\ket{C_{\bm\psi}}$ is passively separable.

\end{proof}

\noindent In particular, in the case where $\ket{\bm\psi}$ is separable, then both $\ket C$ and $\ket{C_{\bm\psi}}$ are separable. The converse does not hold: in the case where $\ket{\bm\psi}$ is a Gaussian state, we have $\ket C=\ket{C_{\bm\psi}}=\ket{\bm0}$, but $\ket{\bm\psi}$ may still be entangled. A state which is not passively separable is referred to as inherently entangled~\cite{walschaers2017statistical} and displays non-Gaussian entanglement~\cite{chabaud2022resources}.

Summarizing the results of this section, we have analysed the structure of entanglement within the multimode stellar hierarchy: entangled states of finite stellar rank may display both Gaussian and non-Gaussian entanglement, and the latter can be thought of as the entanglement of a finite-dimensional system within the infinite-dimensional quantum state.

%------------------------------------------------------------------------

\section{Holomorphic representation of quantum computations}
\label{sec:MMdynamics}

Having generalized the properties of the stellar hierarchy to the multimode case and studied the entanglement structure of states within the hierarchy, we turn to the aspect that initially motivated this generalization: quantum computing. 

We first show how to describe bosonic quantum computations using the Segal--Bargmann holomorphic representation.
Then, we consider a subclass of bosonic computations which we call Rank-Preserving Quantum Computations (RPQC). These correspond to quantum computations preserving the stellar rank---the multimode analogue of the single-mode evolutions from Sec.~\ref{sec:SMdynamics}. We use the formalism developed in Sec.~\ref{sec:multi} to provide a simple classification of existing subuniversal infinite-dimensional quantum computing models based on their computational complexity, which highlights the interplay between non-Gaussianity, squeezing and entanglement at the origin of possible quantum computational speedup in the infinite-dimensional setting. 
Finally, we consider adaptive versions of these RPQC and show that they can encode $\textsf{BQP}$-complete computations.

\subsection{Segal--Bargmann representation of bosonic quantum computations}

In this section, we describe bosonic quantum computations (input state, unitary evolution, measurement) using basic elements of complex analysis through the Segal--Bargmann representation~\cite{segal1963mathematical,bargmann1961hilbert} as follows (see Fig~\ref{fig:HQC}): 
\begin{itemize}
    \item The input states to the computation are described by their stellar function in Segal--Bargmann space, i.e., holomorphic functions over $\mathbb C^m$ that are square-integrable with respect to the Gaussian measure.
    \item The unitary evolution generated by a Hamiltonian $\hat H$ is described by a unitary differential operator $\hat U=e^{-iHt}$, generated by a self-adjoint Hamiltonian $H$ which is a complex power series in the operators $z_k\times$ (multiplication by the $k^{th}$ variable) and $\partial_{z_k}$ (partial derivative with respect to the $k^{th}$ variable) for $k=1,\dots,m$, satisfying $H(z_1,\partial_{z_1},\dots,z_m,\partial_{z_m})=H^*(\partial_{z_1},z_1,\dots,\partial_{z_m},z_m)$. Such unitary operators map holomorphic functions in Segal--Bargmann space to holomorphic functions in Segal--Bargmann space.
    \item The outcome of the computation is sampled from the stellar function of the  output quantum state in two possible ways:
    \begin{itemize}
        \item Continuous measurement: writing $F^\star$ the stellar function of the output state of the computation, a continuous measurement of modes $(1,\dots,m)$ yields a tuple $\bm\alpha=(\alpha_1,\dots,\alpha_m)\in\mathbb C^m$ sampled with the probability density function $\{\frac{e^{-|\bm\alpha|^2}}{\pi^m}|F^\star(\bm\alpha)|^2\}_{\bm\alpha\in\mathbb C^m}$.
        \item Discrete measurement: writing $F^\star$ the stellar function of the output state of the computation, a discrete measurement of modes $(1,\dots,m)$ yields a tuple $\bm n=(n_1,\dots,n_m)\in\mathbb N^m$ sampled from the probability distribution $\{\frac1{\bm n!}|\partial_{\bm z}^{\bm n}F^\star(\bm 0)|^2\}_{\bm n\in\mathbb N^m}$.
    \end{itemize}
    \item Additionally, we also consider adaptive measurements, i.e., when part of the unitary evolution may depend on the outcomes of intermediate (continuous or discrete) measurements.
\end{itemize}

\noindent This correspondence yields a representation of bosonic computations based on holomorphic functions.
With Eq.~(\ref{eq:corresp}), the Hamiltonians are functions of the creation and annihilation operators of the modes, which naturally appear in bosonic experiments. This includes the standard model of CV quantum computing defined by Lloyd and Braunstein~\cite{lloyd1999quantum} wherein Hamiltonians are polynomials in the creation and annihilation operators. In particular, Hamiltonians generating Gaussian unitaries are given by polynomials of degree less or equal to $2$ in the creation and annihilation operators.

Using the relation between the stellar function and the Husimi quasiprobability density in Eq.~(\ref{eq:Husimistellarmulti}), the continuous measurement corresponds to a Gaussian measurement in coherent state basis (a standard heterodyne detection in optics~\cite{Leonhardt-essential}, up to a complex conjugation of the outcomes), i.e., sampling from the probability density function related to the values of the stellar function over complex phase space. On the other hand, using $\bra n=\bra0\frac{\hat a^n}{\sqrt{n!}}$ and Eq.~(\ref{eq:corresp}), the discrete measurement corresponds to a non-Gaussian measurement in Fock basis (a photon-number measurement in optics), i.e., sampling from the probability distribution related to the coefficients of the stellar function. These two types of single-mode measurements are available in current experiments with high efficiencies. 
If some of the modes (say, the first $k$ modes) are measured with the continuous measurement and the remaining $m-k$ modes with the discrete measurement, this yields instead a tuple $(\bm\alpha,\bm n)=(\alpha_1,\dots,\alpha_k,n_1,\dots,n_{m-k})\in\mathbb C^k\times\mathbb N^{m-k}$ sampled with the joint probability density:
\begin{equation}
\left\{\frac{e^{-|\bm\alpha|^2}}{\pi^k\bm n!}|\partial_{\bm z}^{\bm n}F^\star(\bm\alpha,\bm0)|^2\right\}_{(\bm\alpha,\bm n)\in\mathbb C^k\times\mathbb N^{m-k}}.
\end{equation}
Note also that the commonly used homodyne detection (corresponding to a measurement of a position-like operator) is obtained by combining a single-mode Gaussian squeezing operation with a continuous heterodyne detection (yielding an unbalanced heterodyne detection), in the limit where the modulus of the squeezing parameter goes to infinity~\cite{Leonhardt-essential}. Moreover, any Gaussian measurement may be obtained by combining Gaussian unitary operations and homodyne detection~\cite{giedke2002characterization}.

Hence, this holomorphic representation of bosonic computations, while being defined using basic elements of complex analysis, captures infinite-dimensional quantum computational models with bosonic Hamiltonian evolution and Gaussian measurements and/or non-Gaussian photon-number detection (non-adaptive and adaptive). Even though we only consider these two types of measurements, note that Eq.~(\ref{eq:SBscalarprod}) allows for modeling arbitrary projective measurements.

\subsection{Rank-Preserving Quantum Computations}
\label{sec:computational}

In this section, we consider a subclass of bosonic computations which preserve the stellar rank called RPQC. By Theorem~\ref{th:invar}, such computations correspond to Hamiltonian evolutions which are polynomials of degree less of equal to $2$ of the operators $z_k\times$ and $\partial_{z_k}$, for $k=1,\dots,m$, i.e., Gaussian Hamiltonians.
Moreover, by Lemma~\ref{lem:dense}, the set of states of finite stellar rank is dense in Segal--Bargmann space for the trace norm, so we restrict without loss of generality to finite rank stellar functions, of the form $P\times G$, where $P$ is a polynomial and $G$ is a Gaussian function.
These computations (input, evolution, measurement) are thus described as follows (see Fig.~\ref{fig:PxG}):
\begin{itemize}
    \item The input is described by an element of Segal--Bargmann space of finite stellar rank, i.e., a holomorphic function over $\mathbb C^m$ that is square-integrable with respect to the Gaussian measure, of the form $P\times G$, where $P$ is a polynomial and $G$ is a Gaussian function.
    \item The evolution is described by a differential unitary operator $\hat U=e^{-iHt}$ generated by a self-adjoint Hamiltonian $H$ which is a polynomial of degree less or equal to $2$ in the operators $z_k\times$ (multiplication by the $k^{th}$ variable) and $\partial_{z_k}$ (partial derivative with respect to the $k^{th}$ variable), for $k=1,\dots,m$, satisfying $H(z_1,\partial_{z_1},\dots,z_m,\partial_{z_m})=H^*(\partial_{z_1},z_1,\dots,\partial_{z_m},z_m)$. Such unitary operators map holomorphic functions of finite stellar rank in Segal--Bargmann space to other holomorphic functions of the same finite stellar rank in Segal--Bargmann space by Theorem~\ref{th:invar}.
    \item The outcome of the computation is sampled from the stellar function of the  output quantum state in two possible ways:
    \begin{itemize}
        \item Continuous measurement: writing $F^\star$ the output state of the computation, a continuous measurement of modes $(1,\dots,m)$ yields a tuple $\bm\alpha=(\alpha_1,\dots,\alpha_m)\in\mathbb C^m$ sampled with the probability density function $\{\frac{e^{-|\bm\alpha|^2}}{\pi^m}|F^\star(\bm\alpha)|^2\}_{\bm\alpha\in\mathbb C^m}$.
        \item Discrete measurement: writing $F^\star$ the output state of the computation, a discrete measurement of modes $(1,\dots,m)$ yields a tuple $\bm n=(n_1,\dots,n_m)\in\mathbb N^m$ sampled from the probability distribution $\{\frac1{\bm n!}|\partial_{\bm z}^{\bm n}F^\star(\bm 0)|^2\}_{\bm n\in\mathbb N^m}$.
    \end{itemize}
    \item Additionally, we also consider adaptive measurements, i.e., when part of the unitary evolution may depend on the outcomes of intermediate (continuous or discrete) measurements.
\end{itemize}

\noindent Since its unitary evolution is Gaussian, any RPQC may be decomposed efficiently as a combination of displacement, squeezing and passive linear unitary gates (see the Euler decomposition in Eq.~(\ref{eq:decompGmulti})). The latter are described by a finite-dimensional unitary matrix $U$ and act on holomorphic functions in Segal--Bargmann space as
\begin{equation}
    F^\star(\bm z)\mapsto F^\star(U\bm z),
\end{equation}
where we used Eqs.~(\ref{eq:actionU}) and~(\ref{eq:corresp}), together with Lemma~\ref{lem:unique}. In other terms, passive linear operations induce a unitary change of complex coordinates in Segal--Bargmann space. Moreover, the action of a displacement operator on mode $k$ is obtained by fixing all but the $k^{th}$ variable $z_k$ and applying Lemma~\ref{lem:directevoD}.
The case of squeezing, shearing, or any Gaussian quadratic non-linearity follows from the single-mode case: for instance, assuming without loss of generality that a squeezing gate is applied on mode $1$ of a stellar function $F^\star$, we consider the univariate holomorphic function of finite stellar rank given by the section
\begin{equation}
    z_1\mapsto F^\star_{z_2,\dots,z_m}(z_1):=F^\star(z_1,z_2,\dots,z_m),
\end{equation}
where $z_2,\dots,z_m$ are fixed in $\mathbb C$, which follows a single-mode evolution. This function is the product of a Gaussian function with a univariate polynomial in $z_1$, whose coefficients are multivariate polynomials in $z_2,\dots,z_m$. It may be thus treated formally as a single-mode stellar function of finite stellar rank with symbolic coefficients, and we may follow the single-mode derivation of from Sec.~\ref{sec:SMdynamics} to obtain a formal expression for the output stellar function. Note that the zeros of a univariate section of a multivariate polynomial are usually implicit functions of the other variables, thus making more practical the description based on the coefficients of the polynomial part of the stellar function, rather than on its zeros.

An RPQC may therefore be thought qualitatively of as a sequence of motions of Calogero--Moser particles corresponding to the zeros of sections of multimode stellar functions, together with conformal evolutions of Gaussian parameters and unitary changes of coordinates.
While RPQC form a subclass of bosonic computations, we show hereafter that their non-adaptive and adaptive versions already capture various important bosonic quantum computational models.

%------------------------------------------------------------------------

\subsubsection{Non-adaptive RPQC}

We first consider the class of non-adaptive RPQC. We show that this class of computations provides a natural classification of infinite-dimensional subuniversal quantum computational models, based on their computational complexity and their CV properties, summarized in Table~\ref{tab:CVmodels}. For both continuous and discrete measurements, we give reductions of non-adaptive RPQC to existing subuniversal models, which can either be simulated efficiently classically (i.e., their output probabilities may be computed in \textsf{P}), or are hard to sample classically (i.e., estimating multiplicatively a single output probability up to inverse polynomial precision is $\#$\textsf{P}-hard, which implies exact sampling hardness~\cite{Aaronson2013}).

The input stellar functions in RPQC are of the form $P\times G$. A subclass of these holomorphic functions is given by Gaussians $G$. Considering only Gaussian functions and unitary operations that leave invariant the set of Gaussian functions, one obtains restricted bosonic quantum computations over Gaussians, which are a subset of RPQC. The corresponding unitary operations are also Gaussian unitary operations. Adding projections onto coherent states (continuous measurement) yields the Gaussian quantum computing model, which can be efficiently simulated classically~\cite{Bartlett2002}. On the other hand, adding projections onto Fock states (discrete measurement) yields the Gaussian Boson Sampling model, which is hard to sample classically~\cite{Lund2014,Hamilton2016}. However, if no Gaussian quadratic non-linearity is involved such as squeezing, then by Eqs.~(\ref{eq:decompGmulti}) and~(\ref{eq:braidDU}) the computation reduces to sampling the output photon-number of coherent states which undergo a passive linear operation, which corresponds to coherent state sampling~\cite{chabaud2022resources}. These computations can be efficiently simulated classically~\cite{Aaronson2013} since coherent states do not get entangled through passive linear operations. This reflects the necessity of single-mode Gaussian quadratic non-linearity for the presence of Gaussian entanglement (see Sec.~\ref{sec:entanglement}).

Another important subclass of holomorphic functions of the form $P\times G$ is given by polynomials $P$. These functions correspond to stellar functions of core states, i.e., finite superpositions of Fock basis states.
Considering only polynomial functions and unitary operations that leave invariant the set of polynomials of each degree, one obtains restricted bosonic quantum computations (without squeezing) over polynomials, which are a subset of RPQC. The corresponding unitary operations are passive linear operations~\cite{ferraro2005gaussian}. Adding projections onto Fock states (discrete measurement) yields the Boson Sampling model, which is hard to sample classically~\cite{Aaronson2013}. On the other hand, adding projections onto coherent states (continuous measurement) yields a time-reversed version of coherent state sampling, which is also easy to simulate~\cite{Aaronson2013} due to the absence of squeezing.

Finally, if RPQC with continuous measurements do contain squeezing, then they include Boson Sampling with Gaussian measurements such as unbalanced heterodyne, which are hard to sample~\cite{Lund2017,Chakhmakhchyan2017,chabaud2017continuous}. Moreover, these RPQC include CV instantaneous quantum computations, which are also hard to sample~\cite{Douce2017}, by replacing input Gottesman--Kitaev--Preskill (GKP) states~\cite{Gottesman2001} by arbitrarily close approximations of finite stellar rank with Lemma~\ref{lem:dense}. Note that in these cases---since the measurement yields a continuous outcome---hardness of sampling refers to a probability distribution obtained by binning the continuous outcome space.

%------------------------------------------------------------------------

\subsubsection{Adaptive RPQC}

In this section, we consider adaptive versions of RPQC, where part of the unitary evolution may depend on the outcomes of intermediate (continuous or discrete) measurements.

%In order to compare their computational power to that of DV quantum models, we consider uniform versions of these circuits by asking that the descriptions of the input holomorphic function and the unitary evolution are given by the output of a Turing machine, and that the outcome of the computation is written to some precision on the tape of a Turing machine. This is done to avoid the artificial appearance of uncomputable numbers within the model and to provide a fair comparison with DV models. We denote by $\textsf{RPQC}_\textsf{adap}^\textsf{CV}\textsf{P}$ and $\textsf{RPQC}_\textsf{adap}^\textsf{DV}\textsf{P}$ the corresponding class of decision problems that can be solved in polynomial time using these circuits, for continuous and discrete measurements, respectively.

As it turns out, adaptive continuous measurements enhance the computational power of RPQC to universal $\textsf{BQP}$ computations:

\begin{theorem}\label{th:poweradapCV}
RPQC with adaptive continuous measurements can encode \textsf{BQP}-complete computations.
\end{theorem}

\begin{proof}

Bosonic quantum computations taking as input (approximate) GKP and Gaussian states, with Gaussian gates and adaptive Gaussian measurements are capable of universal $\textsf{BQP}$ computations~\cite{Gottesman2001,baragiola2019all}. From the previous section, such computations are instances of RPQC with continuous measurements, since approximate GKP states of finite stellar rank can be obtained using Lemma~\ref{lem:dense}. This shows that one may encode \textsf{BQP}-complete computations as RPQC with adaptive continuous measurements.

\end{proof}

\noindent Adaptive discrete measurements also enhance the computational power of RPQC to universal $\textsf{BQP}$ computations:

\begin{theorem}\label{th:poweradapDV}
RPQC with adaptive discrete measurements can encode \textsf{BQP}-complete computations.
\end{theorem}

\begin{proof}

From the previous section, Boson Sampling are a subset of RPQC with discrete measurements. Moreover, it is shown in~\cite{Aaronson2013} that the equality $\textsf{BosonP}_\textsf{adap}=\textsf{BQP}$ holds, based on the Knill--Laflamme--Milburn scheme~\cite{knill2001scheme}, where $\textsf{BosonP}_\textsf{adap}$ is the class of problems that can be solved in polynomial time using Boson Sampling computations with adaptive Fock basis measurements. This shows that one may encode \textsf{BQP}-complete computations as RPQC with adaptive discrete measurements.

\end{proof}

\noindent Note that the model of CV cluster state quantum computing~\cite{menicucci2006universal,zhang2006continuous} is also a subclass of RPQC with adaptive discrete measurements, where the input has a Gaussian stellar function instead. 

Summarizing the results obtained in this section, the stellar representation provides a holomorphic description of bosonic quantum computations, and the subclass of RPQC capture existing subuniversal and universal bosonic quantum computing schemes, with non-adaptive and adaptive (discrete and continuous) measurements, respectively.

%------------------------------------------------------------------------

\section*{Acknowledgements}

We thank the anonymous referees for their comments. UC acknowledges inspiring discussions with T.\ Vidick, J.\ Preskill, S.\ Ghazi Nezami, P.-E.\ Emeriau, R.\ I.\ Booth, F.\ Arzani, G.\ Ferrini, D.\ Markham, F.\ Grosshans, and M.\ Walschaers. SM acknowledges insightful discussions with J.\ Preskill, T.\ Vidick, J.\ Slote and S.\ Ghazi Nezami. SM gratefully acknowledges the hospitality of the Simons Institute for the theory of computing during Spring 2020 where several basic aspects of this project were motivated. UC and SM acknowledge funding provided by the Institute for Quantum Information and Matter, an NSF Physics Frontiers Center (NSF Grant PHY-1733907).

%------------------------------------------------------------------------

\bibliographystyle{linksen}
\bibliography{refs}

%------------------------------------------------------------------------

\newpage
\appendix

%------------------------------------------------------------------------

\begin{center}
    {\huge Appendix}
\end{center}

\section{Multi-index notations}
\label{app:multi-index}

For $m\in\mathbb N^*$, $\bm p=(p_1,\dots,p_m)\in\mathbb N^m$, $\bm q=(q_1,\dots,q_m)\in\mathbb N^m$, and $\bm z=(z_1,\dots,z_m)\in\mathbb C^m$, we write:
\begin{equation}
    \begin{aligned}
    \bm0&=(0,\dots,0)\\
    \bm p+\bm q&=(p_1+q_1,\dots,p_m+q_m)\\
    |\bm p|&=p_1+\dots+p_m\\
    \bm p!&=p_1!\dots p_m!\\
    \ket{\bm p}&=\ket{p_1}\otimes\dots\otimes\ket{p_m}\\
    \bm p\le\bm q&\Leftrightarrow\forall i\in\{1,\dots,m\},\;p_i\le q_i\\
    \bm z^*&=(z_1^*,\dots,z_m^*)\\
    -\bm z&=(-z_1,\dots,-z_m)\\
    |\bm z|^2&=|z_1|^2+\dots+|z_m|^2\\
    \bm z^{\bm p}&=z_1^{p_1}\cdots z_m^{p_m}\\
    d^{2m}\bm z&=d\Re{(z_1)}d\Im{(z_1)}\dots d\Re{(z_m)}d\Im{(z_m)}\\
    \partial_{\bm z}&=(\partial_{z_1},\dots,\partial_{z_m}).
    \end{aligned}
\end{equation}
%

%------------------------------------------------------------------------

\section{Resolution of the classical Calogero--Moser model}
\label{app:CM}

In this section, we briefly review the resolution of the classical Calogero--Moser model~\cite{calogero1971solution} by the Olshanetsky--Perelomov projection method~\cite{ol1976geodesic}. 
The $n$-body Calogero--Moser Hamiltonian is given by:
\begin{equation}
    H^{CM}=\frac12\sum_{k=1}^n\left(p_k^2+\omega^2q_k^2\right)+\frac12g^2\sum_{k=1}^n\sum_{j\neq k}\frac1{(q_k-q_j)^2}.
\end{equation}
We treat here the case $\omega=0$ (isolated Calogero--Moser model) and refer the reader to~\cite{olshanetsky1981classical} for the general case. 
The equations of motions are given by Hamilton's equations:
\begin{equation}\label{eq:motionCM}
    \frac{dq_k(t)}{dt}=p_k\quad\text{and}\quad\frac{dp_k(t)}{dt}=2g^2\sum_{j\neq k}\frac1{(q_k-q_j)^3},\quad\forall k\in\{1,\dots,n\}.
\end{equation}
These equations of motions may be solved by writing them as the projection of a free motion in a higher-dimensional space.
More precisely, consider the free motion of a hermitian $n\times n$ matrix $\Lambda$, described by the equation
\begin{equation}
    \frac{d^2\Lambda}{dt^2}=0.
\end{equation}
The solution is given by $\Lambda(t)=At+B$, where $A$ and $B$ are hermitian matrices. Because $\Lambda$ is hermitian, we may write
\begin{equation}
    \Lambda(t)=U(t)Q(t)U^\dag(t),
\end{equation}
where $U$ is unitary and $Q$ is diagonal. Differentiating this equation, we obtain
\begin{equation}\label{eq:dXdtA}
    \frac{d\Lambda(t)}{dt}=\frac{dU(t)}{dt}Q(t)U^\dag(t)+U(t)\frac{dQ(t)}{dt}U^\dag(t)+U(t)Q(t)\frac{dU^\dag(t)}{dt}=A.
\end{equation}
Now let
\begin{equation}\label{eq:ML}
    \begin{aligned}
        M(t)&:=-iU^\dag(t)\frac{dU(t)}{dt}\\
        L(t)&:=\frac{dQ(t)}{dt}+[M(t),Q(t)].
    \end{aligned}
\end{equation}
We show that these matrices form a Lax pair, i.e., they satisfy the equation $i\frac{dL}{dt}=[M,L]$. Since $U$ is unitary, we have $U^\dag(t)U(t)=\mathbb I$, and thus $\frac{dU^\dag(t)}{dt}U(t)=-U^\dag(t)\frac{dU(t)}{dt}$. Hence,
\begin{equation}\label{eq:ULUdag}
    \begin{aligned}
        U(t)L(t)U^\dag(t)&=U(t)\frac{dQ(t)}{dt}U^\dag(t)+iU(t)M(t)Q(t)U^\dag(t)-iU(t)Q(t)M(t)U^\dag(t)\\
        &=U(t)\frac{dQ(t)}{dt}U^\dag(t)+\frac{dU(t)}{dt}Q(t)U^\dag(t)-U(t)Q(t)U^\dag(t)\frac{dU(t)}{dt}U^\dag(t)\\
        &=U(t)\frac{dQ(t)}{dt}U^\dag(t)+\frac{dU(t)}{dt}Q(t)U^\dag(t)+U(t)Q(t)\frac{dU^\dag(t)}{dt}\\
        &=A,
    \end{aligned}
\end{equation}
where we used Eq.~(\ref{eq:dXdtA}) in the last line. As a consequence, $L(t)=U^\dag(t)AU(t)$, so the eigenvalues of $L$ are conserved quantities of the evolution. We set $U(0)=\mathbb I$ so that $L(0)=A$, $\Lambda(0)=B=Q(0)$, and thus $\Lambda(t)=L(0)t+Q(0)$. Differentiating again Eq.~(\ref{eq:ULUdag}) yields
\begin{equation}\label{eq:LaxPair}
    \begin{aligned}
        \frac{dU(t)}{dt}L(t)U^\dag(t)&+U(t)\frac{dL(t)}{dt}U^\dag(t)+U(t)L(t)\frac{dU^\dag(t)}{dt}=0\\
       &\Leftrightarrow\frac{dL(t)}{dt}+U^\dag(t)\frac{dU(t)}{dt}L(t)+L(t)\frac{dU^\dag(t)}{dt}U(t)=0\\
        &\Leftrightarrow\frac{dL(t)}{dt}+\left[U^\dag(t)\frac{dU(t)}{dt},L(t)\right]=0\\
        &\Leftrightarrow\frac{dL(t)}{dt}+i[M(t),L(t)]=0,
    \end{aligned}
\end{equation}
where we used $\frac{dU^\dag(t)}{dt}U(t)=-U^\dag(t)\frac{dU(t)}{dt}$ in the third line. Hence the matrices $L$ and $M$ indeed form a Lax pair, and the eigenvalues of $L$ are conserved quantities.

Now setting $Q_{jk}=\delta_{jk}q_j$, $M_{jk}=g\big[\delta_{jk}\sum_{l\neq j}\frac1{(q_j-q_l)^2}-(1-\delta_{jk})\frac1{(q_j-q_k)^2}\big]$, and $L_{jk}=\delta_{jk}p_j+(1-\delta_{jk})\frac{ig}{q_j-q_k}$, where $\delta$ is the Kronecker symbol, ensures that Eq.~(\ref{eq:ML}) is satisfied and that Eq.~(\ref{eq:LaxPair}) is equivalent to Eq.~(\ref{eq:motionCM})~\cite{moser1976three}. By construction, the canonical variables $q_k$ are the eigenvalues of the matrix $\Lambda(t)=L(0)t+Q(0)$ given by:
\begin{equation}
    \Lambda_{jk}(t)=\delta_{jk}q_j(0)+\delta_{jk}p_j(0)t+(1-\delta_{jk})\frac{igt}{q_j-q_k}.
\end{equation}
Moreover, the eigenvalues of the matrix $L$ are conserved quantities of the evolution. A direct consequence is that the scattering process can be described simply by a permutation of the trajectories. This was proved, e.g., in~\cite{olshanetsky1981classical} in the case of particles on a line, and we reproduce the argument below in the case where $q_k\in\mathbb C$: due to the repulsive potential, at large times the particles follow a free motion. Writing
\begin{equation}
    q_k(t)\sim p_k^{\pm}t+q_k^{\pm},\quad t\rightarrow\pm\infty,
\end{equation}
the asymptotic momenta $p_k^{\pm}=p_k(\pm\infty)$ are the eigenvalues of $L(\pm\infty)$ and are thus conserved. Hence, there exists a permutation $\sigma\in\mathcal S_n$ with permutation matrix $P_\sigma$ such that $\text{Diag}(p_1^+,\dots,p_n^+)=P_\sigma\text{Diag}(p_1^-,\dots,p_n^-)P_\sigma^T$, i.e., $L(+\infty)=P_\sigma L(-\infty)P_\sigma^T$. With Eq.~(\ref{eq:ULUdag}) we obtain $P_\sigma=U^\dag(+\infty)U(-\infty)$.
Now we have:
\begin{equation}
    \begin{aligned}
        Q(t)&=U^\dag(t)\Lambda(t)U(t)\\
        &=U^\dag(t)(At+B)U(t)\\
        &=L(t)t+U^\dag(t)BU(t)\\
        &\sim\text{Diag}_k(p_k^{\pm}t+q_k^{\pm}),\quad t\rightarrow\pm\infty.
    \end{aligned}
\end{equation}
Since $L_{jk}(t)=\delta_{jk}p_j(t)+(1-\delta_{jk})\frac{ig}{q_j(t)-q_k(t)}$ we have $L(t)\sim\text{Diag}(p_1^{\pm},\dots,p_n^{\pm})$ when $t\rightarrow\pm\infty$, and thus $\text{Diag}(q_1^{\pm},\dots,q_n^{\pm})=U^\dag(\pm\infty)BU(\pm\infty)$. Hence, $\text{Diag}(q_1^+,\dots,q_n^+)=P_\sigma\text{Diag}(q_1^-,\dots,q_n^-)P_\sigma^T$. This shows that
\begin{equation}
    \text{Diag}_k(p_k^+t+q_k^+)=P_\sigma\text{Diag}_k(p_k^-t+q_k^-)P_\sigma^T,
\end{equation}
that is, the scattering process $(q_1^-,p_1^-),\dots,(q_n^-,p_n^-)\mapsto(q_1^+,p_1^+),\dots,(q_n^+,p_n^+)$ reduces to a permutation $\sigma$ of the trajectories, with asymptotic momenta $p_k^{\pm}$ and asymptotic offsets $q_k^{\pm}$. 

We can extend this analysis to relate the asymptotic parameters to the initial conditions $q_k(0)$ and $p_k(0)$: from
\begin{equation}\label{eq:diagqpmq0}
    \begin{aligned}
        \text{Diag}(q_1^{\pm},\dots,q_n^{\pm})&=U^\dag(\pm\infty)BU(\pm\infty)\\
        &=U^\dag(\pm\infty)Q(0)U(\pm\infty)\\
        &=U^\dag(\pm\infty)\text{Diag}(q_1(0),\dots,q_n(0))U(\pm\infty),
    \end{aligned}
\end{equation}
we deduce that the asymptotic offsets $q_k^{\pm}$ are given by permutations of the initial positions $q_k(0)$. For the asymptotic momenta, we have
\begin{equation}\label{eq:diagppmq0p0}
    \begin{aligned}
        \text{Diag}(p_1^{\pm},\dots,p_n^{\pm})&=L(\pm\infty)\\
        &=U^\dag(\pm\infty)L(0)U(\pm\infty),
    \end{aligned}
\end{equation}
so the asymptotic momenta $p_k^\pm$ are given by permutations of the eigenvalues of the matrix
\begin{equation}
    L(0)=\left(\delta_{jk}p_j(0)+(1-\delta_{jk})\frac{ig}{q_j(0)-q_k(0)}\right)_{1\le j,k\le n}.
\end{equation}
%

%------------------------------------------------------------------------

\section{Gaussian dynamics of the stellar representation}
\label{app:dynamics}

In this section we detail the proofs of the results obtained in Sec.~\ref{sec:SMdynamics}.

\subsection{Decoupling of the zeros and Gaussian parameters evolutions}
\label{app:decoupling}

We first give the proof of Theorem~\ref{th:decoupling}, which we recall below. From Sec.~\ref{sec:SMdynamics}, the evolution of the stellar function is governed by
\begin{equation}\label{eq:schrostellarapp}
    i\partial_tF^\star(z,t)=H(z,\partial_z)F^\star(z,t),
\end{equation}
for all $z\in\mathbb C$ and all $t\ge0$,
and a generic Gaussian Hamiltonian satisfies:
\begin{equation}\label{eq:iHGstellarapp}
    -iH_{\alpha,\xi,\varphi}(z,\partial_z)=\frac 12\xi z^2-\frac 12\xi^*\partial^2_z+i\varphi z\partial_z+\alpha z-\alpha^*\partial_z,
\end{equation}
for $\alpha,\xi\in\mathbb C$ and $\varphi\in\mathbb R$.

\setcounter{theorem}{0}

\begin{theorem}\label{appth:decoupling}
Let $F^\star(z,t)=\prod_{k=1}^n(z-\lambda_k(t))e^{-\frac12a(t)z^2+b(t)z+c(t)}$. Assuming that the initial zeros $\{\lambda_k(0)\}_{k=1\dots n}$ are simple, the evolution of the stellar function under a Gaussian Hamiltonian $\hat H=H_{\alpha,\xi,\varphi}(\hat a^\dag,\hat a)$ may be recast as the following dynamical system:
\begin{equation}\label{eq:dynasystapp}
    \begin{cases}
        \frac{da(t)}{dt}=\xi^*a^2(t)+2i\varphi a(t)-\xi,\\
        \frac{db(t)}{dt}=(i\varphi+\xi^*a(t))b(t)+\alpha+\alpha^*a(t),\\
        \frac{dc(t)}{dt}=\frac12\xi^*a(t)-\frac12\xi b^2(t)-\alpha^*b(t)+n(\xi^*a(t)+i\varphi)\\
        \frac{d^2\lambda_k(t)}{dt^2}=(|\xi|^2-\varphi^2)\lambda_k(t)+(\xi^*\alpha-i\varphi\alpha^*)-2\xi^{*2}\sum_{j\neq k}\frac1{(\lambda_k(t)-\lambda_j(t))^3},\quad\forall k\in\{1,\dots,n\}.
    \end{cases}
\end{equation}
The equations for the Gaussian parameters and the zeros are decoupled, and the relation between the root system $\{\lambda_k\}_{k=1\dots n}$ and the Gaussian parameters $a,b,c$ is specified by the initial conditions:
\begin{equation}
    \frac{d\lambda_k(t)}{dt}\bigg\vert_{t=0}=-\left(\xi^*a(0)+i\varphi\right)\lambda_k(0)+\xi^* b(0)+\alpha^*+\xi^*\sum_{j\neq k}\frac1{\lambda_k(0)-\lambda_j(0)},\quad\forall k\in\{1,\dots,n\}.
\end{equation}
\end{theorem}

\begin{proof}

Combining Eqs. (\ref{eq:schrostellarapp}) and (\ref{eq:iHGstellarapp}) the evolution of the stellar function is governed by the partial differential equation:
\begin{equation}
    \partial_t[F^\star(z,t)]=\left(\frac 12\xi z^2-\frac 12\xi^*\partial^2_z+i\varphi z\partial_z+\alpha z-\alpha^*\partial_z\right)[F^\star(z,t)].
\end{equation}
Moreover, we have:
\begin{equation}\label{partialdstellar}
    \begin{aligned}
        F^\star(z,t)&=P(z,t)G(z,t)=\prod_{k=1}^n(z-\lambda_k(t))e^{-\frac12a(t)z^2+b(t)z+c(t)}\\
        \partial_tF^\star(z,t)&=F^\star(z,t)\left[-\frac12z^2\frac{da(t)}{dt}+z\frac{db(t)}{dt}+\frac{dc(t)}{dt}-\sum_{k=1}^n\frac1{z-\lambda_k(t)}\frac{d\lambda_k(t)}{dt}\right]\\
        \partial_zF^\star(z,t)&=F^\star(z,t)\left[-a(t)z+b(t)+\sum_{k=1}^n\frac1{z-\lambda_k(t)}\right]\\
        \partial_z^2F^\star(z,t)&=F^\star(z,t)\left[\left(-a(t)z+b(t)+\sum_{k=1}^n\frac1{z-\lambda_k(t)}\right)^2-a(t)-\sum_{k=1}^n\frac1{(z-\lambda_k(t))^2}\right].
    \end{aligned}
\end{equation}
Plugging these expressions in Eq.~(\ref{eq:schrostellarapp}) and simplifying the Gaussian functions on each side, we obtain the partial differential equation:
\begin{equation}\label{generalEDsinglemodeG}
    \begin{aligned}
        &\left[-\frac12z^2\frac{da(t)}{dt}+z\frac{db(t)}{dt}+\frac{dc(t)}{dt}-\sum_{k=1}^n\frac1{z-\lambda_k(t)}\frac{d\lambda_k(t)}{dt}\right]P(z,t)\\
        &\quad\quad\quad=\Bigg[\left(\frac 12\xi-\frac12\xi^* a^2(t)-i\varphi a(t)\right)z^2\\
        &\quad\quad\quad\quad+\left(\xi^*a(t)b(t)+i\varphi b(t)+\alpha+\alpha^*a(t)\right)z\\
        &\quad\quad\quad\quad+\left(\frac12\xi^*a(t)-\frac12\xi b^2(t)-\alpha^*b(t)\right)\\
        &\quad\quad\quad\quad+\sum_{k=1}^n\frac1{z-\lambda_k(t)}\left(\left(\xi^*a(t)+i\varphi\right)z-\xi^*b(t)-\alpha^*-\frac12\xi^*\sum_{j\neq k}\frac1{z-\lambda_j(t)}\right)\Bigg]P(z,t).
    \end{aligned}
\end{equation}
This equation is an equality between two polynomials in $z$ with time-dependent coefficients. The leading coefficients give
\begin{equation}\label{eq:ODEa}
    \frac{da(t)}{dt}=\xi^*a^2(t)+2i\varphi a(t)-\xi.
\end{equation}
Removing these coefficients on both sides the differential equation~(\ref{generalEDsinglemodeG}) rewrites
    \begin{align}\label{generalEDsinglemodeG1}
        \nonumber&\left[z\frac{db(t)}{dt}+\frac{dc(t)}{dt}-\sum_{k=1}^n\frac1{z-\lambda_k(t)}\frac{d\lambda_k(t)}{dt}\right]P(z,t)\\
        \nonumber&\quad\quad\quad=\Bigg[\left(\xi^*a(t)b(t)+i\varphi b(t)+\alpha+\alpha^*a(t)\right)z\\
        &\quad\quad\quad\quad+\left(\frac12\xi^*a(t)-\frac12\xi b^2(t)-\alpha^*b(t)\right)\displaybreak\\
        \nonumber&\quad\quad\quad\quad+\sum_{k=1}^n\frac1{z-\lambda_k(t)}\left(\left(\xi^*a(t)+i\varphi\right)z-\xi^*b(t)-\alpha^*-\frac12\xi^*\sum_{j\neq k}\frac1{z-\lambda_j(t)}\right)\Bigg]P(z,t),
\end{align}
and the leading coefficients give
\begin{equation}\label{eq:ODEb}
    \frac{db(t)}{dt}=(i\varphi+\xi^*a(t))b(t)+\alpha+\alpha^*a(t).
\end{equation}
Once again, removing these coefficients on both sides, the differential equation~(\ref{generalEDsinglemodeG1}) rewrites
\begin{equation}\label{generalEDsinglemodeG2}
    \begin{aligned}
        &\left[\frac{dc(t)}{dt}-\sum_{k=1}^n\frac1{z-\lambda_k(t)}\frac{d\lambda_k(t)}{dt}\right]P(z,t)\\
        &\quad\quad\quad=\Bigg[\left(\frac12\xi^*a(t)-\frac12\xi b^2(t)-\alpha^*b(t)\right)\\
        &\quad\quad\quad\quad+\sum_{k=1}^n\frac1{z-\lambda_k(t)}\left(\left(\xi^*a(t)+i\varphi\right)z-\xi^*b(t)-\alpha^*-\frac12\xi^*\sum_{j\neq k}\frac1{z-\lambda_j(t)}\right)\Bigg]P(z,t).
    \end{aligned}
\end{equation}
The leading coefficient on the left hand side is $\frac{dc(t)}{dt}$, while the leading coefficient on the right hand side is
\begin{equation}
    \left(\frac12\xi^*a(t)-\frac12\xi b^2(t)-\alpha^*b(t)\right)+\sum_{k=1}^n\left(\xi^*a(t)+i\varphi\right).
\end{equation}
We thus obtain
\begin{equation}
    \frac{dc(t)}{dt}=\frac12\xi^*a(t)-\frac12\xi b^2(t)-\alpha^*b(t)+n(\xi^*a(t)+i\varphi).
\end{equation}
Finally, removing $(\frac12\xi^*a(t)-\frac12\xi b^2(t)-\alpha^*b(t))z^n$ on both sides, the differential equation~(\ref{generalEDsinglemodeG2}) rewrites
\begin{equation}
    \begin{aligned}
        &\left[n(\xi^*a(t)+i\varphi)-\sum_{k=1}^n\frac1{z-\lambda_k(t)}\frac{d\lambda_k(t)}{dt}\right]P(z,t)\\
        &\quad\quad\quad=\Bigg[\sum_{k=1}^n\frac1{z-\lambda_k(t)}\left(\left(\xi^*a(t)+i\varphi\right)z-\xi^*b(t)-\alpha^*-\frac12\xi^*\sum_{j\neq k}\frac1{z-\lambda_j(t)}\right)\Bigg]P(z,t),
    \end{aligned}
\end{equation}
or equivalently:
\begin{equation}\label{eq:EDprods}
    \begin{aligned}
        &-n\left(\xi^*a(t)+i\varphi\right)\prod_{k=1}^n(z-\lambda_k(t))+\sum_{k=1}^n\frac{d\lambda_k(t)}{dt}\prod_{j\neq k}(z-\lambda_j(t))\\
        &\quad\quad\quad=\left[-\left(\xi^*a(t)+i\varphi\right)z+\xi^* b(t)+\alpha^*\right]\sum_{k=1}^n\prod_{j\neq k}(z-\lambda_j(t))+\frac12\xi^*\sum_{k=1}^n\sum_{j\neq k}\prod_{\substack{l\neq j\\l\neq k}}(z-\lambda_l(t)).
    \end{aligned}
\end{equation}
Let $k\in\{1,\dots,n\}$. Under the assumption that $\lambda_k(t)$ is a simple zero of $z\mapsto P(z,t)$, dividing the above equation by $\prod_{j\neq k}(z-\lambda_j(t))$ and setting $z=\lambda_k(t)$ yields
\begin{equation}\label{eq:ODElambdasingle}
        \frac{d\lambda_k(t)}{dt}=-\left(\xi^*a(t)+i\varphi\right)\lambda_k(t)+\xi^* b(t)+\alpha^*+\xi^*\sum_{j\neq k}\frac1{\lambda_k(t)-\lambda_j(t)}.
\end{equation}
In particular, for all $k\in\{1,\dots,n\}$ we have
\begin{equation}\label{eq:CI}
        \frac{d\lambda_k(t)}{dt}\bigg\vert_{t=0}=-\left(\xi^*a(0)+i\varphi\right)\lambda_k(0)+\xi^* b(0)+\alpha^*+\xi^*\sum_{j\neq k}\frac1{\lambda_k(0)-\lambda_j(0)}.
\end{equation}
At this point we have the system:
\begin{equation}
    \begin{cases}
        \frac{da(t)}{dt}=\xi^*a^2(t)+2i\varphi a(t)-\xi,\\
        \frac{db(t)}{dt}=(i\varphi+\xi^*a(t))b(t)+\alpha+\alpha^*a(t),\\
        \frac{dc(t)}{dt}=\frac12\xi^*a(t)-\frac12\xi b^2(t)-\alpha^*b(t)+n(\xi^*a(t)+i\varphi),\\
        \frac{d\lambda_k(t)}{dt}=-\left(\xi^*a(t)+i\varphi\right)\lambda_k(t)+\xi^* b(t)+\alpha^*+\xi^*\sum_{j\neq k}\frac1{\lambda_k(t)-\lambda_j(t)},\quad\forall k\in\{1,\dots,n\}.
    \end{cases}
\end{equation}
We are left with showing that this differential system indeed decouples by taking the derivative of the last equation:
\begin{equation}
    \begin{aligned}
        \frac{d^2\lambda_k(t)}{dt^2}&=-\xi^*\frac{da(t)}{dt}\lambda_k(t)-(\xi^* a(t)+i\varphi)\frac{d\lambda_k(t)}{dt}+\xi^*\frac{db(t)}{dt}+\xi^*\sum_{j\neq k}\frac{\frac{d\lambda_j(t)}{dt}-\frac{d\lambda_k(t)}{dt}}{(\lambda_k(t)-\lambda_j(t))^2}\\
        &\!\!\!\!\!\!\!\!\!\!\!\!\!\!\!\!\!\!\!\!=-\xi^*(\xi^*a^2(t)+2i\varphi-\xi)\lambda_k(t)-(\xi^*a(t)+i\varphi)\left[-(\xi^*a(t)+i\varphi)\lambda_k(t)+\xi^*b(t)+\alpha^*+\xi^*\sum_{j\neq k}\frac1{\lambda_k(t)-\lambda_j(t)}\right]\\
        &\!\!+\xi^*(i\varphi b(t)+\xi^*a(t)b(t)+\alpha+\alpha^*a(t))\\
        &\!\!+\xi^*\sum_{j\neq k}\frac1{(\lambda_k(t)-\lambda_j(t))^2}\left[(\xi^* a(t)+i\varphi)(\lambda_k(t)-\lambda_j(t))+\xi^*\sum_{l\neq j}\frac1{\lambda_j(t)-\lambda_l(t)}-\xi^*\sum_{l\neq k}\frac1{\lambda_k(t)-\lambda_l(t)}\right]\\
        &\!\!\!\!\!\!\!\!\!\!\!\!\!\!\!\!\!\!\!\!=(|\xi|^2\!-\!\varphi^2)\lambda_k(t)\!+\!(\xi^*\alpha\!-\!i\varphi\alpha^*)+\xi^{*2}\sum_{j\neq k}\frac1{(\lambda_k(t)-\lambda_j(t))^2}\!\!\left[\frac2{\lambda_j(t)\!-\!\lambda_k(t)}\!+\!\sum_{\substack{l\neq j\\l\neq k}}\frac{\lambda_k(t)\!-\!\lambda_j(t)}{(\lambda_j(t)\!-\!\lambda_l(t))(\lambda_k(t)\!-\!\lambda_l(t))}\right]\\
        &\!\!\!\!\!\!\!\!\!\!\!\!\!\!\!\!\!\!\!\!=(|\xi|^2\!-\!\varphi^2)\lambda_k(t)\!+\!(\xi^*\alpha\!-\!i\varphi\alpha^*)-2\xi^{*2}\sum_{j\neq k}\frac1{(\lambda_k(t)\!-\!\lambda_j(t))^3}\!+\!\sum_{\substack{j,l\neq k\\j\neq l}}\frac1{(\lambda_k(t)\!-\!\lambda_l(t))(\lambda_k(t)\!-\!\lambda_j(t))(\lambda_j(t)\!-\!\lambda_l(t))}\\
        &\!\!\!\!\!\!\!\!\!\!\!\!\!\!\!\!\!\!\!\!=(|\xi|^2\!-\!\varphi^2)\lambda_k(t)\!+\!(\xi^*\alpha\!-\!i\varphi\alpha^*)-2\xi^{*2}\sum_{j\neq k}\frac1{(\lambda_k(t)-\lambda_j(t))^3},
    \end{aligned}
\end{equation}
for all $k\in\{1,\dots,n\}$.

\end{proof}

\subsection{Displacement, phase-shift, shearing and squeezing evolutions}
\label{app:dynamicsDRPS}
 
We first prove Lemma~\ref{lem:evoD} from the main text which we recall below:

\setcounter{lemma}{0}

\begin{lemma}
Let $F^\star(z,t)=\prod_{k=1}^n(z-\lambda_k(t))e^{-\frac12a(t)z^2+b(t)z+c(t)}$ and let $\alpha\in\mathbb C$. The evolution under the displacement Hamiltonian $\hat H_\alpha^D$ is given by
\begin{equation}
    \begin{cases}
        a(t)=a(0),\\
        b(t)=(\alpha+\alpha^*a(0))t+b(0),\\
        c(t)=\frac12(\alpha^{*2}a(0)-|\alpha|^2)t^2+b(0)t+c(0),\\
        \lambda_k(t)=\alpha^*t+\lambda(0),\quad\forall k\in\{1,\dots,n\}.
    \end{cases}
\end{equation}
\end{lemma}

\begin{proof}

Setting $\xi=\varphi=0$ in Theorem~\ref{th:decoupling} yields:
\begin{equation}\label{eq:dynasystD}
    \begin{cases}
        \frac{da(t)}{dt}=0,\\
        \frac{db(t)}{dt}=\alpha+\alpha^*a(t),\\
        \frac{dc(t)}{dt}=-\alpha^*b(t)\\
        \frac{d^2\lambda_k(t)}{dt^2}=0,\quad\forall k\in\{1,\dots,n\},
    \end{cases}
\end{equation}
with the initial conditions
\begin{equation}\label{eq:CID}
        \frac{d\lambda_k(t)}{dt}\bigg\vert_{t=0}=\alpha^*,\quad\forall k\in\{1,\dots,n\}.
\end{equation}
Integrating the first equation of Eq.~(\ref{eq:dynasystD}) gives $a(t)=a(0)$. Plugging this in the second equation and integrating gives $b(t)=(\alpha+\alpha^*a(0))t+b(0)$. With the third equation we obtain $c(t)=\frac12(\alpha^{*2}a(0)-|\alpha|^2)t^2+b(0)t+c(0)$. Finally, the fourth equation together with the initial conditions gives $\lambda_k(t)=\alpha^*t+\lambda(0)$ for all $k\in\{1,\dots,n\}$.

\end{proof}

\noindent We now prove Lemma~\ref{lem:evoR} from the main text which we recall below:

\begin{lemma}
Let $F^\star(z,t)=\prod_{k=1}^n(z-\lambda_k(t))e^{-\frac12a(t)z^2+b(t)z+c(t)}$ and let $\varphi\in[0,2\pi]$. The evolution under the phase-shift Hamiltonian $\hat H_\varphi^R$ is given by
\begin{equation}
    \begin{cases}
        a(t)=e^{2i\varphi t}a(0),\\
        b(t)=e^{i\varphi t}b(0),\\
        c(t)=e^{in\varphi t}c(0),\\
        \lambda_k(t)=e^{-i\varphi t}\lambda(0),\quad\forall k\in\{1,\dots,n\}.
    \end{cases}
\end{equation}
\end{lemma}

\begin{proof}

Setting $\alpha=\xi=0$ in Theorem~\ref{th:decoupling} yields:
\begin{equation}\label{eq:dynasystR}
    \begin{cases}
        \frac{da(t)}{dt}=2i\varphi a(t),\\
        \frac{db(t)}{dt}=i\varphi b(t),\\
        \frac{dc(t)}{dt}=in\varphi\\
        \frac{d^2\lambda_k(t)}{dt^2}=-\varphi^2\lambda_k(t),\quad\forall k\in\{1,\dots,n\},
    \end{cases}
\end{equation}
with the initial conditions
\begin{equation}\label{eq:CIR}
        \frac{d\lambda_k(t)}{dt}\bigg\vert_{t=0}=-i\varphi\lambda_k(0),\quad\forall k\in\{1,\dots,n\}.
\end{equation}
We obtain $a(t)=e^{2i\varphi t}a(0)$, $b(t)=e^{i\varphi t}b(0)$, $c(t)=e^{in\varphi t}c(0)$ and $\lambda_k(t)=e^{-i\varphi t}\lambda(0)$ for all $k\in\{1,\dots,n\}$.

\end{proof}

\noindent We now prove Lemma~\ref{lem:evoP} from the main text which we recall below:

\begin{lemma}
Let $F^\star(z,t)=\prod_{k=1}^n(z-\lambda_k(t))e^{-\frac12a(t)z^2+b(t)z+c(t)}$ and let $s\in\mathbb R$. The evolution under the shearing Hamiltonian $\hat H_s^P$ is given by
\begin{equation}
    \begin{cases}
        a(t)=\frac{a(0)-ist(1-a(0))}{1-ist(1-a(0))},\\
        b(t)=\frac{b(0)}{1-ist(1-a(0))},\\
        c(t)=c(0)-\frac{ist}2-\left(n+\frac12\right)\log(1-ist(1-a(0)))-\frac{b^2(0)}{2(1-a(0))(1-ist(1-a(0))},
    \end{cases}
\end{equation}
and the zeros $\lambda_k(t)$ are the eigenvalues of the matrix $\Lambda(t)$ defined as
\begin{equation}
\Lambda_{kl}(t)=\begin{cases}\lambda_k(0)-ist\left[(1-a(0))\lambda_k(0)+b(0)+\sum_{j\neq k}\frac1{\lambda_k(0)-\lambda_j(0)}\right],\quad&k=l\\\frac{ist}{\lambda_l(0)-\lambda_k(0)},\quad&k\neq l.\end{cases}
\end{equation}
\end{lemma} 

\begin{proof}

Recall that $\hat H_s^P=\hat H_{is}^S+\hat H_s^R-\frac s2\hat 1$. Setting $\alpha=0$, $\varphi=s$ and $\xi=is$ in Theorem~\ref{th:decoupling} yields:
\begin{equation}\label{eq:dynasystP}
    \begin{cases}
        \frac{da(t)}{dt}=-is(1-a(t))^2,\\
        \frac{db(t)}{dt}=is(1-a(t))b(t),\\
        \frac{dc(t)}{dt}=-\frac{is}2a(t)-\frac{is}2b^2(t)+ins(1-a(t))\\
        \frac{d^2\lambda_k(t)}{dt^2}=2s\sum_{j\neq k}\frac1{(\lambda_k(t)-\lambda_j(t))^3},\quad\forall k\in\{1,\dots,n\},
    \end{cases}
\end{equation}
with the initial conditions
\begin{equation}\label{eq:CIP}
        \frac{d\lambda_k(t)}{dt}\bigg\vert_{t=0}=is(a(0)-1)\lambda_k(0)-is b(0)-is\sum_{j\neq k}\frac1{\lambda_k(0)-\lambda_j(0)},\quad\forall k\in\{1,\dots,n\}.
\end{equation}
The first differential equation in Eq.~(\ref{eq:dynasystP}) is separable and can be solved directly:
\begin{equation}
    a(t)=\frac{a(0)-ist(1-a(0))}{1-ist(1-a(0))}.
\end{equation}
Plugging this expression in the second differential equation gives another separable equation which yields:
\begin{equation}
        b(t)=\frac{b(0)}{1-ist(1-a(0))}.
\end{equation}
With these expressions for $a(t)$ and $b(t)$, integrating the third differential equation gives
\begin{equation}
        c(t)=c(0)-\frac{ist}2-\left(n+\frac12\right)\log(1-ist(1-a(0)))-\frac{b^2(0)}{2(1-a(0))(1-ist(1-a(0))}.
\end{equation}
Finally, the remaining dynamical system for the zeros is the complex version of the Calogero--Moser system~\cite{calogero1976exactly}, which may be solved using the Olshanetsky--Perelomov projection method (see Appendix~\ref{app:CM}). We obtain that the zeros $\{\lambda_k\}_{k=1,\dots,n}$ are given by the eigenvalues of a matrix $\Lambda(t)$ with analytical time-dependent coefficients defined as:
\begin{equation}
    \begin{aligned}
        \Lambda_{kl}(t)&:=\begin{cases}\lambda_k(0)+\frac{d\lambda(0)}{dt}t,\quad &k=l\\\frac{ist}{\lambda_l(0)-\lambda_k(0)},\quad&k\neq l\end{cases}\\
        &=\begin{cases}\lambda_k(0)-ist\left[(1-a(0))\lambda_k(0)+b(0)+\sum_{j\neq k}\frac1{\lambda_k(0)-\lambda_j(0)}\right],\quad&k=l\\\frac{ist}{\lambda_l(0)-\lambda_k(0)},\quad&k\neq l,\end{cases}
    \end{aligned}
\end{equation}
where we used Eq.~(\ref{eq:CIP}) in the second line. 
\end{proof}

\noindent For completeness, we also derive the evolution under squeezing:

\begin{lemma}\label{lemmapp:evoS}
Let $F^\star(z,t)=\prod_{k=1}^n(z-\lambda_k(t))e^{-\frac12a(t)z^2+b(t)z+c(t)}$ and let $\xi=re^{i\theta}\in\mathbb C$. Let also $A:=\tanh^{-1}(e^{-i\theta}a(0))$. The evolution under the squeezing Hamiltonian $\hat H_\xi^S$ is given by
\begin{equation}
    \begin{cases}
        a(t)=e^{i\theta}\tanh(rt+A),\\
        b(t)=\frac{\cosh(rt+A)}{\cosh(A)}b(0),\\
        c(t)=c(0)+\left(n+\frac12\right)\log\left(\frac{\cosh(rt+A)}{\cosh(A)}\right)-\frac{e^{i\theta}b^2(0)}{4\cosh^2(A)}\left[rt+\sinh(rt)\cosh(rt+2A)\right],
    \end{cases}
\end{equation}
and the zeros $\lambda_k(t)$ are the eigenvalues of the matrix $\Lambda(t)$ defined as
\begin{equation}
    \Lambda_{kl}(t)=\begin{cases}\lambda_k(0)\left[\cosh(rt)+ie^{-i\theta}a(0)\sinh(rt)\right]-ie^{-i\theta}b(0)\sinh(rt)-ie^{-i\theta}\sum_{j\neq k}\frac{\sinh(rt)}{\lambda_k(0)-\lambda_j(0)},\quad&k=l\\\frac{ie^{-i\theta}\sinh(rt)}{\lambda_l(0)-\lambda_k(0)},\quad&k\neq l.\end{cases}
\end{equation}
\end{lemma}

\begin{proof}

Setting $\alpha=\varphi=0$ in Theorem~\ref{th:decoupling} yields:
\begin{equation}\label{eq:dynasystS}
    \begin{cases}
        \frac{da(t)}{dt}=\xi^*a^2(t)-\xi,\\
        \frac{db(t)}{dt}=\xi^*a(t)b(t),\\
        \frac{dc(t)}{dt}=(n+\frac12)\xi^*a(t)-\frac12\xi b^2(t)\\
        \frac{d^2\lambda_k(t)}{dt^2}=|\xi|^2\lambda_k(t)-2\xi^{*2}\sum_{j\neq k}\frac1{(\lambda_k(t)-\lambda_j(t))^3},\quad\forall k\in\{1,\dots,n\},
    \end{cases}
\end{equation}
with the initial conditions
\begin{equation}\label{eq:CIS}
        \frac{d\lambda_k(t)}{dt}\bigg\vert_{t=0}=-\xi^*a(0)\lambda_k(0)+\xi^* b(0)+\xi^*\sum_{j\neq k}\frac1{\lambda_k(0)-\lambda_j(0)},\quad\forall k\in\{1,\dots,n\}.
\end{equation}
The first differential equation in Eq.~(\ref{eq:dynasystS}) is separable and can be solved directly:
\begin{equation}\label{eq:aexpr}
    a(t)=e^{i\theta}\tanh\left[rt+\tanh^{-1}(e^{-i\theta}a(0))\right]=e^{i\theta}\tanh\left(rt+A\right),
\end{equation}
where we have used $\xi=re^{i\theta}$ and $A=\tanh^{-1}(e^{-i\theta}a(0))$.
Plugging this expression in the second differential equation gives another separable equation which yields:
\begin{equation}\label{eq:bexpr}
        b(t)=\frac{\cosh(rt+A)}{\cosh(A)}b(0).
\end{equation}
With these expressions for $a(t)$ and $b(t)$, integrating the third differential equation gives
\begin{equation}\label{eq:cexpr}
        c(t)=c(0)+\left(n+\frac12\right)\log\left(\frac{\cosh(rt+A)}{\cosh(A)}\right)-\frac{e^{i\theta}b^2(0)}{4\cosh^2(A)}\left[rt+\sinh(rt)\cosh(rt+2A)\right].
\end{equation}
Finally, the remaining dynamical system for the zeros is the complex version of the Calogero--Moser system~\cite{calogero1976exactly} with reverted harmonic potential, which may be solved using the Olshanetsky--Perelomov projection method~\cite{ol1976geodesic}. We obtain that the zeros $\{\lambda_k\}_{k=1,\dots,n}$ are given by the eigenvalues of a matrix $\Lambda(t)$ with analytical time-dependent coefficients defined as:
\begin{equation}
    \begin{aligned}
        \Lambda_{kl}(t)&=\begin{cases}\lambda_k(0)\cosh(rt)+\frac{d\lambda_k(0)}{dt}\frac{\sinh(rt)}{ir},\quad &k=l\\\frac{i\xi^*\sinh(rt)}{r(\lambda_l(0)-\lambda_k(0))},\quad&k\neq l\end{cases}\\
        &=\begin{cases}\lambda_k(0)\left[\cosh(rt)+ie^{-i\theta}a(0)\sinh(rt)\right]-ie^{-i\theta}b(0)\sinh(rt)-ie^{-i\theta}\sum_{j\neq k}\frac{\sinh(rt)}{\lambda_k(0)-\lambda_j(0)},\quad&k=l\\\frac{ie^{-i\theta}\sinh(rt)}{\lambda_l(0)-\lambda_k(0)},\quad&k\neq l,\end{cases}
    \end{aligned}
\end{equation}
where $\xi=re^{i\theta}$ and where we used Eq.~(\ref{eq:CIS}) in the second line.

\end{proof}

%------------------------------------------------------------------------

\section{Proofs of the properties of the multimode stellar hierarchy}
\label{app:proofsMM}

\setcounter{theorem}{2}

In this section we prove Theorem~\ref{th:invar}, Corollary~\ref{coro:nonincr}, and Theorem~\ref{th:robust} from Sec.~\ref{sec:multi}, which we recall below:

\begin{theorem}
Let $\hat U$ be a unitary operator over $m$ modes and let $\hat{\mathbb I}$ denote the identity operator over $m$ modes.
The following propositions are equivalent: 
    \begin{enumerate}
    \item\label{enum:1} $\hat U$ is a Gaussian unitary;
    \item\label{enum:2} $\hat U$ maps Gaussian states to Gaussian states;
    \item\label{enum:3} $\hat U$ and $\hat U\otimes\hat{\mathbb I}$ preserve the stellar rank of any quantum state.
    \end{enumerate}
\end{theorem}

\begin{proof} 

We prove \ref{enum:2} $\Rightarrow$ \ref{enum:1} $\Rightarrow$ \ref{enum:3} $\Rightarrow$ \ref{enum:2}.

\medskip

\ref{enum:2} $\Rightarrow$ \ref{enum:1}: this is a standard result, $\hat U$ is a Gaussian unitary if and only if it maps (pure and mixed) Gaussian states to Gaussian states~\cite{demoen1977completely,giedke2002characterization}.

\medskip

\ref{enum:1} $\Rightarrow$ \ref{enum:3}:
We show that Gaussian unitary operations leave the stellar rank invariant. We first consider the case of pure states. Since any $m$-mode Gaussian unitary operation $\hat G$ can be written as
\begin{equation}
    \hat G=\hat U\hat S\hat D\hat V,
\end{equation}
by the Euler decomposition in Eq.~(\ref{eq:decompGmulti}), where $\hat U$ and $\hat V$ are passive linear operations, $\hat S$ is a tensor product of single-mode squeezing operators, and $\hat D$ is a tensor product of single-mode displacement operators, it suffices to prove that these three types of Gaussian unitaries leave the stellar rank invariant.

We first consider the case of states of finite stellar rank. Let $n\in\mathbb N$ and let $\ket{\bm\psi}\in\mathcal H^{\otimes m}$ be a state of stellar rank $n\in\mathbb N$. 
By Lemma~\ref{lem:decomp1}, there exists a Gaussian state $\ket G$ such that $\ket{\bm\psi}\propto P(\hat{\bm a}^\dag)\ket G$, where $P$ is a multivariate polynomial of degree $n$. 
Moreover, for all $\bm\alpha\in\mathbb C^m$,
\begin{equation}
    \begin{aligned}
        \hat D(\bm\alpha)\ket{\bm\psi}&\propto\hat D(\bm\alpha)P(\hat{\bm a}^\dag)\ket G\\
        &\propto P(\hat{\bm a}^\dag+\bm\alpha^*)\hat D(\bm\alpha)\ket G,
    \end{aligned}
\end{equation}
where we used Eq.~(\ref{eq:commutDSR}). The state $\hat D(\bm\alpha)\ket G$ is a Gaussian state and $\bm z\mapsto P(\bm z+\bm\alpha^*)$ is a multivariate polynomial of degree $n$, so by Lemma~\ref{lem:decomp1}, the state $\hat D(\bm\alpha)\ket{\bm\psi}$ has stellar rank $n$.

Similarly, for any passive linear operation $\hat U$, with unitary matrix $U$ we have 
\begin{equation}
    \begin{aligned}
        \hat U\ket{\bm\psi}&\propto\hat UP(\hat{\bm a}^\dag)\ket G\\
        &\propto P(U\hat{\bm a}^\dag)\hat U\ket G,
    \end{aligned}
\end{equation}
where we used Eq.~(\ref{eq:actionU}). The state $\hat U\ket G$ is a Gaussian state and $\bm z\mapsto P(U\bm z)$ is a multivariate polynomial of degree $n$, since $U$ is unitary, so by Lemma~\ref{lem:decomp1}, the state $\hat U\ket{\bm\psi}$ has stellar rank $n$.

The case of squeezing is more technical. For all $\bm\xi=(\xi_1,\dots,\xi_m)\in\mathbb C^m$, $\hat S(\bm\xi)=\bigotimes_{i=1}^m\hat S(\xi_i)$, so it is sufficient prove the result for a single-mode squeezing of an $m$-mode state. Without loss of generality (up to a passive linear operation), we consider a squeezing on mode $1$. For all $\xi\in\mathbb C$,
\begin{equation}
    \begin{aligned}
        (\hat S(\xi)\otimes\mathbb I^{\otimes m-1})\ket{\bm\psi}&\propto(\hat S(\xi)\otimes\mathbb I^{\otimes m-1})P(\hat{\bm a}^\dag)\ket G\\
        &\propto P((\cosh r)\hat a_1^\dag-e^{-i\theta}(\sinh r)\hat a_1,\hat a_2^\dag,\dots,\hat a_m^\dag)(\hat S(\xi)\otimes\mathbb I^{\otimes m-1})\ket G,
    \end{aligned}
\end{equation}
where we used Eq.~(\ref{eq:commutDSR}) and where we have set $\xi:=re^{i\theta}$. With the correspondence $\hat a_k\leftrightarrow\partial_{z_k}$ and $\hat a_k^\dag\leftrightarrow z_k\times$, the stellar function of this state reads:
\begin{equation}
    P((\cosh r)z_1-e^{-i\theta}(\sinh r)\partial_{z_1},z_2,\dots,z_m)G(\bm z),
\end{equation}
where the Gaussian function $G(\bm z)$ is the stellar function of the Gaussian state $(\hat S(\xi)\otimes\mathbb I^{\otimes m-1})\ket G$. Computing the partial derivatives gives an expression of the form $\tilde P(\bm z)G(\bm z)$, and we are left with checking that the degree of $\tilde P$ is indeed equal to $n$. 

We write $P(\bm z):=\sum_{|\bm n|\le n}p_{\bm n}\bm z^{\bm n}$ and $G(\bm z):=e^{-\frac12\bm z^TA\bm z+B^T\bm z}$ with $|A_{ij}|<1$ for all $1\le i,j\le m$ by Eq.~(\ref{eq:ABGaussian}). 
The leading monomials of $\tilde P$ are given by the leading monomials of: 
\begin{equation}
    \begin{aligned}
        &G(\bm z)^{-1}\left(\sum_{\substack{\bm n=(n_1,\dots,n_m)\\ n_1+\dots+n_m=n}}p_{\bm n}((\cosh r)z_1-e^{-i\theta}(\sinh r)\partial_{z_1})^{n_1}z_2^{n_2}\cdots z_m^{n_m}\right)G(\bm z)\\
        &\quad\quad=G(\bm z)^{-1}\left(\sum_{\substack{\bm n=(n_1,\dots,n_m)\\ n_1+\dots+n_m=n}}p_{\bm n}((\cosh r)z_1-e^{-i\theta}(\sinh r)\partial_{z_1})^{n_1}z_2^{n_2}\cdots z_m^{n_m}\right)e^{-\frac12\sum_{ij}A_{ij}z_iz_j+\sum_k B_kz_k}.
    \end{aligned}
\end{equation}
Without loss of generality, we may discard the terms depending on $B$ in the exponentials, since they only lead to lower order terms. Hence, the leading monomials of $\tilde P$ are given by the leading monomials of:
\begin{equation}
    e^{\frac12\sum_{ij}A_{ij}z_iz_j}\sum_{\substack{\bm n=(n_1,\dots,n_m)\\ n_1+\dots+n_m=n}}p_{\bm n}\cosh^{n_1}r(z_1-t\partial_{z_1})^{n_1}z_2^{n_2}\cdots z_m^{n_m}e^{-\frac12\sum_{ij}A_{ij}z_iz_j},
\end{equation}
where we have set $t=e^{-i\theta}\tanh r$, with $|t|<1$.
Now for all $t,z\in\mathbb C$ and all $n\in\mathbb N$ we have
\begin{equation}
    (z-t\partial_z)^n=\sum_{j=0}^{\lfloor\frac n2\rfloor}\frac{(-1)^jn!t^j}{(n-2j)!j!2^j}\sum_{l=0}^{n-2j}\binom{n-2j}lz^l(-t\partial_z)^{n-2j-l},
\end{equation}
so the leading monomials of $\tilde P$ are given by those of:
\begin{equation}
    e^{\frac12\sum_{ij}A_{ij}z_iz_j}\!\!\!\!\!\!\!\!\!\!\sum_{\substack{\bm n=(n_1,\dots,n_m)\\ n_1+\dots+n_m=n}}\!\!\!\!\!\!\!\!p_{\bm n}\cosh^{n_1} r\sum_{j=0}^{\lfloor\frac {n_1}2\rfloor}\frac{(-1)^jn_1!t^j}{(n_1-2j)!j!2^j}\sum_{l=0}^{n_1-2j}\binom{n_1-2j}lz_1^lz_2^{n_2}\cdots z_m^{n_m}(-t\partial_{z_1})^{n_1-2j-l}e^{-\frac12\sum_{ij}A_{ij}z_iz_j}.
\end{equation}
Since the matrix $A$ is symmetric, each partial derivative $-t\partial_{z_1}$ of the exponential yields a term of the form $(t\sum_k{A_{1k}}z_k)e^{-\frac12\sum_{ij}A_{ij}z_iz_j}$, thus increasing the degree of the corresponding monomials by $1$ (other terms being of lower order), so the leading monomials of $\tilde P$ are given by those of:
\begin{equation}
    \sum_{\substack{\bm n=(n_1,\dots,n_m)\\ n_1+\dots+n_m=n}}\!\!p_{\bm n}\cosh^{n_1} r\sum_{j=0}^{\lfloor\frac {n_1}2\rfloor}\frac{(-1)^jn_1!t^j}{(n_1-2j)!j!2^j}\sum_{l=0}^{n_1-2j}\binom{n_1-2j}lz_1^l\left(t\sum_k{A_{1k}}z_k\right)^{n_1-2j-l}z_2^{n_2}\cdots z_m^{n_m}.
\end{equation}
Terms corresponding to $j>0$ in the above equation are of lower order, while the terms $j=0$ give:
\begin{equation}
    \begin{aligned}
        &\sum_{\substack{\bm n=(n_1,\dots,n_m)\\ n_1+\dots+n_m=n}}\!\!p_{\bm n}\cosh^{n_1} r\sum_{l=0}^{n_1}\binom{n_1}lz_1^l\left(t\sum_k{A_{1k}}z_k\right)^{n_1-l}z_2^{n_2}\cdots z_m^{n_m}\\
        &=\sum_{\substack{\bm n=(n_1,\dots,n_m)\\ n_1+\dots+n_m=n}}\!\!p_{\bm n}\cosh^{n_1} r\left(z_1+t\sum_k{A_{1k}}z_k\right)^{n_1}z_2^{n_2}\cdots z_m^{n_m}\\
        &=\sum_{\substack{\bm n=(n_1,\dots,n_m)\\ n_1+\dots+n_m=n}}\!\!p_{\bm n}\cosh^{n_1} r\left[(1+tA_{11})z_1+t\sum_{k>1}{A_{1k}}z_k\right]^{n_1}z_2^{n_2}\cdots z_m^{n_m}.
    \end{aligned}
\end{equation}
The polynomial $\tilde P$ thus has degree less or equal to $n$. Moreover, the polynomial $P$---being of degree $n$---has nonzero monomials of degree $n$. Let $\bm n'=(n_1',\dots,n_m')$ with $|\bm n'|=N$ be one such monomial with a maximal power for $z_1$ (note that it is not necessarily unique). We have $p_{\bm n'}\neq0$, and since $n_1'$ is maximal the coefficient $\bm n'$ of $\tilde P$ is given by
\begin{equation}
    p_{\bm n'}\cosh^{n_1'}r(1+tA_{11})^{n_1'},
\end{equation}
which is nonzero since $|t|<1$ and $|A_{11}|<1$. This shows that $\tilde P$ has degree equal to $n$. Hence, all Gaussian unitary operations leave the finite stellar rank of pure states invariant.

Now assume by contradiction that a Gaussian unitary operation $\hat G$ maps a state $\ket\psi$ of infinite stellar rank onto a state of finite stellar rank $\ket\phi$. Then, $\hat G^\dag$ would map the state $\ket\phi$ of finite stellar rank onto $\ket\psi$ of infinite stellar rank, which contradicts the previous proof that Gaussian unitary operations leave the finite stellar rank invariant. 
Gaussian unitary operations thus also map pure states of infinite stellar rank onto pure states of infinite stellar rank.

For a mixed state $\rho$, the stellar rank is given by Eq.~(\ref{eq:rankmixed}). For any Gaussian unitary operation $\hat G$,
\begin{equation}
    \begin{aligned}
        r^\star(\hat G\rho\hat G^\dag)&=\inf_{\{p_i,\psi_i\}}\sup_ir^\star(\hat G\ket{\psi_i})\\
        &=\inf_{\{p_i,\psi_i\}}\sup_ir^\star(\psi_i)\\
        &=r^\star(\rho),
    \end{aligned}
\end{equation}
where we used the fact that Gaussian unitary operations leave the stellar rank of pure states invariant.
Hence, Gaussian unitary operations leave the stellar rank of any quantum state invariant. 

Finally, if $\hat U$ is a Gaussian unitary, then $\hat U\otimes\hat{\mathbb I}$ is also a Gaussian unitary, and thus also preserves the stellar rank of any quantum state. 

\medskip

\ref{enum:3} $\Rightarrow$ \ref{enum:2}: Let $\hat U$ be a unitary operator such that $\hat U\otimes\hat{\mathbb I}$ preserves the stellar rank of any quantum state.
By Theorem~\ref{th:Husimizeros}, the set of pure states of stellar rank $0$ is the set of Gaussian pure states, so $\hat U\otimes\hat{\mathbb I}$ maps Gaussian pure states to Gaussian pure states.

Let $\rho$ by a (possibly mixed) Gaussian density matrix over $m$ modes. By the Williamson decomposition~\cite{ferraro2005gaussian}, $\rho$ can be brought to an $m$-mode thermal state by Gaussian unitary operations. Hence, there exists a $(2m)$-mode Gaussian pure state $\ket{G_\rho}$ such that
\begin{equation}
    \rho=\Tr_{k>m}(\ket{G_\rho}\!\bra{G_\rho}),
\end{equation}
where the partial trace is over modes $k$ for $k\in\{m+1,\dots,2m\}$.
As a result,
\begin{equation}
    \begin{aligned}
        \hat U\rho\hat U^\dag&=\hat U\Tr_{k>m}(\ket{G_\rho}\!\bra{G_\rho})\hat U^\dag\\
        &=\Tr_{k>m}\left[(\hat U\otimes\hat{\mathbb I})\ket{G_\rho}\!\bra{G_\rho}(\hat U\otimes\hat{\mathbb I})^\dag\right] .
    \end{aligned}
\end{equation}
The operator $\hat U\otimes\hat{\mathbb I}$ maps Gaussian pure states to Gaussian pure states, so $\hat U\rho\hat U^\dag$ is a reduced state of a Gaussian state and thus is also a Gaussian state~\cite{serafini2017quantum}, which completes the proof.

\end{proof}

\noindent An interesting consequence is the following result:

\begin{coro}
The stellar rank is non-increasing under Gaussian channels and measurements.
\end{coro}

\begin{proof}

The result is trivial if the state has infinite rank. We consider a finite rank state $\rho$ over $m$ modes in what follows.

Any Gaussian measurement can be implemented by a Gaussian unitary followed by heterodyne detection~\cite{giedke2002characterization,chabaud2020classical}, with (single-mode) positive operator-valued measure $\{\frac1\pi\ket\alpha\!\bra\alpha\}_{\alpha\in\mathbb C}$, so with Theorem~\ref{th:invar} it suffices to show that measuring a subset of the modes of $\rho$ with heterodyne detection does not increase its stellar rank. We pick a decomposition $\rho=\sum_ip_i\ket{\psi_i}\!\bra{\psi_i}$ with $r^\star(\rho)=\sup_ir^\star(\psi_i)$. The unnormalized state obtained from $\rho$ by measuring, say, its first $k$ modes with heterodyne detection and obtaining the outcomes $\bm\alpha=(\alpha_1,\dots,\alpha_k)\in\mathbb C^k$ is given by:
\begin{equation}\label{eq:decomphetmeas}
    \text{Tr}_{1\dots k}(\Pi_{\bm\alpha}\rho)=\sum_ip_i\ket{\phi^{\bm\alpha}_i}\!\bra{\phi^{\bm\alpha}_i},
\end{equation}
where $\Pi_{\bm\alpha}=\frac1{\pi^k}\ket{\alpha_1\dots\alpha_k}\!\bra{\alpha_1\dots\alpha_k}\otimes\mathbb I_{m-k}$, and where $\ket{\phi^{\bm\alpha}_i}$ are the unnormalized pure states obtained by projecting the first $k$ modes of $\ket{\psi_i}$ onto the $k$-mode coherent state $\ket{\alpha_1\dots\alpha_k}$. For all $i$, the Husimi $Q$ function of the state $\ket{\phi^{\bm\alpha}_i}$ is given by
\begin{equation}
    \begin{aligned}
        Q_{\phi^{\bm\alpha}_i}(\beta_1,\dots,\beta_{m-k})&=\frac1{\pi^{m-k}}\Tr(\ket{\phi^{\bm\alpha}_i}\!\bra{\phi^{\bm\alpha}_i}\ket{\beta_1\dots\beta_{m-k}}\!\bra{\beta_1\dots\beta_{m-k}})\\
        &=\frac1{\pi^{m-k}}\Tr\left(\frac1{\pi^k}\text{Tr}_{1\dots k}(\ket{\psi_i}\!\bra{\psi_i}\ket{\alpha_1\dots\alpha_k}\!\bra{\alpha_1\dots\alpha_k})\ket{\beta_1\dots\beta_{m-k}}\!\bra{\beta_1\dots\beta_{m-k}})\right)\\
        &=\frac1{\pi^m}\Tr\left(\ket{\psi_i}\!\bra{\psi_i}\ket{\alpha_1\dots\alpha_k\beta_1\dots\beta_{m-k}}\!\bra{\alpha_1\dots\alpha_k\beta_1\dots\beta_{m-k}}\right)\\
        &=Q_{\psi_i}(\alpha_1,\dots,\alpha_k,\beta_1,\dots,\beta_{m-k}).
    \end{aligned}
\end{equation}
Recalling the relation between the stellar function and the Husimi $Q$ function in Eq.~(\ref{eq:Husimistellarmulti}), we obtain that the stellar function of the state $\ket{\phi^{\bm\alpha}_i}$ is given by:
\begin{equation}
    F^\star_{\phi^{\bm\alpha}_i}(z_1,\dots,z_{m-k})\propto F^\star_{\psi_i}(\alpha_1^*,\dots,\alpha_k^*,z_1,\dots,z_{m-k}),
\end{equation}
up to a constant prefactor.
Since $r^\star(\psi_i)\le r^\star(\rho)$, the stellar function of the state $\ket{\psi_i}$ is of the form $P\times G$, where $P$ is a polynomial of degree less or equal to $r^\star(\rho)$, and $G$ is a Gaussian function. Hence, the stellar function of the state $\ket{\phi^{\bm\alpha}_i}$ is also of the form $P\times G$, where $P$ is a polynomial of degree less or equal to $r^\star(\rho)$, and $G$ is a Gaussian function. In particular, $r^\star(\phi^{\bm\alpha}_i)\le r^\star(\rho)$. With Eq.~(\ref{eq:decomphetmeas}) we obtain $r^\star(\Tr_{1\dots k}(\Pi_{\bm\alpha}\rho))\le r^\star(\rho)$ and therefore the stellar rank is non-increasing under Gaussian measurements.

With Eq.~(\ref{eq:weakconvex}) this directly implies that the stellar rank is non-increasing under partial trace, since we have
\begin{equation}
    \begin{aligned}
        \text{Tr}_{1\dots k}(\rho)&=\int_{\bm\alpha\in\mathbb C^k}\text{Tr}_{1\dots k}(\Pi_{\bm\alpha}\rho)d^{2k}\bm\alpha\\
        &=\int_{\bm\alpha\in\mathbb C^k}p_{\bm\alpha}\rho_{\bm\alpha}d^{2k}\bm\alpha,
    \end{aligned}
\end{equation}
where $\Pi_{\bm\alpha}=\frac1{\pi^k}\ket{\alpha_1\dots\alpha_k}\!\bra{\alpha_1\dots\alpha_k}\otimes\mathbb I_{m-k}$, $p_{\bm\alpha}:=\Tr(\Pi_{\bm\alpha}\rho)$ and $\rho_{\bm\alpha}:=\frac{\text{Tr}_{1\dots k}(\Pi_{\bm\alpha}\rho)}{\Tr(\Pi_{\bm\alpha}\rho)}$.

Finally, any Gaussian channel may be implemented by taking the tensor product with a Gaussian states, applying a Gaussian unitary operation, applying a Gaussian measurement and taking partial trace. By Eq.~(\ref{eq:tensorproductrank}), taking the tensor product with a Gaussian state leaves the stellar rank invariant. By Theorem~\ref{th:invar}, applying a Gaussian unitary operation also leave the stellar rank invariant. By the above, applying a Gaussian measurement and taking partial trace do not increase the stellar rank. Hence, the stellar rank is non-increasing under Gaussian channels.

\end{proof}

\noindent We now turn to the proof of Theorem~\ref{th:robust}:

\begin{theorem}
The $m$-mode stellar hierarchy is robust for the trace norm. Formally, for all $n\in\mathbb N$,
\begin{equation}
    \overline{\bm R_n}=\bigcup_{0\le k\le n}\bm R_k,
\end{equation}
where $\overline X$ denotes the closure of $X$ for the trace norm.
\end{theorem}

\begin{proof}

The proof proceeds by showing double inclusion.

\medskip

We first show that $\bigcup_{0\le k\le n}\bm R_k\subset\overline{\bm R_n}$. We have $\bm R_n\subset\overline{\bm R_n}$. Let $\ket{\bm\psi}\in\bigcup_{0\le k<n}\bm R_k$. By Lemma~\ref{lem:decomp2}, there exists an $m$-mode Gaussian unitary operation $\hat G_{\bm\psi}$ and a core state $\ket{C_{\bm\psi}}$ of stellar rank smaller than $n$ such that $\ket{\bm\psi}=\hat G_{\bm\psi}\ket{C_{\bm\psi}}$. Let $\bm n=(n,0,\dots,0)\in\mathbb N^m$. For all $l\ge1$, let
\begin{equation}
    \begin{aligned}
        \ket{\bm\psi_l}&:=\sqrt{\frac1m}\ket{\bm\psi}+\sqrt{1-\frac1m}\hat G_{\bm\psi}\ket{\bm n}\\
        &=\hat G_{\bm\psi}\left(\sqrt{\frac1m}\ket{C_{\bm\psi}}+\sqrt{1-\frac1m}\ket{\bm n}\right).
    \end{aligned}
\end{equation}
By Lemma~\ref{lem:decomp2}, $\ket{\bm\psi_l}\in\bm R_n$ for all $l\ge1$, and by construction the sequence $(\ket{\bm\psi_l})_{l\ge1}$ converges to $\ket{\bm\psi}$ in trace norm. This shows that $\bigcup_{0\le k\le n}\bm R_k\subset\overline{\bm R_n}$. 

\medskip

For the reverse inclusion, let $(\ket{\bm\psi_l})_{l\ge1}\in\bm R_n^{\mathbb N}$ be a converging sequence and let $\ket{\bm\psi}\in\mathcal H^{\otimes m}$ denote its limit. By Lemma~\ref{lem:decomp2}, there exist a sequence of passive linear unitaries $(\hat U_l)_{l\ge1}$, a sequence of squeezing amplitudes $(\bm\xi_l)_{l\ge1}$, a sequence of displacement amplitudes $(\bm\beta_l)_{l\ge1}$ and a sequence of core states $(\ket{C_{\bm\psi_l}})_{l\ge1}$ of stellar rank $n$ such that for all $l\ge1$, $\ket{\bm\psi_l}=\hat U_l\hat S(\bm\xi_l)\hat D(\bm\beta_l)\ket{C_{\bm\psi_l}}$. Moreover, with Eq.~(\ref{eq:braidDSR}) we may assume without loss of generality that the squeezing amplitudes $(\bm\xi_l)_{l\ge1}$ are real-valued.

The set of normalized core states of stellar rank less of equal to $n$ is compact (it identifies with the set of normalized multivariate polynomials of degree less or equal to $n$), so up to extracting a converging subsequence we may assume that the sequence $(\ket{C_{\bm\psi_l}})_{l\ge1}$ converges in trace norm to some core state $\ket C:=\sum_{|\bm n|\le n}c_{\bm n}\ket{\bm n}$ of stellar rank less or equal to $n$.
Similarly, the set of passive linear unitaries is compact (it identifies with the set of $m\times m$ unitary matrices), so up to extracting a converging subsequence we may assume that the sequence $(\hat U_l)_{l\ge1}$ converges in strong operator topology to some passive linear unitary operator $\hat U$. Since the trace distance is invariant under unitary operations, we obtain that $\hat S(\bm\xi_l)\hat D(\bm\beta_l)\ket C$ converges to $\hat U^\dag\ket{\bm\psi}$ in trace norm. 

In order to conclude, we combine several intermediate results (Lemmas~\ref{lemmapp:bounded},~\ref{lemmapp:continuous1} and~\ref{lemmapp:continuous2}), the proofs of which are deferred to the end for clarity:

\begin{lemma}\label{lemmapp:bounded}
The sequences $(\bm\xi_l)_{l\ge1}\in(\mathbb R^m)^{\mathbb N}$ and $(\bm\beta_l)_{l\ge1}\in(\mathbb C^m)^{\mathbb N}$ are bounded.
\end{lemma}

\noindent As a consequence, there exists a nondecreasing function $\phi:\mathbb N\rightarrow\mathbb N$ such that the subsequences and $(\bm\xi_{\phi(l)})_{l\ge1}$ and $(\bm\beta_{\phi(l)})_{l\ge1}$ converge. We denote by $\bm\xi\in\mathbb R^m$ and $\bm\beta\in\mathbb C^m$ their respective limits.
We also have:

\begin{lemma}\label{lemmapp:continuous1}
For any $m$-mode core state $\ket C$, the function $\bm\beta\mapsto\hat D(\bm\beta)\ket C$ is continuous over $\mathbb C^m$ with respect to the trace norm.
\end{lemma}

\noindent This result implies that $\hat D(\bm\beta_l)\ket C$ converges in trace distance along the subsequence $\phi$ to $\hat D(\bm\beta)\ket C$, i.e., for $\epsilon>0$ there exists $l_1\in\mathbb N^*$ such that $l\ge l_1$ implies
\begin{equation}
    D\left[\hat D(\bm\beta_{\phi(l)})\ket C,\hat D(\bm\beta)\ket C\right]<\epsilon,
\end{equation}
where $D$ is the trace distance. Since the trace distance is invariant under unitary operations, this implies
\begin{equation}\label{eq:intermediateTD}
    D\left[\hat S(\bm\xi_{\phi(l)})\hat D(\bm\beta_{\phi(l)})\ket C,\hat S(\bm\xi_{\phi(l)})\hat D(\bm\beta)\ket C\right]<\epsilon,
\end{equation}
for all $l\ge l_1$.
Finally, we have:

\begin{lemma}\label{lemmapp:continuous2}
For any $m$-mode core state $\ket C$ and any $\bm\beta\in\mathbb C^m$, the function $\bm\xi\mapsto\hat S(\bm\xi)\hat D(\bm\beta)\ket C$ is continuous over $\mathbb R^m$ with respect to the trace norm.
\end{lemma}

\noindent This result implies that $\hat S(\bm\xi_l)\hat D(\bm\beta)\ket C$ converges in trace distance along the subsequence $\phi$ to $\hat S(\bm\xi)\hat D(\bm\beta)\ket C$, i.e., for $\epsilon>0$ there exists $l_2\in\mathbb N^*$ such that $l\ge l_2$ implies
\begin{equation}
    D\left[\hat S(\bm\xi_{\phi(l)})\hat D(\bm\beta)\ket C,\hat S(\bm\xi)\hat D(\bm\beta)\ket C\right]<\epsilon.
\end{equation}
With Eq.~(\ref{eq:intermediateTD}) and the triangle inequality this implies
\begin{equation}
    D\left[\hat S(\bm\xi_{\phi(l)})\hat D(\bm\beta_{\phi(l)})\ket C,\hat S(\bm\xi)\hat D(\bm\beta)\ket C\right]<2\epsilon,
\end{equation}
for all $l\ge\max(l_1,l_2)$, i.e., $\hat S(\bm\xi_l)\hat D(\bm\beta_l)\ket C$ converges along the subsequence $\phi$ to $\hat S(\bm\xi)\hat D(\bm\beta)\ket C$. Since $\hat S(\bm\xi_l)\hat D(\bm\beta_l)\ket C$ also converges to $\hat U^\dag\ket{\bm\psi}$, we obtain $\ket{\bm\psi}=\hat U\hat S(\bm\xi)\hat D(\bm\beta)\ket C$ and thus $\ket{\bm\psi}\in\bigcup_{0\le k\le n}\bm R_k$ by Lemma~\ref{lem:decomp2}, which concludes the proof of Theorem~\ref{th:robust}.

\end{proof}

\begin{proof}[Proof of Lemma~\ref{lemmapp:bounded}.]
We make use of the following expression from Eq.~(\ref{eq:sqFock}) of the stellar function of a single-mode state of the form $\hat S(\xi)\ket n$, when $\xi\in\mathbb R$ and $\ket n$ is a Fock state:
\begin{equation}
    F^\star_{\hat S(\xi)\ket n}(z)=\frac{(\tanh r)^{n/2}}{\sqrt{n!\cosh r}}\mathit{He}_n\left(\frac z{\cosh r\sqrt{\tanh r}}\right)e^{\frac12(\tanh r)z^2},
\end{equation}
for all $z\in\mathbb C$, where $\mathit{He}_n(x)=\sum_{k=0}^{\lfloor\frac n2\rfloor}\frac{(-1)^kn!}{2^kk!(n-2k)!}x^{n-2k}$ is the $n^{th}$ Hermite polynomial. In particular, the $Q$ function of the state $\hat S(\bm\xi_l)\hat D(\bm\beta_l)\ket C$ is given by
\begin{equation}
    \begin{aligned}
        Q_{\hat S(\bm\xi_l)\hat D(\bm\beta_l)\ket C}(\bm\alpha)&=Q_{\hat D(\bm\gamma_l)\hat S(\bm\xi_l)\ket C}(\bm\alpha)\\
        &=Q_{\hat S(\bm\xi_l)\ket C}(\bm\alpha-\bm\gamma_l)\\
        &=\frac1{\pi^m}e^{-|\bm\alpha-\bm\gamma_l|^2}\left|F^\star_{\hat S(\bm\xi_l)\ket C}(\bm\alpha^*-\bm\gamma_l^*)\right|^2\\
        &=\frac1{\pi^m}e^{-|\bm\alpha-\bm\gamma_l|^2}\left|\sum_{\substack{\bm n=(n_1,\dots,n_m)\\|\bm n|\le n}}\!\!\!\!\!c_{\bm n}\prod_{j=1}^mF^\star_{\hat S(\xi_{j,l})\ket{n_j}}(\alpha_j^*-\gamma_{j,l}^*)\right|^2\\
        &=\frac1{\pi^m}e^{-|\bm\alpha-\bm\gamma_l|^2}\Bigg|\sum_{\substack{\bm n=(n_1,\dots,n_m)\\|\bm n|\le n}}\!\!\!\!\!c_{\bm n}\prod_{j=1}^m\frac{(\tanh \xi_{j,l})^{n_j/2}}{\sqrt{n_j!\cosh \xi_{j,l}}}\\
        &\quad\quad\quad\times\mathit{He}_{n_j}\!\!\left(\frac{\alpha_j^*-\gamma_{j,l}^*}{\cosh \xi_{j,l}\sqrt{\tanh \xi_{j,l}}}\right)e^{\frac12(\tanh \xi_{j,l})(\alpha_j^*-\gamma_{j,l}^*)^2}\Bigg|^2,
    \end{aligned}
\end{equation}
for all $\bm\alpha\in\mathbb C^m$, where we have set $\bm\xi_l=(\xi_{1,l},\dots,\xi_{m,l})\in\mathbb R^m$, $\bm\beta_l=(\beta_{1,l},\dots,\beta_{m,l})\in\mathbb C^m$, and $\bm\gamma_l=(\gamma_{1,l},\dots,\gamma_{1,l})$ with $\gamma_{k,l}=\beta_{k,l}\cosh \xi_{k,l}+\beta_{k,l}^*\sinh \xi_{k,l}$. The core state $\ket C$ is normalized, so using Cauchy--Schwarz inequality, we obtain the bound:
\begin{equation}
    \begin{aligned}
        Q_{\hat S(\bm\xi_l)\hat D(\bm\beta_l)\ket C}(\bm\alpha)&\le\frac1{\pi^m}e^{-|\bm\alpha-\bm\gamma_l|^2}\!\!\!\!\sum_{\substack{\bm n=(n_1,\dots,n_m)\\|\bm n|\le n}}\prod_{j=1}^m\Bigg|\frac{(\tanh \xi_{j,l})^{n_j/2}}{\sqrt{n_j!\cosh \xi_{j,l}}}\\
        &\quad\quad\quad\times\mathit{He}_{n_j}\!\!\left(\frac{\alpha_j^*-\gamma_{j,l}^*}{\cosh \xi_{j,l}\sqrt{\tanh \xi_{j,l}}}\right)e^{\frac12(\tanh \xi_{j,l})(\alpha_j^*-\gamma_{j,l}^*)^2}\Bigg|^2\\
        &=\frac{e^{\sum_{j=1}^m-|\alpha_j-\gamma_{j,l}|^2+\frac12\tanh \xi_{j,l}\left((\alpha_j-\gamma_{j,l})^2+(\alpha^*_j-\gamma^*_{j,l})^2\right)}}{\pi^m\prod_{j=1}^m\cosh \xi_{j,l}}\\
        &\quad\quad\quad\times\sum_{\substack{\bm n=(n_1,\dots,n_m)\\|\bm n|\le n}}\prod_{j=1}^m\left|\sum_{k=0}^{\lfloor\frac {n_j}2\rfloor}\frac{(-1)^k\sqrt{n_j!}}{2^kk!(n_j-2k)!}\left(\frac{\alpha_j^*-\gamma_{j,l}^*}{\cosh \xi_{j,l}}\right)^{n_j-2k}\right|^2\\
        &\le\frac{e^{\sum_{j=1}^m-|\alpha_j-\gamma_{j,l}|^2+\frac12\tanh \xi_{j,l}\left((\alpha_j-\gamma_{j,l})^2+(\alpha^*_j-\gamma^*_{j,l})^2\right)}}{\pi^m\prod_{j=1}^m\cosh \xi_{j,l}}\\
        &\quad\quad\quad\times\sum_{\substack{\bm n=(n_1,\dots,n_m)\\|\bm n|\le n}}\prod_{j=1}^m\left(\sum_{k=0}^{\lfloor\frac {n_j}2\rfloor}\frac{\sqrt{n_j!}}{2^kk!(n_j-2k)!}\left(\frac{|\alpha_j-\gamma_{j,l}|}{\cosh \xi_{j,l}}\right)^{n_j-2k}\right)^2.
    \end{aligned}
\end{equation}
Setting, for all $l\ge1$ and all $j\in\{1,\dots,m\}$,
\begin{equation}\label{eq:zjl}
    \begin{aligned}
        z_{j,l}(\bm\alpha)&:=\frac{\alpha_j-\gamma_{j,l}}{\cosh \xi_{j,l}}\\
        &=\frac{\alpha_j-\beta_{j,l}\cosh \xi_{j,l}-\beta_{j,l}^*\sinh \xi_{j,l}}{e^{i\theta_{j,l}}\cosh \xi_{j,l}},
    \end{aligned}
\end{equation}
we obtain
\begin{equation}\label{eq:finalproofboundQ}
    \begin{aligned}
        Q_{\hat S(\bm\xi_l)\hat D(\bm\beta_l)\ket C}(\bm\alpha)&\le\frac{e^{\sum_{j=1}^m-\cosh^2 \xi_{j,l}|z_{j,l}(\bm\alpha)|^2+\frac12\cosh \xi_{j,l}\sinh \xi_{j,l}\left(z_{j,l}(\bm\alpha)^2+z_{j,l}(\bm\alpha)^{*2}\right)}}{\pi^m\prod_{j=1}^m\cosh \xi_{j,l}}\\
        &\quad\quad\quad\times\sum_{\substack{\bm n=(n_1,\dots,n_m)\\|\bm n|\le n}}\prod_{j=1}^m\left(\sum_{k=0}^{\lfloor\frac {n_j}2\rfloor}\frac{\sqrt{n_j!}}{2^kk!(n_j-2k)!}|z_{j,l}(\bm\alpha)|^{n_j-2k}\right)^2\\
        &\le\frac{e^{-\frac12\sum_{j=1}^m|z_{j,l}(\bm\alpha)|^2}}{\pi^m\prod_{j=1}^m\cosh \xi_{j,l}}\sum_{\substack{\bm n=(n_1,\dots,n_m)\\|\bm n|\le n}}\prod_{j=1}^m\left(\sum_{k=0}^{\lfloor\frac {n_j}2\rfloor}\frac{\sqrt{n_j!}}{2^kk!(n_j-2k)!}|z_{j,l}(\bm\alpha)|^{n_j-2k}\right)^2\\
        &:=\frac{e^{-\frac12\sum_{j=1}^m|z_{j,l}(\bm\alpha)|^2}}{\pi^m\prod_{j=1}^m\cosh \xi_{j,l}}T(|z_{1,l}(\bm\alpha)|,\dots,|z_{m,l}(\bm\alpha)|),
    \end{aligned}
\end{equation}
where we used $(\cosh r-\sinh r)\cosh r>\frac12$ and $(\cosh r+\sinh r)\cosh r>\frac12$ in the second line, and where we have defined the multivariate polynomial $T$. Note that the function $(z_1,\dots,z_m)\mapsto e^{-\frac12\sum_{j=1}^m|z_j|^2}T(|z_1|,\dots,|z_m|)$ is bounded.

The sequence $(\hat S(\bm\xi_l)\hat D(\bm\beta_l)\ket C)_{l\ge1}$ converges to a normalized state, so by property of the trace norm, the sequence of $Q$ functions (which are overlaps with coherent states) converges point-wise to the $Q$ function of this normalized state. In particular, if the sequence $(\bm \xi_l)_{l\ge1}$ is unbounded, then we have $Q_{\hat S(\bm\xi_l)\hat D(\bm\beta_l)\ket C}(\bm\alpha)\rightarrow0$ for all $\bm\alpha\in\mathbb C^m$ by Eq.~(\ref{eq:finalproofboundQ}), so the sequence $(\bm\xi_l)_{l\ge1}$ must be bounded for the limit state to be normalizable. 

Similarly, if the sequence $(\bm\gamma_l)_{l\ge1}$ is unbounded while $(\bm\xi_l)_{l\ge1}$ is bounded, then at least one of the sequences $(z_{j,l})_{l\ge1}$ from Eq.~(\ref{eq:zjl}) is unbounded and we have $Q_{\hat S(\bm\xi_l)\hat D(\bm\beta_l)\ket C}(\bm\alpha)\rightarrow0$ for all $\bm\alpha\in\mathbb C^m$. This shows that all sequences $(z_{j,l})_{l\ge1}$, $(\bm\gamma_l)_{l\ge1}$, and thus $(\bm\beta_l)_{l\ge1}$ are bounded. Hence, both sequences $(\bm\xi_l)_{l\ge1}$ and $(\bm\beta_l)_{l\ge1}$ are bounded.

\end{proof}

\begin{proof}[Proof of Lemma~\ref{lemmapp:continuous1}.]
We give a proof for $m=1$ mode, as the proof in the general case is analogous.

Let $n\in\mathbb N$ and let $\ket C:=\sum_{k\le n}c_k\ket k$ be a normalized core state. We denote by $W_C$ the Wigner function of this state (see section~\ref{sec:Wigner}).
We define the function $f_C$ over $\mathbb C$ as
\begin{equation}
    \beta\longmapsto f_C(\beta):=\hat D(\beta)\ket C.
\end{equation}
For all $\beta,\beta'\in\mathbb C$,
\begin{equation}\label{eq:TDdispcore}
    \begin{aligned}
        D\left[f_C(\beta),f_C(\beta')\right]&=\sqrt{1-\left|\braket{C|\hat D(\beta)^\dag\hat D(\beta')|C}\right|^2}\\
        &=\sqrt{1-\left|\braket{C|\hat D(\beta'-\beta)|C}\right|^2},
    \end{aligned}
\end{equation}
where we used the fact that $f_C(\beta)$ and $f_C(\beta')$ are pure states in the first line, and Eq.~(\ref{eq:prodDSR}) in the second line.
By Eq.~(\ref{eq:Wigneroverlap}) we have
\begin{equation}
    \begin{aligned}
        \left|\braket{C|\hat D(\beta'-\beta)|C}\right|^2&=\int_{\mathbb R\times\mathbb R}W_C(x,p)W_{f_C(\beta'-\beta)}(x,p)\frac{dxdp}{2\pi}\\
        &=\int_{\mathbb R\times\mathbb R}W_C(x,p)W_C(x+\sqrt2\Re(\beta'-\beta),p+\sqrt2\Im(\beta'-\beta))\frac{dxdp}{2\pi},    
    \end{aligned}
\end{equation}
where we used the action of displacement on the Wigner function~\cite{ferraro2005gaussian}. Hence,
\begin{equation}\label{eq:prooflemcont1inter}
    \begin{aligned}
        &1-\left|\braket{C|\hat D(\beta'-\beta)|C}\right|^2\\
        &\quad\quad=1-\int_{\mathbb R\times\mathbb R}W_C(x,p)W_C(x+\sqrt2\Re(\beta'-\beta),p+\sqrt2\Im(\beta'-\beta))\frac{dxdp}{2\pi}\\
        &\quad\quad=\int_{\mathbb R\times\mathbb R}W_C(x,p)\left[W_C(x,p)-W_C(x+\sqrt2\Re(\beta'-\beta),p+\sqrt2\Im(\beta'-\beta))\right]\frac{dxdp}{2\pi}\\
        &\quad\quad\le\int_{\mathbb R\times\mathbb R}|W_C(x,p)|\left|W_C(x,p)-W_C(x+\sqrt2\Re(\beta'-\beta),p+\sqrt2\Im(\beta'-\beta))\right|\frac{dxdp}{2\pi},
    \end{aligned}
\end{equation}
where we used the fact that the core state $\ket C$ is a normalized pure state and Eq.~(\ref{eq:Wignerpurity}). The convergence of this last integral in ensured by the fact that the Wigner function of a core state is a product of a converging Gaussian function with a polynomial (a finite sum of generalized Laguerre polynomials~\cite{wunsche1998laguerre}). This also implies that $\int_{\mathbb R\times\mathbb R}|W_C(x,p)|\frac{dxdp}{2\pi}$ is bounded.

The Wigner function of a quantum state is uniformly continuous~\cite{cahill1969density}, so for any $\epsilon'>0$ there exists $\delta>0$ such that for all $x,x',p,p'\in\mathbb R$, $|x-x'|<\delta$ and $|p-p'|<\delta$ implies $|W_C(x,p)-W_C(x',p')|<\epsilon'$.
With Eq.~(\ref{eq:prooflemcont1inter}), choosing $|\beta-\beta'|<\frac\delta{\sqrt2}$ thus gives
\begin{equation}
    1-\left|\braket{C|\hat D(\beta'-\beta)|C}\right|^2<\epsilon'\int_{\mathbb R\times\mathbb R}|W_C(x,p)|\frac{dxdp}{2\pi}.
\end{equation}
Setting $\epsilon'=\epsilon^2/(\int_{\mathbb R\times\mathbb R}|W_C(x,p)|\frac{dxdp}{2\pi})$, we obtain $D\left[f_C(\beta),f_C(\beta')\right]<\epsilon$ with Eq.~(\ref{eq:TDdispcore}), which completes the proof.

\end{proof}

\begin{proof}[Proof of Lemma~\ref{lemmapp:continuous2}.]
The proof is similar to that of Lemma~\ref{lemmapp:continuous1}, up to slight technical differences. Once again, we give a proof for $m=1$ mode, as the proof in the general case is analogous.

Let $n\in\mathbb N$, let $\ket C:=\sum_{k\le n}c_k\ket k$ be a normalized core state, and let $\beta\in\mathbb C$.
We define the function $g_{C,\beta}$ over $\mathbb R$ as
\begin{equation}
    \xi\longmapsto g_{C,\beta}(\xi):=\hat S(\xi)\hat D(\beta)\ket C.
\end{equation}
For all $\xi,\xi'\in\mathbb R$,
\begin{equation}\label{eq:TDsqcore}
    \begin{aligned}
        D\left[g_{C,\beta}(\xi),g_{C,\beta}(\xi')\right]&=\sqrt{1-\left|\braket{C|\hat D(\beta)^\dag\hat S(\xi)^\dag\hat S(\xi')\hat D(\beta)|C}\right|^2}\\
        &=\sqrt{1-\left|\braket{C|\hat D(\beta)^\dag\hat S(\xi'-\xi)\hat D(\beta)|C}\right|^2},
    \end{aligned}
\end{equation}
where we used the fact that $g_{C,\beta}(\xi)$ and $g_{C,\beta}(\xi')$ are pure states in the first line, and Eq.~(\ref{eq:prodDSR}) in the second line.
By Eq.~(\ref{eq:Wigneroverlap}) we have
\begin{equation}
    \begin{aligned}
        \left|\braket{C|\hat D(\beta)^\dag\hat S(\xi'-\xi)\hat D(\beta)|C}\right|^2&=\int_{\mathbb R\times\mathbb R}W_{\hat D(\beta)\ket C}(x,p)W_{g_{C,\beta}(\xi'-\xi)}(x,p)\frac{dxdp}{2\pi}\\
        &=\int_{\mathbb R\times\mathbb R}W_{\hat D(\beta)\ket C}(x,p)W_{\hat D(\beta)\ket C}(e^{\xi'-\xi}x,e^{\xi-\xi'}p)\frac{dxdp}{2\pi},
    \end{aligned}
\end{equation}
where we used the action of squeezing on the Wigner function~\cite{ferraro2005gaussian}. Hence,
\begin{equation}\label{eq:prooflemcont2inter}
    \begin{aligned}
        &1-\left|\braket{C|\hat D(\beta)^\dag\hat S(\xi'-\xi)\hat D(\beta)|C}\right|^2\\
        &\quad\quad=1-\int_{\mathbb R\times\mathbb R}W_{\hat D(\beta)\ket C}(x,p)W_{\hat D(\beta)\ket C}(e^{\xi'-\xi}x,e^{\xi-\xi'}p)\frac{dxdp}{2\pi}\\
        &\quad\quad=\int_{\mathbb R\times\mathbb R}W_{\hat D(\beta)\ket C}(x,p)\left[W_{\hat D(\beta)\ket C}(x,p)-W_{\hat D(\beta)\ket C}(e^{\xi'-\xi}x,e^{\xi-\xi'}p)\right]\frac{dxdp}{2\pi}\\
        &\quad\quad\le\int_{\mathbb R\times\mathbb R}|W_{\hat D(\beta)\ket C}(x,p)|\left|W_{\hat D(\beta)\ket C}(x,p)-W_{\hat D(\beta)\ket C}(e^{\xi'-\xi}x,e^{\xi-\xi'}p)\right|\frac{dxdp}{2\pi},
    \end{aligned}
\end{equation}
where we used the fact that $\hat D(\beta)\ket C$ is a normalized pure state and Eq.~(\ref{eq:Wignerpurity}). The convergence of this last integral in ensured by the fact that the Wigner function of a displaced core state is the displaced Wigner function of a core state, i.e., a product of a converging Gaussian function with a polynomial (a finite sum of generalized Laguerre polynomials~\cite{wunsche1998laguerre}). 

This also implies that $\int_{\mathbb R\times\mathbb R}|W_{\hat D(\beta)\ket C}(x,p)|\frac{dxdp}{2\pi}=\int_{\mathbb R\times\mathbb R}|W_C(x,p)|\frac{dxdp}{2\pi}$ is bounded. Hence, for all $\epsilon'>0$, there exists $M>0$ such that
\begin{equation}
    \int_{x^2+p^2>M^2}|W_{\hat D(\beta)\ket C}(x,p)|\frac{dxdp}{2\pi}<\epsilon',
\end{equation}
so Eq.~(\ref{eq:prooflemcont2inter}) gives
\begin{equation}\label{eq:prooflemcont2inter2}
    \begin{aligned}
        &1-\left|\braket{C|\hat D(\beta)^\dag\hat S(\xi'-\xi)\hat D(\beta)|C}\right|^2\\
        &\quad\quad\le\int_{x^2+p^2>M^2}|W_{\hat D(\beta)\ket C}(x,p)|\left|W_{\hat D(\beta)\ket C}(x,p)-W_{\hat D(\beta)\ket C}(e^{\xi'-\xi}x,e^{\xi-\xi'}p)\right|\frac{dxdp}{2\pi}\\
        &\quad\quad+\int_{x^2+p^2\le M^2}|W_{\hat D(\beta)\ket C}(x,p)|\left|W_{\hat D(\beta)\ket C}(x,p)-W_{\hat D(\beta)\ket C}(e^{\xi'-\xi}x,e^{\xi-\xi'}p)\right|\frac{dxdp}{2\pi}\\
        &\quad\quad\le4\int_{x^2+p^2>M^2}|W_{\hat D(\beta)\ket C}(x,p)|\frac{dxdp}{2\pi}\\
        &\quad\quad+\int_{x^2+p^2\le M^2}|W_{\hat D(\beta)\ket C}(x,p)|\left|W_{\hat D(\beta)\ket C}(x,p)-W_{\hat D(\beta)\ket C}(e^{\xi'-\xi}x,e^{\xi-\xi'}p)\right|\frac{dxdp}{2\pi}\\
        &\quad\quad\le4\epsilon'+\int_{x^2+p^2\le M^2}|W_{\hat D(\beta)\ket C}(x,p)|\left|W_{\hat D(\beta)\ket C}(x,p)-W_{\hat D(\beta)\ket C}(e^{\xi'-\xi}x,e^{\xi-\xi'}p)\right|\frac{dxdp}{2\pi},
    \end{aligned}
\end{equation}
where we used the fact that the Wigner function is bounded by $2$ (see Eq.~(\ref{eq:Wignerbound})) in the fourth line.

The Wigner function of a quantum state is uniformly continuous~\cite{cahill1969density}, so for any $\epsilon'>0$ there exists $\delta>0$ such that for all $x,x',p,p'\in\mathbb R$, $|x-x'|<\delta$ and $|p-p'|<\delta$ implies $|W_{\hat D(\beta)\ket C}(x,p)-W_{\hat D(\beta)\ket C}(x',p')|<\epsilon'$.
Choosing $|\xi-\xi'|<\log(1+\frac\delta M)$ ensures that $|x-e^{\xi'-\xi}x|<\delta$ and $|p-e^{\xi-\xi'}p|<\delta$ for all $x,p$ with $x^2+p^2\le M^2$, so with Eq.~(\ref{eq:prooflemcont2inter2}) we obtain
\begin{equation}
    1-\left|\braket{C|\hat D(\beta'-\beta)|C}\right|^2<\epsilon'\left(4+\int_{\mathbb R\times\mathbb R}|W_{\hat D(\beta)\ket C}(x,p)|\frac{dxdp}{2\pi}\right).
\end{equation}
Setting $\epsilon'=\epsilon^2/(4+\int_{\mathbb R\times\mathbb R}|W_{\hat D(\beta)\ket C}(x,p)|\frac{dxdp}{2\pi})$, we obtain $D\left[g_{C,\beta}(\xi),g_{C,\beta}(\xi')\right]<\epsilon$ with Eq.~(\ref{eq:TDsqcore}), which completes the proof.

\end{proof}

%------------------------------------------------------------------------

\end{document}